\documentclass[11pt, letterpaper]{article}
\usepackage[margin=1in, centering]{geometry}
\usepackage{mathtools}
\usepackage{amsthm}
\usepackage{thmtools}
\usepackage{dsfont}
\usepackage{amssymb}
\usepackage{biolinum}
\usepackage{mathpazo, tgpagella}
\usepackage[table,dvipsnames]{xcolor}
\usepackage{subcaption}
\usepackage[normalem]{ulem}
\definecolor{HANADA}{RGB}{0, 98, 132}
\definecolor{KURENAI}{RGB}{203, 27, 69}
\usepackage[
    pdfstartview=FitH,
    pdfpagemode=UseNone,
    colorlinks=true,
    citecolor=KURENAI,
    linkcolor=HANADA,
    backref=page,
    linktocpage=true
]{hyperref}
\usepackage[capitalise, nameinlink]{cleveref}
\usepackage[nottoc]{tocbibind}
\usepackage{appendix}
\usepackage{array}
\usepackage{braket}
\usepackage{ytableau}
\usepackage{graphicx}
\usepackage{enumitem}
\usepackage{thm-restate}
\usepackage{tikz}
\usepackage{tikz-cd}
\usetikzlibrary{calc}
\usepackage[framemethod=TikZ]{mdframed}
\usepackage{setspace}
\usepackage{authblk}

\renewcommand\Affilfont{\small}
\makeatletter
\renewcommand\AB@affilsepx{\protect\\[0.2em]\protect\Affilfont}
\makeatother

\newtheorem{theorem}{Theorem}[section]
\newtheorem{proposition}[theorem]{Proposition}
\newtheorem{fact}[theorem]{Fact}
\newtheorem{lemma}[theorem]{Lemma}
\newtheorem{corollary}[theorem]{Corollary}

\newtheorem{definition}[theorem]{Definition}

\newtheoremstyle{obsstyle}
  {6pt}
  {6pt}
  {\normalfont}
  {}
  {\bfseries}
  {.}
  {.5em}
  {}

\theoremstyle{obsstyle}
\newtheorem{observation}[theorem]{Observation}
\newtheorem{remark}[theorem]{Remark}



\newcommand{\fig}[1]{\hyperref[fig:#1]{Figure~\ref*{fig:#1}}}
\newcommand{\eq}[1]{\hyperref[eq:#1]{(\ref*{eq:#1})}}
\newcommand{\lem}[1]{\hyperref[lem:#1]{Lemma~\ref*{lem:#1}}}
\newcommand{\thm}[1]{\hyperref[thm:#1]{Theorem~\ref*{thm:#1}}}
\newcommand{\defi}[1]{\hyperref[def:#1]{Definition~\ref*{def:#1}}}
\newcommand{\app}[1]{\hyperref[app:#1]{Appendix~\ref*{app:#1}}}
\newcommand{\fct}[1]{\hyperref[fact:#1]{Fact~\ref*{fact:#1}}}
\newcommand{\clr}[1]{\hyperref[clr:#1]{Corollary~\ref*{clr:#1}}}
\newcommand{\sct}[1]{\hyperref[sec:#1]{Section~\ref*{sec:#1}}}
\newcommand{\subsec}[1]{\hyperref[subsec:#1]{Subsection~\ref*{subsec:#1}}}
\newcommand{\itm}[2]{\hyperref[itm:#1]{#2}}
\newcommand{\clm}[1]{\hyperref[clm:#1]{Claim~\ref*{clm:#1}}}
\newcommand{\rmk}[1]{\hyperref[rmk:#1]{Remark~\ref*{rmk:#1}}}

\definecolor{lightcyan}{RGB}{0.88,1,1}
\definecolor{darkgreen}{RGB}{0, 128, 0}
\definecolor{darkblue}{RGB}{0, 0, 128}

\newcommand{\N}{\mathbb{N}}  %
\newcommand{\R}{\mathbb{R}} %
\newcommand{\C}{\mathbb{C}} %
\DeclareMathOperator*{\E}{\mathbb{E}}  %
\newcommand{\su}{\mathrm{SU}} %
\renewcommand{\i}{\mathrm{i}} %
\newcommand{\e}{\mathrm{e}} %
\newcommand{\A}{\mathcal{A}}  %
\renewcommand{\S}{\mathcal{S}} %
\newcommand{\PP}{\mathcal{P}} %
\newcommand{\expect}[2]{\E_{\substack{#1}}\!\Br{#2}} %
\newcommand{\prob}[2]{\underset{#1}{\mathrm{Pr}}\!\Br{#2}} %
\newcommand{\br}[1]{\left(#1\right)} %
\newcommand{\Br}[1]{\left[#1\right]} %
\newcommand{\st}[1]{\left\{#1\right\}} %
\newcommand{\tr}[1]{\mathrm{Tr}\!\Br{#1}} %
\newcommand{\abs}[1]{\left|#1 \right|} %
\newcommand{\norm}[1]{\left\lVert #1 \right\rVert} %
\newcommand{\poly}[1]{\mathrm{poly}\!\br{#1}} %
\newcommand{\negl}[1]{\mathrm{negl}\!\br{#1}} %
\newcommand{\ketbratwo}[2]{\ket{#1}\hspace{-0.4em}\bra{#2}} %
\newcommand{\ketbra}[1]{\ketbratwo{#1}{#1}} %
\newcommand{\id}{\ensuremath{\mathds{1}}} %
\newcommand{\ugroup}[1]{\mathrm{U}\!\br{#1}} %
\newcommand{\secpar}{\zeta} %
\newcommand{\cprim}[1]{\textnormal{\textup{\textsf{#1}}}} %
\newcommand{\prf}{\cprim{PRF}} %
\newcommand{\prp}{\cprim{PRP}} %
\newcommand{\pru}{\cprim{PRU}} %
\usepackage[normalem]{ulem}

\newcommand{\kac}{\ensuremath{\mathsf{Kac}}} %

\newcommand{\Tr}{\mathrm{Tr}} %

\renewcommand{\set}[1]{\ensuremath{\left\{#1\right\}}}

\newcommand{\btdist}[1]{\mathrm{Beta}(#1)}

\newcommand{\asym}{\ensuremath{\cprim{Asym}}}
\newcommand{\sym}{\ensuremath{\cprim{Sym}}}

\newcommand{\partitions}[1]{\ensuremath{\mathbb{Y}_{#1}}}

\newcommand{\Unif}{\mathrm{Unif}}

\newif\ifnotes
\notestrue          

\newcommand{\symsubspace}[2]{\vee^{#1}\mathbb{C}^{#2}}
\newcommand{\symsubspaceprojfull}[2]{\Pi_{\sym}^{#1,#2}}
\newcommand{\symsubspaceproj}{\Pi_{\sym}}
\newcommand{\typesymsubspace}[3]{\vee_{#1}^{#2}\mathbb{C}^{#3}}

\newcommand{\repappkacchannel}[2]{\widetilde{\Phi}_{\kac}^{#1,#2}}
\newcommand{\repkacchannel}[2]{\Phi_{\kac}^{#1,#2}}
\newcommand{\kacchannel}[1]{\Phi_{\kac}^{#1}}
\newcommand{\appkacchannel}[1]{\widetilde{\Phi}_{\kac}^{#1}}
\newcommand{\haarchannel}[1]{\Phi_{\mathsf{Haar}}^{#1}}
\newcommand{\blockhaarchannel}[1]{\Phi_{\mathsf{BlockHaar}}^{#1}}
\newcommand{\randpermchannel}[1]{\Phi_{\mathsf{Perm}}^{#1}}
\newcommand{\kacdist}{\ensuremath{\nu_{\kac}}} 
\newcommand{\haardist}{\nu_{\mathsf{Haar}}}

\newcommand{\lkmpchain}{M_{\scriptscriptstyle \mathsf{LKMP}}}
\newcommand{\lkmpdist}{\pi_{\scriptscriptstyle \mathsf{LKMP}}}
\newcommand{\kmpchain}{M_{\scriptscriptstyle \mathsf{KMP}}}
\newcommand{\kmpdist}{\pi_{\scriptscriptstyle \mathsf{KMP}}}
\newcommand{\liftedchain}{M_{\scriptscriptstyle \mathsf{PKMP}}}
\newcommand{\labeledchain}{M_{\scriptscriptstyle \mathsf{LPKMP}}}
\newcommand{\lifteddist}{\pi_{\scriptscriptstyle \mathsf{PKMP}}}
\newcommand{\labeleddist}{\pi_{\scriptscriptstyle \mathsf{LPKMP}}}
\newcommand{\contchain}{N}

\newcommand{\darrow}{\scalebox{0.6}{$\downarrow$}}
\newcommand{\stab}{\mathrm{Stab}}
\newcommand{\partition}[2]{\mathbb{Y}_{#1,#2}}
\newcommand{\match}{\mathcal{M}}
\newcommand{\kernel}{\mathcal{K}}

\newcommand{\Dirichlet}{\mathrm{Dirichlet}}

\newcommand{\glap}{\mathbf{L}}
\begin{document}

\newcommand{\titletxt}{A Tale of Two Walks: From Kac to KMP and Back}
\title{{\fontsize{16pt}{18pt}\selectfont\textbf{\titletxt}}}

\author[1]{Qian Chen}
\author[2]{Jingcheng Liu}
\author[3]{Minglong Qin}
\author[4]{Leonard J. Schulman}
\author[5]{Fang Song}
\author[2,6]{Penghui Yao}
\author[2]{Mingnan Zhao}

\makeatletter
\@namedef{@sep6}{\\[0.4em]}
\makeatother

\affil[1]{Mathematics Research Center, School of Science and Engineering, The Chinese University of Hong Kong, Shenzhen, China}
\affil[ ]{\texttt{chenqian.phys@gmail.com}}
\affil[2]{State Key Laboratory for Novel Software Technology, Nanjing University, Nanjing 210023, China}
\affil[ ]{\texttt{liu@nju.edu.cn, phyao1985@gmail.com, mingnanzh@gmail.com}}
\affil[3]{Centre for Quantum Technologies, National University of Singapore, Singapore}
\affil[ ]{\texttt{mlqin6@gmail.com}}
\affil[4]{California Institute of Technology, Pasadena CA 91125, USA}
\affil[ ]{\texttt{schulman@caltech.edu}}
\affil[5]{Computer Science Department, Portland State University, USA}
\affil[ ]{\texttt{fsong@pdx.edu}}
\affil[6]{Hefei National Laboratory, Hefei 230088, China}

\date{}
\maketitle
\thispagestyle{empty}
\setcounter{tocdepth}{2}

\allowdisplaybreaks

\vspace{-2em}

\begingroup
\setlength{\leftmargini}{1.6em}
\setstretch{1.1}
\begin{abstract}
We reveal an unexpected connection between two interacting particle systems:
Kac's walk and the Kipnis–Marchioro–Presutti (KMP) process.
We first establish a mixing-time analogue of Aldous's conjecture for
the KMP process on general graphs:
the mixing time of the \(t\)-particle process on any connected \(d\)-vertex graph
is at most \(O(\log d)\) times the corresponding single-particle mixing time,
uniformly in \(t\).
As an application, we improve the mixing-time bound for coordinate hit-and-run
on the \(n\)-dimensional standard simplex from \(\widetilde O(n^3)\) \cite{KV26} to \(O(n\log n)\),
while removing the dependence on the initial distribution.
We then show that the twirling channel induced by the parallel Kac's walk on the symmetric subspace can be exactly encoded by a Markov chain on partitions, which naturally lifts to a parallel KMP chain on the complete graph.
This correspondence allows us to translate our KMP mixing bound into a convergence bound for the parallel Kac twirling channel.
We prove that \(O(\log d+\log(1/\varepsilon))\)
repetitions of the parallel Kac's walk suffice to approximate Haar twirling on the symmetric subspace of
\((\mathbb C^d)^{\otimes t}\) to error \(\varepsilon\), uniformly in all \(t\).

Our main technical contribution is a new structural property of
KMP-type dynamics: \emph{conditional product structure},
which holds for both the standard KMP processes on general graphs
and the parallel KMP processes on complete graphs.
Conditioned on suitable auxiliary randomness, the evolution of the labeled KMP process factorizes into independent single-particle updates.
Combining this structure with an exact coupling,
we remove the dependence on the number of particles from the mixing bounds for both unlabeled KMP models.
This bypasses the $\log(\pi_{\min}^{-1})$ dependency from spectral-gap arguments, which is often prohibitively large for many-particle systems.
Using this framework,
we prove mixing bounds that are sharp up to logarithmic factors for both the standard KMP processes on general graphs
and the parallel KMP processes on complete graphs,
which in turn yield rapid convergence of the parallel Kac twirling channel on the symmetric subspace.

\end{abstract}

\newpage
\thispagestyle{empty}
\tableofcontents
\setcounter{page}{0}
\endgroup

\newpage
\begingroup
\setstretch{1.1}

\section{Introduction}
\label{sec:introduction}
The study of interacting particle systems, originating in statistical mechanics, has become a rich area of probability theory with broad connections to the natural and social sciences, engineering, and computer science. A recurring theme is to understand how simple local interactions among many components give rise to global equilibrium behavior. This work develops a fruitful interplay between two notable examples, Kac's walk and the Kipnis--Marchioro--Presutti (KMP) process, through which ideas from each advance our understanding of the other.

Kac's walk was introduced as a stochastic $N$-particle model for the spatially homogeneous Boltzmann equation. In its original formulation~\cite{Kac56}, binary collisions act as random rotations of pairs of particle velocities, preserving the total kinetic energy. Equivalently, this defines a random walk on a high-dimensional sphere, where at each step a random pair of coordinates is chosen, and a random rotation is applied in the chosen two-dimensional coordinate plane. The walk also admits a natural formulation on $\mathrm{SO}(d)$ and extensions to $\mathrm{U}(d)$ and $\mathrm{SU}(d)$~\cite{Hastings70,Diaconis2000,PS17}. Its convergence to equilibrium has driven and benefited from techniques including spectral-gap estimates~\cite{Jan2001,Carlen2003,Caputo2008,Sasada2013}, entropy decay~\cite{CCE2018,CCL+10}, and coupling arguments~\cite{Oliveira09,PS18}. These efforts have essentially led to sharp mixing time bounds. At constant total-variation accuracy, Kac's walk on the sphere mixes in $\Theta(d\log d)$ steps~\cite{PS17}. On $\mathrm{SO}(d)$, Pillai and Smith obtained an $O(d^2\log d)$ upper bound~\cite{PS26}, recently improved by Jiang to $O(d^2)$~\cite{Jiang2026}, matching the dimensional lower bound.

Kac's walk has also found applications in algorithms and quantum information. Its structured randomness has been exploited in fast dimensionality reduction and trapdoored matrices for fast matrix multiplication~\cite{JPSSS22,PSV26,VZ26}. Brand{\~a}o, Harrow and Horodecki~\cite{BHH16} adapted Oliveira's path-recording analysis for Kac's walk~\cite{Oliveira09} as a critical component to analyze their local random circuits, with subsequent work obtaining increasingly efficient unitary designs~\cite{PhysRevA.104.022417,Haferkamp2022randomquantum, OSRP:2023,HMH+23,CHHLMT25}. More recently, Lu, Qin, Song, Yao, and Zhao used Kac's walk itself as a building block for quantum pseudorandomness~\cite{LQSY+24,LQSY+25}. The standard walk, however, updates only one pair of coordinates at a time and is therefore prohibitively slow when the ambient dimension is exponential in the number of qubits. This motivated their \emph{parallel Kac's walk}, in which all coordinates are randomly paired and independent two-dimensional random unitaries are applied simultaneously. It was shown that parallel Kac's walk mixes exponentially faster on the sphere, and each step admits an efficient quantum-circuit implementation when supplied with random functions and permutations. This then gave rise to quantum pseudorandom primitives such as pseudorandom state scramblers and pseudorandom unitaries. These developments warrant a closer investigation of the convergence behavior of this parallel dynamics, and in particular of the extent to which it may provide an efficient approximation to Haar randomness.

Turning to the KMP process~\cite{KMP82}, it was introduced to study how \emph{macroscopic} laws such as Fourier's law of heat conduction arise from \emph{microscopic} dynamics.
In its \emph{continuous} formulation, each vertex of a graph holds a nonnegative amount of energy. At each step, an edge is chosen uniformly at random, and the total energy at its endpoints is redistributed uniformly between them. In the \emph{discrete} version, $t$ indistinguishable particles reside on the $d$ vertices of the graph, and an update redistributes the particles at the chosen endpoints uniformly at random over all possible splits. Like Kac's walk, KMP combines local random updates with a global conservation law.

The study of convergence in the KMP process reveals an interesting tension between simple one-particle behavior and strongly correlated many-particle dynamics. In discrete KMP, if the particles are labeled and any one particle is tracked in isolation, its marginal evolution is governed by the one-particle KMP chain. The one-particle chain therefore provides a natural benchmark for the timescale of convergence, while the key question is how much additional slowdown can arise from the interactions among many particles. A notable manifestation for such a many-versus-one comparison is Aldous's spectral-gap conjecture for the \emph{interchange process}, which asserts that, on any graph, the spectral gap of the full labeled-particle system is exactly that of the underlying single-particle random walk. This conjecture was settled in the celebrated work of Caputo, Liggett, and Richthammer~\cite{caputo2010proof}. For KMP, the exact analogue is false in general, but strong spectral relations nevertheless hold: the continuous KMP gap is comparable to the single-particle gap~\cite{kim2025spectral}, while Caputo, Quattropani, and Sau showed that the spectral gap of every discrete $t$-particle KMP chain with $t\ge 2$ is already determined by the corresponding two-particle chain~\cite{caputo2026}.

What about mixing time? Spectral-gap comparisons do not in general translate into comparable total-variation mixing bounds. The state space of the $t$-particle KMP chain has size $\binom{t+d-1}{d-1}$, so the standard conversion from spectral gap to mixing time incurs a logarithmic dependence on this quantity, and hence on the number of particles. Oliveira established a particularly relevant result in this direction for the \emph{symmetric exclusion} process~\cite{oliveira2013mixing}, where particles perform the underlying random walk subject to the exclusion constraint. Oliveira showed that the many-particle mixing time is controlled, up to a logarithmic factor in the number of vertices, by the underlying single-particle mixing time, uniformly over the number of particles. This raises the question of whether a mixing-time analogue of Aldous’s conjecture holds for KMP: \emph{can its many-particle mixing time likewise be controlled by the single-particle timescale with no dependence on the number of particles}?

More broadly, this question belongs to a central theme in the modern analysis of Markov chains: understanding how local structure and dependence control global convergence. Recent developments based on log-concavity, spectral independence, and high-dimensional expansion have yielded sharp mixing bounds in a wide range of settings~\cite{anari2018log,anari2019log,anari2024log,anari2021log,branden2020lorentzian,anari2021spectral,chen2023rapid,chen2024spectral,alev2020improved,cryan2019modified,anari2022entropic,chen2021optimal,feng2022rapid,chen2022localization,chen2022optimal,chen2024rapid}. Much of this progress concerns broad classes of Boolean and spin systems. Interacting particle systems such as KMP, with a conserved quantity and correlated redistribution dynamics, present a rather different setting.  Mixing questions for KMP also have direct algorithmic significance. Coordinate hit-and-run, an important Markov-chain method for geometric sampling, reduces on the standard simplex to the continuous KMP process on a star graph~\cite{narayanan2022mixing,laddha2023convergence,narayanan2023sampling,KV26}. Thus, understanding KMP mixing also bears on geometric sampling, a fundamental problem at the interface of probability, algorithms, and optimization.

In this work, we establish new convergence results for both KMP and Kac through a two-way connection. For the KMP process, we obtain a mixing-time analogue of Aldous's conjecture, while for parallel Kac's walk, we establish a new convergence bound for its induced quantum channel. For KMP, we combine a new conditional product structure with Smith's two-stage coupling framework~\cite{Smi14}, also used in the analysis of Kac's walk~\cite{PS17}, to obtain mixing bounds on general graphs that are uniform in the number of particles. Conversely, we show that the action of parallel Kac's walk on the symmetric subspace reduces exactly to a KMP-type dynamics, allowing our KMP mixing bounds to be translated into convergence of parallel Kac's walk. These results have immediate applications. For coordinate hit-and-run on the $n$-dimensional standard simplex, we improve Kook and Vempala's $\widetilde O(n^3)$ mixing bound~\cite{KV26} to $\widetilde O(n)$, and remove their assumption on the initial distribution. On the quantum side, the convergence result strengthens a pseudorandom unitary construction to achieve \emph{scalable} security in the non-adaptive setting for inputs supported on the symmetric subspace.

\subsection{Main Results}

We first study the mixing time of the KMP chain on general graphs,
and establish \emph{a mixing time analogue of Aldous's conjecture}
for it.

\begin{theorem}[Informal version of \cref{thm:kmp-mixing-time-general-graph,rmk:kmp-mixing-time-general-graph}]
On any connected $d$-vertex graph, the mixing time of the
$t$-particle KMP chain is at most $O(\log d)$ times
the corresponding single-particle mixing time, uniformly over all
particle numbers $t$.
\end{theorem}

As we mentioned earlier, this relation for mixing time is not possible by
a black-box conversion from a spectral gap relation,
since a factor proportional to the log of the support size occurs.
To get around this,
we first lift this chain by adding labels to the particles,
and obtain a labeled version of the KMP chain.
We identify a crucial structural property of the labeled chain,
which we call \emph{conditional product structure}, and use it
to prove a strong mixing-time bound.
This property provides a pathwise interpretation of KMP
duality~\cite{KMP82}.
An analogous representation for binomial splitting appears
in~\cite{QS23}, while a corresponding representation for KMP
is also obtained independently in the concurrent
work~\cite{caputo2026}.

\begin{center}
\begin{minipage}{0.8\textwidth} 
    \begin{mdframed}[
        hidealllines=true,
        backgroundcolor=gray!10,
        roundcorner=5pt,
        innerleftmargin=15pt,
        innerrightmargin=15pt,
        innertopmargin=10pt,
        innerbottommargin=10pt,
        frametitle={{Conditional Product Structure} (Informal. See \cref{sec:labeled-kmp-chain-conditional-structure})},
        frametitlealignment=\centering,
        frametitleaboveskip=12pt,
        frametitlebelowskip=0pt
    ]
    Conditioned on the auxiliary randomness, the global update decomposes into
    independent updates of the labeled particles.
    \end{mdframed}
\end{minipage}
\end{center}

We first use this conditional independence to reduce the multi-particle evolution
to a single-particle evolution.
The conditional single-particle evolution alone does not give step-wise contraction under the worst-case conditioning;
our idea is to study the evolution of the covariance matrix of the single-particle evolution.
The expected evolution of the covariance matrix follows a two-particle labeled KMP process.
By showing that the two-particle labeled KMP has the same order of spectral gap as the single-particle random walk, we are able to establish a mixing time bound for the $t$-particle processes by only paying an extra factor of $O(\log d + \log t)$.
This gives the desired bound when \(t\) is polynomially bounded in \(d\).
A similar idea of controlling mixing through two-particle
behavior appears in~\cite{PR26}, which bounds the mixing
time by $O(\log(d+t))$ times the maximal
expected meeting time of two independent random walks.

For larger \(t\), our starting point is Smith's two-stage
coupling framework for continuous KMP on the complete graph,
which he also noted its possible extension
to discrete KMP with sufficiently many particles~\cite{Smi14}.
This framework was subsequently adapted to Kac's walk~\cite{PS17}.
Here, we adopt this framework and additionally have to keep track of
the mixing time with reference to the single-particle process and perform
the analysis on general graphs rather than the complete graph.
Together with the bound above,
this removes the dependence on \(t\) from the mixing time.

We also establish the same mixing-time bound for the continuous
KMP chain on general graphs (see \cref{lem:continuous-kmp-mixing}),
by transferring our discrete KMP bound through an intertwining relation.
As an application, we obtain an improved mixing-time bound for
coordinate hit-and-run on the \(n\)-dimensional standard simplex.
The connection is simple: adding the coordinate
\(x_0=1-\sum_{i=1}^n x_i\) identifies this walk with the continuous
KMP chain on the star graph \(K_{1,n}\).
Our result therefore gives the following mixing-time bound,
which improves the bound obtained from Kook and Vempala~\cite{KV26}
from \(\widetilde O(n^3)\) to \(\widetilde O(n)\),
and removes their assumption on the initial distribution.

\begin{theorem}[Informal version of \cref{lem:char-simplex-mixing}]
Coordinate hit-and-run on the $n$-dimensional standard simplex
mixes to within $\varepsilon$ in total variation distance in
$O(n(\log n+\log(1/\varepsilon)))$ steps,
from any initial distribution.
\end{theorem}

We further identify a new avenue connecting Kac's walk to the KMP chain,
and show that the twirling channel of the parallel Kac's walk \cite{LQSY+24} converges rapidly to the Haar twirling channel
$\haarchannel{t}$ on the symmetric subspace.
Specifically,
let
\(
	K=P\cdot W
\)
be the $n$-qubit random unitary realizing one step of the parallel Kac's walk,
where \(P\) is a random permutation matrix and \(W\) is a block-diagonal unitary
with independent \(2\times 2\) Haar random unitaries in each block.
Let $\repkacchannel{t}{k}$ denote the quantum channel obtained by applying
\( (K_k \cdots K_1)^{\otimes t} \) on an input state, where \(K_1,\dots,K_k\) are $k$ independent steps of the parallel Kac's walk.

\begin{theorem}[Informal version of \cref{thm:twirling}]
There exists
$k=O(\log d+\log(1/\varepsilon))$ such that
$\repkacchannel{t}{k}(\sigma)$ and $\haarchannel{t}(\sigma)$
are $\varepsilon$-close in trace norm
for every positive integer $t$ and every state $\sigma$
supported on the symmetric subspace $\symsubspace{t}{d}$.
\end{theorem}

The proof connects the parallel Kac's walk to the KMP chain.
The key observation is that,
on the symmetric subspace, each step of the
parallel Kac's walk acts much like a KMP update, redistributing multiplicities
between paired coordinates.
Building on this, we are able to identify the action of the \emph{quantum} channel \(\repkacchannel{t}{k}\)
on the symmetric subspace with the evolution of a simple \emph{classical} Markov chain, which we call a \emph{partition resampling chain}. This Markov chain can be lifted first to a parallelized and then to a labeled version of the KMP chain. The above conditional product structure hence enables us to prove a strong mixing time bound and the convergence of the parallel Kac channel \(\repkacchannel{t}{k}\) follows.

This information-theoretical approximation of a Haar random unitary by parallel Kac's walk complements the existing result in the computational setting, which gives a pseudorandom unitary. As an immediate application (details in~\cref{sec:non-adap-pru}), the efficient quantum circuit implementation strengthens the pseudorandom unitary to have \emph{scalable} security (albeit restricted to the non-adaptive setting and the symmetric subspace), where security loss can be tuned by a dedicated parameter independent of the system size.

\subsection{Proof Overview}
We now give a proof overview of our main results and highlight our technical contribution.

\subsubsection{Mixing of the KMP Chain}

\begin{figure}[ht]
    \centering
    \resizebox{\linewidth}{!}{%
    \begin{tikzpicture}[x=1cm, y=1cm, >=stealth,
        box/.style={draw=black!75, line width=0.6pt, rounded corners=3pt,
            align=center, inner sep=6pt, outer sep=2pt, font=\scriptsize},
        method/.style={box, text width=3.5cm, minimum height=1.05cm},
        detail/.style={box, text width=4.5cm, minimum height=0.9cm,
            inner ysep=4pt},
        route/.style={->, line width=0.65pt},
        label text/.style={fill=white, align=center,
            font=\fontsize{7.3}{9}\selectfont, inner sep=2pt}]
        \node[box, text width=2.6cm, minimum height=0.85cm] (goal) at (1.6,0) {
            \textbf{KMP mixing time}\\[3pt]
        };

        \node[font=\scriptsize\bfseries, anchor=west, inner sep=0pt]
            at (3.8,0.35) {Small particle number};
        \node[method] (labeled) at (7.0,1.9) {\textbf{Labeled KMP analysis}};
        \node[detail] (product) at (12.2,2.65) {
            \textbf{Conditional product structure}\\[1pt]
            \fontsize{7.3}{9}\selectfont Reduce to single-particle errors
        };
        \node[detail] (moment) at (12.2,1.15) {
            \textbf{Second-moment analysis}\\[1pt]
            \fontsize{7.3}{9}\selectfont
            \resizebox{\linewidth}{!}{Control errors via two-particle dynamics}
        };
        \draw[route] (goal.north east)
            to[out=55,in=180] (labeled.west);
        \node[label text] at (4.05,1.55) {Lifting};
        \coordinate (fork) at (9.35,1.9);
        \draw[draw=black!75, line width=0.6pt]
            ([xshift=-2pt]labeled.east) -- (fork);
        \draw[draw=black!75, line width=0.6pt, rounded corners=3pt]
            ([xshift=2pt]product.west) -| (fork) |- ([xshift=2pt]moment.west);

        \draw[black!35, dashed, line width=0.4pt] (3.7,0) -- (15.0,0);

        \node[font=\scriptsize\bfseries, anchor=west, inner sep=0pt]
            at (3.8,-0.35) {Large particle number};
        \node[method] (coupling) at (7.0,-1.9) {
            \textbf{Two-stage exact coupling}
        };
        \node[detail] (contraction) at (12.2,-1.15) {
            \textbf{Stage I: Contraction stage}\\[1pt]
            \fontsize{7.3}{9}\selectfont Bring the two copies close
        };
        \node[detail] (coalescence) at (12.2,-2.65) {
            \textbf{Stage II: Coalescence stage}\\[1pt]
            \fontsize{7.3}{9}\selectfont Make the two copies agree
        };
        \draw[route] (goal.south east)
            to[out=-55,in=180] (coupling.west);
        \node[label text] at (3.85,-1.55) {Coupling};
        \coordinate (coupling fork) at (9.35,-1.9);
        \draw[draw=black!75, line width=0.6pt]
            ([xshift=-2pt]coupling.east) -- (coupling fork);
        \draw[draw=black!75, line width=0.6pt, rounded corners=3pt]
            ([xshift=2pt]contraction.west) -| (coupling fork)
            |- ([xshift=2pt]coalescence.west);
    \end{tikzpicture}
    }
    \caption{Proof roadmap for the KMP chain on general graphs.}
    \label{fig:proof-roadmap-kmp-mixing}
\end{figure}

We first briefly describe the proof of our mixing time bound for the KMP chain,
following the roadmap in~\cref{fig:proof-roadmap-kmp-mixing}.
Recall that the KMP chain on a connected graph \(G=([d],E)\) has state space
\(\Delta=\{r\in\N^d:\sum_{i=1}^d r_i=t\}\),
where \(r_i\) is the number of particles at vertex \(i\).
At each step, it chooses a uniformly random edge \(\{i,j\}\) and uniformly
resamples the split of the \(r_i+r_j\) particles between vertices \(i\) and \(j\).

\vspace{-0.8em}
\paragraph{Lifting to the labeled KMP chain.}
We begin by lifting the KMP chain to
the \emph{labeled KMP chain} (see \cref{def:labeled-kmp}).
To do so, we view a multiplicity vector \(r=(r_1,\dots,r_d)\in\Delta\) as describing a configuration of \(t\) distinct labeled particles, with \(r_i\) particles placed at vertex \(i\).
This leads to a new chain on the state space
\(
\Omega \coloneqq [d]^t,
\)
where for \(x=(x_1,\dots,x_t)\in\Omega\), the coordinate \(x_\ell\in[d]\) records the location of particle \(\ell\).
This lifting keeps track not only of the number of particles at each vertex but also of which particles are there.
At each step, after choosing a random edge \(\{i,j\}\),
we uniformly resample the number of particles at \(i\)
and choose uniformly at random which of the labeled particles
at \(i\) and \(j\) are placed at \(i\).
By forgetting the labels and retaining only the number of particles
at each vertex, we recover the KMP chain.
Thus, any mixing bound for the labeled KMP chain also applies to the unlabeled KMP chain, since projection cannot increase total variation distance (see \cref{lem:labeled-dominates-lifted-general-kmp}).

\vspace{-0.8em}
\paragraph{Conditional product structure.}
To bound the mixing time of the labeled KMP chain,
a natural first attempt is to bound the spectral gap of
the labeled KMP chain.
However, even a spectral-gap bound independent of \(t\)
leads, through the standard conversion, to a mixing-time bound
that grows polynomially in \(t\),
since the state space has size \(d^t\).
To improve this dependence, we observe and exploit
an essential structural feature of the labeled chain,
which we call the \emph{conditional product structure}
(see \cref{sec:labeled-kmp-chain-conditional-structure}).

The key is an equivalent description of the redistribution step
in the labeled chain.
Instead of viewing the update on a selected edge \(\{i,j\}\)
as a correlated resampling step,
we may equivalently first sample a uniformly random parameter
\(p_{i,j}\in[0,1]\), and then let each particle currently
at vertices $i$ and $j$ independently choose vertex \(i\) with probability \(p_{i,j}\) and vertex \(j\)
with probability \(1-p_{i,j}\) (see \cref{obs:equiv-rule}).
Therefore, conditioned on the entire sequence of selected edges and sampled parameters, the particles evolve independently, and the joint distribution of the chain admits a product structure (see \cref{obs:conditional-product-structure}).
Concretely, once we fix the entire sequence of selected edges and sampled parameters, each labeled particle independently evolves as a single-particle Markov chain with transition matrices \(A_1,\ldots,A_k\).
As a result, after $k$ steps,
the conditional distribution of the full $t$-particle configuration decomposes as a product of the corresponding single-particle marginals.
Another crucial observation is that the stationary distribution
admits a similar conditional product structure (see \cref{lem:labeled-kmp-stationary-product-structure}).
Starting from a random vector \(u\) drawn uniformly from the probability simplex, we evolve it using the same sequence of random matrices to obtain \(w^{(k)} = u B_k\), where \(B_k=A_1\cdots A_k\). Averaging \((w^{(k)})^{\otimes t}\) over the same randomness recovers exactly the stationary distribution. This lets us represent both distributions using the same randomness. As a result, their total variation distance can be reduced to a sum of single-particle total variation distances.
This is the key step that allows us to avoid applying a spectral-gap bound directly to the full \(t\)-particle chain
and obtain a sharper mixing bound.

To turn this into a mixing bound,
we need to control the differences
between the single-particle distributions.
We do so by studying the second moments of these differences,
whose evolution is governed by the two-particle labeled KMP chain
(see \cref{thm:lkmp-mixing-time-general-graph}).
The state space of the two-particle labeled chain splits into two invariant subspaces: the symmetric subspace and the antisymmetric subspace. On the symmetric subspace, the dynamics agree with the unlabeled two-particle KMP chain. On the antisymmetric subspace, the dynamics are entirely determined by the single-particle random walk on \(G\).  Putting these together, we obtain a mixing bound for the (labeled) KMP chain on general graphs that,
up to an $O(\log d + \log t)$ factor, matches the mixing time of the single-particle random walk (see \cref{cor:mixing-time-labeled-kmp-general-graph}).

\vspace{-0.8em}
\paragraph{Two-stage exact coupling.}
The bound above already yields the desired \(O(\log d)\) mixing time
when \(t\) is polynomially bounded in \(d\),
but its \(\log t\) dependence becomes prohibitive for larger \(t\).
In this regime, we adapt Smith's two-stage coupling framework
for continuous KMP on the complete graph~\cite{Smi14} to general graphs.
This framework was subsequently adapted to Kac's walk~\cite{PS17}.
Specifically, we couple two copies of the KMP chain, one started from an arbitrary
configuration and the other from stationarity,
so that they coincide with high probability
(see \cref{lem:kmp-large-particle-one-block}).
By the coupling lemma (see \cref{lem:coupling_lemma}), this bounds the
total variation distance to stationarity.

In the first contraction stage, we run the two copies as labeled chains
using the same selected edges and sampled parameters.
By the conditional product structure, the conditional means of the two
normalized configurations evolve under the same random matrices,
and their difference contracts at a rate governed by the two-particle
labeled chain, as in the analysis above.
Since \(t\) is large, each normalized configuration concentrates around
its conditional mean.
Hence, the discrepancy between the configurations
of the two copies also shrinks.
Specifically, after \(O(\log d)\) steps\footnote{More precisely, the number of steps in each stage is
\(O(\gamma^{-1}\log d)\), where \(\gamma\) is the spectral gap of the
single-particle random walk on \(G\). We suppress the factor
\(\gamma^{-1}\) in this overview.},
they differ by only a small fraction of \(t\) in \(\ell_1\) distance.

In the second coalescence stage, we make the two copies agree exactly
(see \cref{lem:kmp-scheduled-coalescence}).
We sample a sequence of \(O(\log d)\) edges in advance,
which connects \(G\) with high probability,
and construct a sequence of vertex partitions backward in time,
starting from the singleton partition and merging the two parts
joined by each sampled edge.
Since the sampled edges connect \(G\), the initial partition
consists of a single part, namely the whole vertex set \([d]\).
We then run the two copies forward while maintaining that their particle
totals agree on every part of the current partition.
This holds initially, since both copies have \(t\) particles in total.
Whenever an update splits a part into two,
we couple the two copies so that their totals also agree on both new parts.
In this way, the two copies agree on finer and finer parts,
and eventually at every vertex.

The key step is to keep the totals in agreement whenever a part splits,
which succeeds with high probability for two reasons.
First, by the first stage, the two copies differ only slightly at every vertex.
Second, since \(t\) is large, each update involves many particles,
which leaves enough room to absorb this small discrepancy.
Hence, slightly adjusting how the particles are redistributed
suffices to make the totals agree.
Altogether, the two stages take \(O(\log d)\) steps,
which yields a mixing bound that is uniform in \(t\).

\subsubsection{From parallel Kac’s walk to parallel KMP chain}
To analyze the convergence of the parallel Kac twirling channel,
we introduce the \emph{parallel KMP chain}
(see \cref{def:discrete-parallel-kmp}), which at each step samples a
uniformly random perfect matching and resamples the particles on each
matched pair of vertices.
We first reduce the problem to the mixing of the partition resampling chain,
and then lift this chain to the parallel KMP chain,
whose mixing can be analyzed using the techniques developed above.
This reduction is summarized in \cref{fig:proof-roadmap-partition-chain},
and we elaborate on this connection in more detail below.

\begin{figure}[!t]
    \centering
    \resizebox{\linewidth}{!}{
    \begin{tikzpicture}[remember picture]
    \node (roadmap) {
    $
    \begin{tikzcd}[ampersand replacement=\&, column sep=13em, row sep=7em,
        cells={
            nodes={
                draw,
                align=center,
                rounded corners,
                inner sep=8pt,
                outer sep=4pt
            }
        }]
    {\shortstack{Parallel Kac twirling $\repkacchannel{t}{k}$ on $\symsubspace{t}{d}$}}
    \arrow[d, <-> ,"{\shortstack{\small Equivalent dynamics\\ \small (\cref{prop:markov-chain-transition-matrix-correspondence})}}"{name=equivalent-dynamics, description}]
    \&
    {\shortstack{Convergence of $\repkacchannel{t}{k}$ to Haar twirling
        \\ (\cref{thm:main})}}
    \\
    {\shortstack{Partition resampling chain \\ (\cref{sec:markov-chain-interpretation})}}
    \arrow[d, "\text{\small Lifting}"{name=kmp-lifting, description}]
    \&
    {\shortstack{Mixing time of the partition resampling chain \\ (\cref{thm:partition-resampling-mixing})}}
    \arrow[u, "{\shortstack{\small From chain mixing to channel mixing\\ \small (\cref{prop:spectral-gap-implies-channel-mixing})}}"{name=channel-mixing, description}]
    \\
    {\shortstack{Parallel KMP chain \\ (\cref{sec:convergence})}}
    \arrow[r,
    "{\text{\small Conditional product structure}}",
    "{\text{\small Exact coupling}}"'
    ]
    \&
    {\shortstack{Mixing time of the parallel KMP chain \\ (\cref{thm:partition-resampling-mixing})}}
    \arrow[u, "{\shortstack{\small Projection does not increase the mixing time\\ \small (\cref{lem:partition-dominates-by-lifted})}}"{name=kmp-projection, description}]
    \end{tikzcd}
    $
    };
    \draw[dashed, ->]
      (equivalent-dynamics.east) -- (channel-mixing.west);
    \draw[dashed, ->]
      (kmp-lifting.east) -- (kmp-projection.west);
    \end{tikzpicture}
    }

    \caption{Reduction from the parallel Kac twirling to the parallel KMP chain.
    The left column shows the relevant quantum and classical dynamics, and the right column shows the corresponding quantitative bounds.}
    \label{fig:proof-roadmap-partition-chain}
\end{figure}

\vspace{-0.8em}
\paragraph{Action of the Kac twirling channel on the symmetric subspace.}
We start by analyzing how one step of the parallel Kac twirling channel
\(
    \kacchannel{t} \coloneq \repkacchannel{t}{1}
\)
acts on the symmetric subspace.
By definition, \(\kacchannel{t}\) is the composition of
two simpler channels:
it first applies the block-Haar twirling channel \(\blockhaarchannel{t}\),
and then the random permutation twirling channel \(\randpermchannel{t}\).
That is,
\(
    \kacchannel{t} = \randpermchannel{t} \circ \blockhaarchannel{t}.
\)
To describe the action of these channels,
we study them on the basis of the symmetric subspace,
namely, the basis of \emph{generalized Dicke states}
(see \cref{eq:generalized-Dicke} for the definition).
Roughly, these states are indexed by multiplicity vectors
\(
r=(r_1,\dots,r_d)\in\N^d
\)
with
\(
\sum_{i=1}^d r_i=t
\),
where \(r_i\) records how many times the basis vector \(\ket{i}\) appears,
and a generalized Dicke state \(\ket{D_r}\) is the normalized uniform superposition over all basis strings with multiplicity vector \(r\).

A key fact is that the block-Haar channel \(\blockhaarchannel{t}\)
kills all off-diagonal terms in the generalized Dicke basis
(see \cref{lem:block-haar-kills-off-diagonal}),
and the random permutation channel \(\randpermchannel{t}\) then averages
$\ketbra{D_r}$ to the maximally mixed state of the subspace spanned by all
\(\ket{D_{r'}}\) such that \(r'\) is a permutation of \(r\)
(see \cref{lem:rand-perm-channel}).
Thus, after permutation twirling, the only information that remains
is the orbit of \(r\) under permutations of its coordinates.
We refer to this orbit as the \emph{type} of \(r\).
Intuitively, the type records only the pattern of \(r\), forgetting which coordinate carries which multiplicity number.
Equivalently, the type of a multiplicity vector \(r\) is
the partition \(\lambda\) of \(t\) with length at most \(d\),
obtained by sorting the entries of \(r\) in non-increasing order.
Let $\typesymsubspace{\lambda}{t}{d}$ denote the subspace spanned by all
\(\ket{D_{r'}}\) such that \(r'\) has type \(\lambda\).

Therefore, if we decompose the symmetric subspace into a direct sum of subspaces $\typesymsubspace{\lambda}{t}{d}$,
then for a state $\sigma$ in $\symsubspace{t}{d}$,
$\kacchannel{t}(\sigma)$ can be expressed as a convex combination of
the maximally mixed states $\rho_\lambda$ on the type subspaces $\typesymsubspace{\lambda}{t}{d}$ (see \cref{prop:decomp-after-kac-channel}).
Repeated application of \(\kacchannel{t}\) then induces discrete-time dynamics on these coefficients.

\vspace{-0.8em}
\paragraph{A classical description: the partition resampling chain.}
A main contribution of our work is to show
that this induced coefficient update is exactly captured by a {\em Markov chain}.
More specifically,
we identify the action of \(\kacchannel{t}\) with
a stochastic matrix $M$,
where the entry \(M(\lambda,\lambda')\) represents the probability of transitioning from $\rho_\lambda$ to $\rho_{\lambda'}$ under $\kacchannel{t}$
(see \cref{def:transition-matrix-kac-channel}).
Then a key conceptual shift is to identify the induced dynamics as a random walk on partitions, which we call the 
\emph{partition resampling chain}
(see \cref{def:partition-resampling-chain}).
Starting from a partition \(\lambda\),
this chain randomly pairs the \(d\) coordinates,
redistributes the total multiplicity number on each matched pair uniformly,
and then sorts the resulting multiplicity vector
to obtain a new partition \(\lambda'\).

An intuitive explanation
comes from the locality of the block-Haar action.
Indeed, a generalized Dicke state \(\ket{D_r}\)
is obtained by symmetrizing
\(
    \ket{1}^{\otimes r_1}\otimes \ket{2}^{\otimes r_2}\otimes \cdots \otimes \ket{d}^{\otimes r_d}.
\)
Now consider one block \(j\), whose basis states are \(\{2j-1,2j\}\).
Since the block-Haar channel acts independently on different blocks, it suffices to analyze its action on
\(
    \ket{2j-1}^{\otimes r_{2j-1}} \otimes \ket{2j}^{\otimes r_{2j}}
\)
for each block separately.
Notice that the block-Haar operation does not move a state
outside this block.
Thus,
it only redistributes the total multiplicity number of
\(r_{2j-1} + r_{2j}\) between the basis states \(\ket{2j-1}\) and \(\ket{2j}\).
As a result, on each block, \(\kacchannel{t}\) performs a local resampling
of how many copies are assigned to the two states,
while preserving their sum.
This coincides exactly with the mechanism of the partition resampling chain.
This connection allows us to translate the convergence of \(\repkacchannel{t}{k}\) into a mixing problem for the partition resampling chain.

\vspace{-0.8em}
\paragraph{Lifting to the parallel KMP chain.}
However, the partition resampling chain is not the most convenient object to study directly,
because sorting suppresses much of the information about the update.
We therefore lift it to the \emph{parallel KMP chain},
whose state space consists of full multiplicity vectors
\(
r\in\Delta=\{r\in\N^d:\sum_{i=1}^d r_i=t\},
\)
rather than only their sorted types.
Instead of recording only the partition \(\lambda\),
we remember exactly how much mass sits at each coordinate.
The update rule is as follows:
we choose a random perfect matching on the \(d\) coordinates,
and for each matched pair \(\{i,j\}\),
uniformly redistribute the total mass \(r_i+r_j\) between the two coordinates.
Projecting this multiplicity vector back to its type recovers the original partition resampling chain.
Since projection cannot increase total variation distance,
any mixing bound for the parallel KMP chain also applies
to the partition resampling chain
(see \cref{lem:partition-dominates-by-lifted}).

We then adapt the conditional product structure and exact coupling
techniques developed for KMP on general graphs to the parallel setting,
obtaining a mixing-time bound of
\(O(\log d+\log(1/\varepsilon))\), uniformly over all \(t\)
(see \cref{thm:partition-resampling-mixing,cor:patition-resampling-chain-mixing-time}).
This yields the same mixing bound for the partition resampling chain
and hence the convergence of \(\repkacchannel{t}{k}\)
to Haar twirling on the symmetric subspace (see \cref{thm:main}).

\subsection{Discussion and Future Work}
Our results uncover a two-way connection between the parallel Kac's walk and the KMP models in both directions. Ideas from the coupling analysis of Kac’s walk, combined with conditional product structure, yield mixing bounds for KMP that are uniform in the number of particles. Conversely, an exact reduction to a parallel KMP chain establishes rapid convergence of the parallel Kac twirling channel on the symmetric subspace. This convergence gives scalable pseudorandomness against non-adaptive queries whose joint input is supported on that subspace. These results suggest several directions for further investigation.

\begin{itemize}[topsep=0pt, itemsep=0pt, leftmargin=1.6em]

    \item Our analysis combines a two-particle description of covariance evolution with conditional product structure and an exact coupling. Related connections between many-particle mixing and simpler particle dynamics arise for exclusion processes~\cite{morris2006mixing,oliveira2013mixing}, while Caputo, Quattropani, and Sau have recently established two-particle spectral-gap identities for broad classes of stochastic exchange models~\cite{caputo2026}. A natural question is which structural conditions allow such information to yield total-variation mixing bounds uniform in the number of particles. In particular, can the combination of conditional product structure and exact coupling be extended to more general redistribution laws or block-update dynamics?


    \item Our result shows that the parallel Kac's walk as a twirling operation closely approximates a Haar random unitary in the symmetric subspace. Our result raises the question of how far scalable pseudorandomness can extend beyond symmetric query inputs. Brakerski and Yuen~\cite{BY26} show that an efficient construction of scalable $\pru$s in their random-oracle model would imply a positive solution to the Aaronson–Kuperberg unitary synthesis problem, a longstanding problem in quantum complexity for which a recent manuscript by OpenAI~\cite{OpenAI26} claims a positive resolution in the constant-error setting. Huang, Ramkumar, and Wright subsequently rule out query-efficient scalable $\pru$s in the random Boolean-oracle model~\cite{HRW26}, even against non-adaptive distinguishers with unrestricted joint inputs.

    First, which families of Schur-Weyl sectors admit convergence bounds uniform in the moment order \(t\)? Our analysis answers this question for the totally symmetric sector. A concrete next question is whether, for each fixed \(k\), such bounds hold on individual sectors indexed by Young diagrams \(\lambda\vdash t\) with \(t-\lambda_1\le k\), that is, with at most \(k\) boxes outside the first row. Can the twirling dynamics on these sectors be related to a tractable extension of the partition resampling chain? Understanding how convergence depends on the representation sector would clarify the scope of scalable pseudorandomness under restrictions on query inputs. Extending security to superpositions of sectors would additionally require controlling coherences between them.
    
    Second, can computational assumptions yield scalable $\pru$s beyond the statistical random-oracle approach? The above impossibility results do not rule out computationally secure scalable $\pru$s in general. A first target is security against arbitrary non-adaptive queries: do quantum-secure one-way functions suffice for such a construction? Addressing this question would require exploiting computational hardness without first obtaining a statistically secure full-space random-oracle ensemble.

\end{itemize}

\paragraph{Organization.}
The rest of the paper is organized as follows.
\cref{sec:preliminary} introduces the notation and preliminaries used throughout the paper.
In \cref{sec:kmp-chain}, we formally define the KMP chain on general graphs
and its labeled version, and establish their basic properties.
\cref{sec:mixing-kmp} is devoted to the mixing time of the KMP chain.
We turn to the continuous KMP chain,
bound its mixing time and apply the result to coordinate hit-and-run
on the simplex in \cref{sec:continuous-kmp}.
\cref{sec:scalable-pru} establishes the connection between the parallel
Kac's walk and the parallel KMP chain, and proves that the parallel Kac
twirling channel converges to Haar twirling on the symmetric subspace.
Building on this convergence, \cref{sec:non-adap-pru} constructs
scalable non-adaptive \pru s on the symmetric subspace.
Some proofs are deferred to the appendix.

\paragraph{Acknowledgments.}
Q.C. was supported in part by the National Natural Science Foundation of China (Grants Nos. 12271460 and 12341101), Guangdong provincial grants (Grant Nos. 2024A1
515011456 and GDZX2403006), and the Shenzhen Fundamental Research Program (Grant No. JCYJ20241202124023031).
J.L. is supported by National Natural Science Foundation of China Grant No. 62472212. Some of the KMP analysis was completed while J.L. was a postdoc at Caltech.
P.Y. and M.Z. were supported by the National Natural Science Foundation of China (Grant Nos. 62332009 and 12347104), the Quantum Science and Technology-National Science and Technology Major Project (Grant No. 2021ZD0302901), the NSFC/RGC Joint Research Scheme (Grant No. 12461160276), the Natural Science Foundation of Jiangsu Province (No. BK20243060), the Fundamental and Interdisciplinary Disciplines Breakthrough Plan of the Ministry of Education of China (No. JYB2025XDXM118), the ``111 Center'' (No. B26023), and the Fundamental Research Funds for the Central Universities (Grant No. 2026300376).

\paragraph{AI Disclosure.}
OpenAI's GPT-5.4 helped us explore the proof of conditional product structure in the case of the complete graph.
The generalization to arbitrary connected graphs was completed  by the authors,
as GPT at that time was unable to complete the proof even after we supplied a high-level sketch of the proof idea.
In preparation of this version, OpenAI's GPT-6 Astra was used to further scrutinize the proofs and polish the exposition of this manuscript.  The authors have reviewed and verified the entire manuscript and take full responsibility for its content.

\section{Preliminaries}
\label{sec:preliminary}
\paragraph{Basic notations.}
We write $\N=\{0,1,2,\dots\}$ for the set of natural numbers.
For $d\in\N$, let $\Br{d}=\{1,\dots,d\}$.
For a length-$t$ string $x \in [d]^t$ and an element $a \in [d]$,
we denote by $\#_a(x)$ the number of occurrences of $a$ in $x$,
namely,
\begin{equation}\label{eq:counting-function}
    \#_a(x) = |\{i \in [t] : x_i = a\}|\enspace.
\end{equation}
For pairwise disjoint sets $A_1$ and $A_2$, we write
\(
    A_1 \sqcup A_2
\)
for their disjoint union. In particular,
\(
    A_1 \sqcup \cdots \sqcup A_k = [t]
\)
means that the sets $A_1,\dots,A_k$ are pairwise disjoint and their union is $[t]$.
For a set \(A\subseteq\Omega\), we write \(A^{\mathsf c}=\Omega\setminus A\)
for its complement.

For $d\in\N$, let $\S_d$ denote the permutation group on $d$ elements,
$\su(d)$ the special unitary group on $d\times d$ complex matrices,
and $\ugroup{d}$ the unitary group of $d\times d$ complex matrices.
The Haar measure $\haardist$ on $\ugroup{d}$
is the unique probability measure that is
left- and right-invariant under the group action.
For a matrix $A$ and a positive integer $p$,
we define $\|A\|_p = (\mathrm{Tr}[(A^\dagger A)^{p/2}])^{1/p}$
to be the Schatten $p$-norm of $A$.
Moreover, $\|A\|_\infty$ is the operator norm of $A$.
For a vector $v \in \mathbb{C}^d$, define
$
    \norm{v}_1 = \sum_{i=1}^d |v_i|
$
and
$
    \norm{v}_2 = \sqrt{\sum_{i=1}^d |v_i|^2}
$.

\paragraph{Probability theory.}
For a finite set or a bounded interval \(S\), we write
\(\Unif(S)\) for the uniform distribution on \(S\). For two probability distributions $\mu$ and $\nu$ over a finite set $\Omega$,
the total variation distance between $\mu$ and $\nu$ is defined as
\[
    \|\mu - \nu\|_{\mathrm{TV}} = \frac{1}{2} \sum_{x \in \Omega} |\mu(x) - \nu(x)| \enspace.
\]

Let $\mu$ be a probability distribution over a finite set $\Omega_1$,
and let $f : \Omega_1 \to \Omega_2$ be a function,
where $\Omega_2$ is another finite set.
Then, we use $\mu f^{-1}$ to denote the distribution over $\Omega_2$ defined by
\[
    \mu f^{-1}(y) = \sum_{x\in\Omega_1:f(x)=y} \mu(x) \qquad
    \forall y \in \Omega_2 \enspace.
\]
We have the following facts about total variation distance:
\begin{fact}\cite[Section~2.1]{Meckes2019}
\label{fact:tv-continuous-functions}
Let \(K\subseteq\mathbb R^d\) be compact, and let
\(\mu,\nu\) be probability distributions on \(K\).
Then
\[
    \|\mu-\nu\|_{\mathrm{TV}}
    =
    \sup_{\substack{f:K\to[0,1]\\ f\text{ continuous}}}
    \left|
        \int_K f\,\mathrm d\mu
        -\int_K f\,\mathrm d\nu
    \right|.
\]
\end{fact}

\begin{fact}\label{fact:tv-tensor}
For distributions $\mu_1,\nu_1$ on a finite set $\Omega_1$,
and $\mu_2,\nu_2$ on a finite set $\Omega_2$,
we have
\[
\|\mu_1 \otimes \mu_2 - \nu_1 \otimes \nu_2\|_{\mathrm{TV}}
\le
\|\mu_1 - \nu_1\|_{\mathrm{TV}} + \|\mu_2 - \nu_2\|_{\mathrm{TV}} \enspace.
\]
\end{fact}

\begin{fact}\cite[Lemma 7.10]{LPW09}\label{fact:tv-distance-function}
    Let $\mu$ and $\nu$ be two probability distributions over a finite set $\Omega_1$,
    and let $f : \Omega_1 \to \Omega_2$ be a function, where $\Omega_2$ is another finite set.
    Then, we have
    \[
    	\|\mu f^{-1} - \nu f^{-1}\|_{\mathrm{TV}} \leq \|\mu - \nu\|_{\mathrm{TV}} \enspace.
    \]
\end{fact}

We will use several facts about
the Dirichlet distribution \cite[Chapter 27]{BN04}.
The distribution $\mathrm{Dirichlet}(1,\dots,1)$ is the uniform distribution over the simplex
\(
\left\{ u \in \mathbb{R}_{\ge 0}^d : \sum_{i=1}^d u_i = 1 \right\}.
\)
It can be generated from i.i.d.\ $\Gamma(1,1)$ random variables:
if $G_1,\dots,G_d$ are independently drawn from \ $\Gamma(1,1)$,
\[
\left(\frac{G_1}{\sum_{j=1}^d G_j},\dots,\frac{G_d}{\sum_{j=1}^d G_j}\right)
\sim \mathrm{Dirichlet}(1,\dots,1) \enspace.
\]
In particular, each marginal is distributed as $\mathrm{Beta}(1,d-1)$.
Let $u = (u_1,\dots,u_d)$ be drawn from $\mathrm{Dirichlet}(1,\dots,1)$.
Its product moment of order $(\alpha_1,\dots,\alpha_d) \in \N^d$ is
given by
\begin{equation} \label{eq:moment-diri}
    \E_u \Br{ \prod_{i=1}^d u_i^{\alpha_i} }
    =
    \frac{ (d-1)! \cdot \prod_{i=1}^d \alpha_i ! }{ ( d + \alpha_1 + \cdots + \alpha_d - 1)! } \enspace.
\end{equation}

\paragraph{Markov chain.}
A Markov chain on a finite state space $\Omega$ is specified by a transition
matrix $M \in \mathbb{R}^{\Omega \times \Omega}$ such that
$M(x,y) \ge 0$ for all $x,y \in \Omega$
and $\sum_{y \in \Omega} M(x,y)=1$ for every $x \in \Omega$.
For $k \in \mathbb{N}$, we write $M^k$ for the $k$-step transition matrix,
so that $M^k(x,\cdot)$ denotes the distribution of the chain
after $k$ steps when started from $x$.

A probability distribution $\pi$ on $\Omega$ is \emph{stationary} if
\(
\pi \cdot M = \pi
\),
or equivalently,
\[
\pi(y)=\sum_{x \in \Omega} \pi(x) \cdot M(x,y)
\qquad \forall y \in \Omega \enspace.
\]
The chain is \emph{irreducible} if for every $x,y \in \Omega$,
there exists some $k \in \mathbb{N}$ such that $M^k(x,y)>0$.
It is \emph{aperiodic} if for every $x \in \Omega$,
\[
\gcd\{k \ge 1 : M^k(x,x)>0\}=1 \enspace.
\]
Every irreducible and aperiodic Markov chain on a finite state space
has a unique stationary distribution.
The chain is \emph{reversible} w.r.t. a distribution $\pi$
if the detailed balance condition holds:
\[
\pi(x) \cdot M(x,y) = \pi(y) \cdot M(y,x)
\qquad \forall x,y \in \Omega \enspace.
\]
This detailed balance condition also implies that $\pi$ is stationary.
For $\varepsilon \in (0,1)$, the total-variation mixing time of the chain is defined as
\begin{equation*}
T^{\mathrm{mix}}(\varepsilon)
:=
\min\Bigl\{
k \in \mathbb{N} :
\max_{x \in \Omega} \|M^k(x,\cdot)-\pi\|_{\mathrm{TV}} \le \varepsilon
\Bigr\},
\end{equation*}
where $\pi$ is the stationary distribution. We also define 
\begin{equation*}
    \delta(k)
    :=\max_{x\in\Omega}
      \|M^k(x,\cdot)-\pi\|_{\mathrm{TV}}.
\end{equation*}
The following fact captures the exponential decay of the total variation distance,
which is a standard consequence of the eigenvalue analysis of reversible Markov chains; see the proof of \cite[Theorem 12.4]{LPW09}.

\begin{fact}\label{fact:mixing_time}
    Let $M$ be the transition matrix of a reversible, irreducible Markov chain with state space $\Omega$, and let $\pi_{\min}=\min_{x\in\Omega}\pi(x).$ Then
    \[ \delta(k) \le\frac{\lambda_2(M)^k}{2\sqrt{\pi_{\min}}} \enspace,\]
    where $\lambda_2(M)$ is the second largest eigenvalue in absolute value of $M$.
\end{fact}

We will use the following fact to extend a total variation bound to later times.
\begin{fact}[{\cite[Lemmas 4.10 and 4.11]{LPW09}}]
\label{fact:tv-block-decay}
Let \(M\) be the transition matrix of a Markov chain on a
finite state space \(\Omega\), with stationary distribution
\(\pi\). If \(\delta(\ell)\leq\varepsilon\) for some positive integer
\(\ell\) and \(0<\varepsilon<1/2\), then, for every integer
\(k\geq0\),
\(
    \delta(k)\leq
    (2\varepsilon)^{\lfloor k/\ell\rfloor}.
\)
\end{fact}

We also recall the following coupling lemma.
\begin{lemma}[Coupling lemma, Theorem 5.4 in \cite{LPW09}]
\label{lem:coupling_lemma}
Let \(M\) be the transition matrix of a Markov chain on a state space \(\Omega\),
with unique stationary distribution \(\pi\).
Let \((X_t)_{t\geq0}\) and \((Y_t)_{t\geq0}\) be two Markov chains with transition matrix \(M\),
started at \(X_0=x\in\Omega\) and \(Y_0\sim\pi\), respectively.
Then, for any coupling of the two chains and every integer \(t\geq0\),
\(\norm{M^t(x,\cdot)-\pi}_{\mathrm{TV}}
\leq\Pr[X_t\neq Y_t]\).
\end{lemma}

\paragraph{Partitions.}
A partition of a positive integer $t$ is a non-increasing sequence
$\lambda=(\lambda_1,\lambda_2,\dots,\lambda_k)$
of positive integers
$\lambda_1 \geq \lambda_2 \geq \cdots \geq \lambda_k > 0$
such that $\sum_{i=1}^k \lambda_i = t$.
We write $\lambda \vdash t$ to
indicate that $\lambda$ is a partition of $t$.
The length of a partition $\lambda$
is the number of positive integers in $\lambda$.
We denote by $\partition{t}{d}$ the set of
all partitions of $t$ of length at most $d$.
If $d$ is clear from context, we simply write $\partitions{t}$.

\paragraph{Symmetric subspace.}
For positive integers $t$ and $d$, let
$
    \symsubspace{t}{d} =
    \mathrm{span}\{\ket{u}^{\otimes t} : \ket{u} \in \C^d\}
$
denote the symmetric subspace of
$(\C^d)^{\otimes t}$,
and let $\symsubspaceprojfull{t}{d}$ denote the projector onto
$\symsubspace{t}{d}$.
Let
\begin{equation}\label{eq:maximally-mixed-state-symmetric-subspace}
    \rho_{\sym}^{t,d} =\frac{\symsubspaceprojfull{t}{d}}{\dim(\vee^t\C^d)}
\end{equation}
be the maximally mixed state on the symmetric subspace $\vee^t \mathbb{C}^d$.
When $t$ and $d$ are clear from context,
we simply write $\symsubspaceproj$ and $\rho_{\sym}$.

An orthonormal basis for $\symsubspace{t}{d}$
is given by the \emph{generalized Dicke states}.
For each \emph{multiplicity vector} $r \in \mathbb{N}^d$
satisfying $\sum_{j=1}^d r_j = t$,
the corresponding generalized Dicke state is defined as
\begin{equation}\label{eq:generalized-Dicke}
    \ket{D_r} \coloneqq \frac{1}{\sqrt{|\Lambda_r|}} \sum_{x \in \Lambda_r} \ket{x} \enspace,
\end{equation}
where $\Lambda_r$ is the set of all length-$t$ strings in $[d]^{t}$
that have multiplicity vector $r$, that is,
\begin{equation}\label{eq:omega-r}
    \Lambda_r \coloneqq \{x \in [d]^t : \#_{a}(x) = r_a \text{ for all } a \in [d]\} \enspace.
\end{equation}
The size of each set $\Lambda_r$ is given by the multinomial coefficient
\[
    |\Lambda_r| = \frac{t!}{r_1! r_2! \cdots r_d!} \enspace.
\]
The set of all generalized Dicke states
$\{\ket{D_r} : r \in \mathbb{N}^d, \sum_{j=1}^d r_j = t\}$
forms an orthonormal basis for the symmetric subspace $\symsubspace{t}{d}$.
In particular, the dimension of $\symsubspace{t}{d}$
equals the number of distinct multiplicity vectors $r$,
and hence is given by the following formula
\[
    \dim(\symsubspace{t}{d}) = \binom{t + d - 1}{d - 1} \enspace.
\]
The following fact shows that Haar averaging acts like
a completely depolarizing channel on the symmetric subspace
$\symsubspace{t}{d}$.

\begin{fact}\label{eq:haar-averaging-symmetric-subspace}
    Let $U$ be a Haar random unitary over $\ugroup{d}$.
    For any two generalized Dicke states
    $\ket{D_r}$ and $\ket{D_{r'}}$,
    we have
    \[
        \E_{U\gets\haardist}
        \Br{
            U^{\otimes t} \ketbratwo{D_r}{D_{r'}} (U^{\dagger})^{\otimes t}
        }
        = \delta_{r,r'} \cdot \rho_{\sym}^{t,d} \enspace,
    \]
    where $\delta_{r,r'}$ is $1$ if $r = r'$ and $0$ otherwise.
\end{fact}

We say that a multiplicity vector $r$ has type $\lambda$
if $\lambda = r^{\darrow}$, where $r^{\darrow}$ denotes the vector
obtained by rearranging the entries of $r$ in non-increasing order
and removing all zero entries.
For example, the multiplicity vector $r = (1, 2, 0)$ has type $\lambda = (2, 1)$.
Given a partition $\lambda \vdash t$ with length at most $d$,
we denote the set of all multiplicity vectors with type $\lambda$
by $\Xi_{\lambda, d}$, i.e.,
\begin{equation}\label{eq:xi-lambda}
    \Xi_{\lambda, d} = \{r \in \mathbb{N}^d : r^{\darrow} = \lambda\}\enspace.
\end{equation}
We denote by $\typesymsubspace{\lambda}{t}{d}$
the subspace of $\symsubspace{t}{d}$ spanned by
the generalized Dicke states of type $\lambda$, i.e.,
\[
    \typesymsubspace{\lambda}{t}{d} = \mathrm{span}\{\ket{D_r} : r \in \Xi_{\lambda, d}\} \enspace.
\]
The projector onto $\typesymsubspace{\lambda}{t}{d}$
and the maximally mixed state on $\typesymsubspace{\lambda}{t}{d}$
are defined by
\begin{equation}\label{eq:projector-and-state-lambda}
    \Pi_{\lambda,d} = \sum_{r \in \Xi_{\lambda, d}} \ketbra{D_r} \enspace, \enspace\enspace
    \rho_{\lambda,d} = \frac{\Pi_{\lambda,d}}{|\Xi_{\lambda, d}|}\enspace.
\end{equation}
We may write $\Xi_{\lambda}$, $\Pi_{\lambda}$ and $\rho_{\lambda}$,
when $d$ is clear from the context.
We have the following facts.

\begin{fact}\label{fact:orthogonality}
    For two partitions $\lambda, \mu \in \partition{t}{d}$,
    $\tr{\Pi_{\mu,d}\cdot\rho_{\lambda,d}}$ equals $1$ if $\mu = \lambda$ and equals $0$ otherwise.
\end{fact}

\begin{fact}\label{fact:decomposition-maximally-mixed-state}
    For integers $t, d \in \N$, $ \rho_{\sym}^{t,d} = \sum_{\lambda \in \partition{t}{d}}
    \frac{ |\Xi_{\lambda, d}| }{ \dim(\symsubspace{t}{d}) } \cdot \rho_{\lambda,d} $.
\end{fact}

For the quantum results, we adopt the following convention:
$n$ denotes the number of qubits,
$d = 2^n$ the dimension of the corresponding Hilbert space,
and $t$ the number of copies of this Hilbert space.
For the parallel classical chains, $d$ may be any even integer at least $2$.

\section{The KMP Process}
\label{sec:kmp-chain}
In this section, we formalize the KMP process \cite{KMP82}
and its labeled version on general graphs,
and establish their basic properties.

\subsection{The KMP Chain on General Graphs}
\label{sec:kmp-chain-general}
\begin{definition}[KMP chain on a general graph]
\label{def:kmp-general}
For positive integers \(t\) and \(d\), let \(G=(V = [d],E)\) be a connected graph on \(d\) vertices.
The KMP chain on \(G\) is the Markov chain on \(\Delta_G\) where
\[
    \Delta_G
    :=
    \set{r\in\N^d:\sum_{v\in V} r_v=t} \enspace.
\]
Given the current state \(r\in \Delta_G\),
the next state is obtained by the following procedure:
\begin{enumerate}
    \item sample a uniformly random edge \(e=\{i,j\}\in E\);\footnote{Throughout, we write each edge \(\{i,j\}\) with \(i<j\).}
    \item let \(m_{i,j}=r_i+r_j\) and resample \((r_i,r_j)\) uniformly from
    the set
    \[
        \set{(x,y)\in\N^2: x + y = m_{i,j} } \enspace,
    \]
    while keeping all other coordinates unchanged.
\end{enumerate}
We denote its transition matrix by \(\kmpchain^{(t,G)}\) or simply \(\kmpchain\) when the context is clear.
\end{definition}

Intuitively, the KMP chain repeatedly redistributes mass
between the two vertices of a uniformly random edge of $G$,
while preserving the total mass.
Its transition matrix is symmetric,
and the uniform distribution on $\Delta_G$ is the stationary distribution.

\begin{lemma}
    \label{lem:discrete-kmp-stationary}
    Let \(G=(V=[d],E)\) be a connected graph on \(d\geq 2\) vertices,
    and let \(t\) be a positive integer.
    The transition matrix \(\kmpchain^{(t,G)}\) is symmetric.
    Moreover, the KMP chain is irreducible, aperiodic, and reversible
    with respect to the uniform distribution on \(\Delta_G\),
    denoted by \(\kmpdist\), which is its unique stationary distribution.
\end{lemma}

\begin{proof}
    We first show that the transition matrix
    $\kmpchain$ is symmetric.
    Let \(R\) and \(R'\) denote the current and next configurations,
    respectively, and let \(e\) denote the edge selected in the update.
    Fix \(r,r'\in\Delta_G\).
    The transition probability from $r$ to $r'$ can be written as
    \[
    \kmpchain(r,r')
    =
    \expect{e\sim \Unif(E)}{
        \Pr[R'=r' \mid R=r,e]
    } \enspace.
    \]
    For a fixed edge \(e=\{i,j\}\in E\), 
    if \(r_i+r_j\neq r'_i+r'_j\) or $r_v\neq r'_v$ for some $v\in V\setminus\{i,j\}$, then
    \[
    \Pr[R'=r' \mid R=r, e]= \Pr[R'=r \mid R=r', e] = 0 \enspace.
    \]
    Otherwise, if \(r_i+r_j = r'_i+r'_j\) and $r_v = r'_v$ for all $v\in V\setminus\{i,j\}$, then
    \[
    \Pr[R'=r' \mid R=r, e]
    = \Pr[R'=r \mid R=r', e]
    = \frac{1}{ r_i+r_j+1} \enspace.
    \]
    Averaging over $ e$, we obtain
    \(
    \kmpchain(r,r')=\kmpchain(r',r)
    \).

    Since $\Delta_G$ is finite,
    symmetry implies that the uniform distribution on $\Delta_G$
    satisfies the detailed balance condition,
    and hence the chain is reversible and the uniform distribution is stationary.
    The chain is irreducible since, from any $r\in\Delta_G$,
    one can move all mass to $(t,0,\dots,0)$ and
    then split to any $r'\in\Delta_G$ with positive probability.
    It is aperiodic since every state has a positive self-loop probability.
    Since the chain is irreducible on a finite state space,
    the uniform distribution is the unique stationary distribution.
\end{proof}

\subsection{The Labeled  KMP Chain}
\label{sec:labeled-kmp-chain}

To analyze the mixing time of the KMP chain,
it is convenient to pass to a labeled version of the process.
For a vector \(r=(r_1,\dots,r_d)\in\Delta_G\),
we think of the total mass \(t\) as \(t\) distinct labeled particles,
with \(r_u\) particles
initially placed at vertex \(u\).
At each step, after sampling a uniformly random edge,
one redistributes all particles at the two vertices
according to a uniform split of their total mass.

\begin{definition}[Labeled KMP chain on a general graph]
\label{def:labeled-kmp}
For positive integers \(t\) and \(d\),
let \(G=(V = [d],E)\) be a connected graph on \(d\) vertices.
The space of labeled configurations of \(t\) particles on \(G\) is
\(\Omega_G\coloneqq[d]^t\).
For \(x\in\Omega_G\),
the coordinate \(x_\ell\in[d]\) records the vertex
where particle \(\ell\) is located.

The labeled KMP chain on $G$ with
$t$ particles is a Markov chain on
\(\Omega_G:=[d]^t\).
Given the current configuration \(x=(x_1,\dots,x_t)\in\Omega_G\),
the chain evolves as follows:
\begin{enumerate}
    \item sample a uniformly random edge \(e=\{i,j\}\in E\);

    \item let
    \(
        m_{i,j} := \#_i(x) + \#_j(x)
    \),
    where \(\#_i(x)\) counts the particles at vertex \(i\), as defined in
    \cref{eq:counting-function}.
    Choose an integer
    \[
        a_{i,j} \in \set{0,1,\dots,m_{i,j}}
    \]
    uniformly at random.
    Then, among the \(m_{i,j}\) particles at vertices \(i\) and \(j\),
    choose exactly \(a_{i,j}\) particles uniformly at random and place them
    at vertex \(i\); place all remaining particles at vertex \(j\).
    Leave all particles at other vertices unchanged.
\end{enumerate}
We denote its transition matrix by \(\lkmpchain^{(t,G)}\) or simply \(\lkmpchain\) when the context is clear.
\end{definition}

The stationary distribution of the labeled chain is no longer uniform on \(\Omega_G\).
In fact, a labeled configuration is weighted
according to the occurrence numbers it induces,
and all labeled configurations with
the same multiplicity vector contribute
equally to
the uniform stationary distribution of the unlabeled chain.

\begin{lemma}
\label{lem:labeled-kmp-chain-stationary}
Let \(G=(V=[d],E)\) be a connected graph on \(d\geq 2\) vertices,
    and let \(t\) be a positive integer.
The labeled KMP chain is irreducible and aperiodic.
Moreover, it is reversible with respect to the probability distribution
\(\lkmpdist\) on \(\Omega_G\) defined by
\[
    \lkmpdist(x)
    \coloneqq
    \frac{\prod_{i=1}^d \#_i(x)!}{t!\cdot\binom{t+d-1}{d-1}}
    \qquad \text{for every } x\in\Omega_G \enspace.
\]
In particular, \(\lkmpdist\) is the unique stationary
distribution of  \(\lkmpchain^{(t,G)}\).
\end{lemma}

\begin{proof}
    We first check that \(\lkmpdist\) is a probability distribution.
    For each \(r=(r_1,\dots,r_d)\in\Delta_G\), the number of labeled configurations
    \(x\in\Omega_G\) satisfying
    \(
        \bigl(\#_1(x),\dots,\#_d(x)\bigr)=r
    \)
    is exactly
    \(
        \frac{t!}{\prod_{i=1}^d r_i!}.
    \)
    Hence, we have
    \[
        \sum_{x\in\Omega_G} \prod_{i=1}^d \#_i(x)!
        =
        \sum_{r\in\Delta_G}
        \frac{t!}{\prod_{i=1}^d r_i!} \cdot \prod_{i=1}^d r_i!
        =
        t!\cdot|\Delta_G|
        =
        t!\cdot\binom{t+d-1}{d-1} \enspace.
    \]

    The chain is irreducible since, from any configuration \(x\in\Omega_G\),
    one can merge all particles to one vertex and then split to any configuration \(x'\in\Omega_G\) with positive probability.
    It is aperiodic since every configuration has a positive self-loop probability.
    It remains to verify detailed balance. Let \(X\) and \(X'\)
    denote the current and next configurations, respectively,
    and let \(e\) denote the edge selected in the update.
    Fix \(x,x'\in\Omega_G\). We will show that
    \[
        \lkmpdist(x)\cdot\lkmpchain(x,x')
        =
        \lkmpdist(x')\cdot\lkmpchain(x',x) \enspace.
    \]
    Fix an edge \(e=\{i,j\}\in E\), and define
    \(
        S_{i,j}(x):=\set{\ell\in[t]:x_\ell\in\set{i,j}}
    \)
    to be the set of particles located at vertices \(i\) and \(j\) in configuration \(x\),
    and similarly \(S_{i,j}(x')\).
    If \(S_{i,j}(x)\neq S_{i,j}(x')\), or \(x_\ell\neq x'_\ell\)
    for some \(\ell\notin S_{i,j}(x)\), then
    \[
        \Pr[X'=x' \mid X=x,e]= \Pr[X'=x \mid X=x',e]=0 \enspace.
    \]
    Now suppose that
    \(
        S_{i,j}(x)=S_{i,j}(x')
    \)
    and \(x_\ell=x'_\ell\) for every \(\ell\notin S_{i,j}(x)\).
    Let $m_{i,j} = |S_{i,j}(x)|$.
    Therefore, we have
    \[
        \Pr[X'=x' \mid X=x,e]
        =
        \frac{1}{m_{i,j}+1}\cdot \frac{1}{\binom{m_{i,j}}{\#_i(x')}} \enspace,
    \]
    since we must choose \(a_{i,j}=\#_i(x')\) and
    choose exactly the particles at vertex \(i\) in \(x'\).
    Similarly,
    \[
        \Pr[X'=x \mid X=x',e]
        =
        \frac{1}{m_{i,j}+1}\cdot \frac{1}{\binom{m_{i,j}}{\#_i(x)}} \enspace.
    \]
    On the other hand,
    the factors in \(\lkmpdist\) corresponding to vertices outside
    \(\{i,j\}\) agree:
    \[
        \prod_{v\notin\{i,j\}} \#_v(x)!
        =
        \prod_{v\notin\{i,j\}} \#_v(x')! \enspace.
    \]
    Moreover, since
    \(
        \#_i(x)+\#_j(x)=m_{i,j}=\#_i(x')+\#_j(x'),
    \)
    we have
    \[
        \#_i(x)!\cdot\#_j(x)!\cdot \frac{1}{\binom{m_{i,j}}{\#_i(x')}}
        =
        \#_i(x')!\cdot\#_j(x')!\cdot \frac{1}{\binom{m_{i,j}}{\#_i(x)}}
        =
        \frac{\#_i(x)!\cdot\#_j(x)!\cdot\#_i(x')!\cdot\#_j(x')!}{m_{i,j}!} \enspace.
    \]
    Multiplying by the common factors from the other vertices and
    by \(1/\bigl((m_{i,j}+1)t!\binom{t+d-1}{d-1}\bigr)\) yields
    \[
        \lkmpdist(x)\cdot\Pr[X'=x' \mid X=x,e]
        =
        \lkmpdist(x')\cdot\Pr[X'=x \mid X=x',e] \enspace.
    \]
    Finally, averaging over the uniformly random edge \(e\in E\)
    yields
    \[
        \lkmpdist(x)\cdot\lkmpchain(x,x')
        =
        \lkmpdist(x')\cdot\lkmpchain(x',x) \enspace.
    \]
    Thus \(\lkmpchain\) satisfies the detailed balance condition with respect to
    \(\lkmpdist\), and hence is reversible with respect to \(\lkmpdist\).
    Since the chain is irreducible and aperiodic,
    this stationary distribution is unique.
\end{proof}

The labeled chain contains strictly more information than
the unlabeled chain, as it records the exact locations of
individual particles rather than only the resulting multiplicity vector.
Thus, the unlabeled chain is obtained by forgetting the labels,
and its convergence can only be faster.

\begin{lemma}
\label{lem:labeled-dominates-lifted-general-kmp}
Let \(G=(V=[d],E)\) be a connected graph on \(d\geq 2\) vertices,
and let \(t\) be a positive integer.
Let \(x\in\Omega_G\), and let
\(
    r=(\#_1(x),\dots,\#_d(x))\in\Delta_G.
\)
Then for every \(k\ge 0\),
\[
    \Vert \kmpchain^k(r,\cdot)-\kmpdist \Vert_{\mathrm{TV}}
    \le
    \Vert \lkmpchain^k(x,\cdot)-\lkmpdist \Vert_{\mathrm{TV}} \enspace.
\]
\end{lemma}

\begin{proof}
    Let $f:\Omega_G\to\Delta_G$ be defined as
    \[
        f(x):= (\#_1(x),\dots,\#_d(x) ) \enspace.
    \]
    For a distribution $\mu$ on $\Omega_G$,
    recall that $\mu f^{-1}$ denotes the distribution on $\Delta_G$ defined by
    \[
        \mu f^{-1}(r) = \sum_{x\in\Omega_G:f(x)=r} \mu(x)
        = \sum_{x\in \Lambda_r} \mu(x)
        \qquad \text{for every } r\in\Delta_G \enspace,
    \]
    where $\Lambda_r$ is defined in \cref{eq:omega-r}.
    By the definitions of the two chains, we have
    \[
        \lkmpchain^k(x,\cdot)\,f^{-1}
        =
        \kmpchain^k(r,\cdot) \enspace.
    \]
    For each \(r\in\Delta_G\), summing over all labeled configurations
    with multiplicity vector \(r\) gives
    \[
       \lkmpdist\,f^{-1} (r)
       =
       \sum_{x\in \Lambda_r} \lkmpdist(x)
       =
       \frac{t!}{\prod_{i=1}^d r_i!}\cdot
        \frac{\prod_{i=1}^d r_i!}{t!\binom{t+d-1}{d-1}}
        =
        \frac{1}{\binom{t+d-1}{d-1}} 
        = 
        \kmpdist(r) \enspace.
    \]
    Thus, we have
    \[
        \lkmpdist\,f^{-1}=\kmpdist \enspace.
    \]
    Hence, by \cref{fact:tv-distance-function},
    \[
        \Vert \kmpchain^k(r,\cdot)-\kmpdist \Vert_{\mathrm{TV}}
        =
        \Vert \lkmpchain^k(x,\cdot)\,f^{-1}-\lkmpdist\,f^{-1} \Vert_{\mathrm{TV}}
        \le
        \Vert \lkmpchain^k(x,\cdot)-\lkmpdist \Vert_{\mathrm{TV}} \enspace.
    \]
    This completes the proof.
\end{proof}

\section{Mixing Time of the KMP Chain}
\label{sec:mixing-kmp}
In this section, we analyze the mixing time of the KMP chain
introduced in \cref{sec:kmp-chain-general}.
Our analysis uses the labeled KMP chain from
\cref{sec:labeled-kmp-chain}, whose dynamics admit a
\emph{conditional product structure},
as explained in \cref{sec:labeled-kmp-chain-conditional-structure}.
We exploit this structure to derive a mixing bound for the
labeled chain in \cref{sec:mixing-time-labeled-kmp-chain}.
This bound suffices when the number of particles is small.
For larger particle numbers, we use the coupling method
developed in \cref{sec:kmp-exact-coupling} to remove the
dependence on \(t\), obtaining a uniform mixing bound for
the unlabeled KMP chain in \cref{sec:kmp-uniform-mixing}.

\subsection{Conditional Product Structure of the Labeled KMP Chain}
\label{sec:labeled-kmp-chain-conditional-structure}

Before proving the mixing time of the labeled chain,
we describe several structural observations
that will be used in the argument.
The first observation gives a convenient reformulation
of the update rule using auxiliary randomness.
The second shows that, conditional on the auxiliary randomness,
the evolution of the labeled chain admits a product structure,
so that the labeled particles evolve independently.
Finally, we show that the stationary distribution admits
an analogous product structure.

Our first observation is that, 
for each selected edge \(\{i,j\}\),
there is an equivalent way to redistribute the labeled particles
at vertices \(i\) and \(j\).

\begin{observation}\label{obs:equiv-rule}
The following two sampling experiments are equivalent.
\begin{enumerate}
    \item Sample a uniformly random real number
    \(p_{i,j} \sim \mathrm{Unif}[0,1]\).
    Then, each of the $m_{i,j}$ particles
    independently decides to go to vertex $i$ with probability $p_{i,j}$, and to vertex $j$ with $1-p_{i,j}$.
    \item Sample a uniformly random integer \(a_{i,j} \in \set{0,1,\dots,m_{i,j}}
    \).
    Then, among all particles currently located at vertices \(i\)
    and \(j\), choose exactly \(a_{i,j}\) particles uniformly at random and place them
    at vertex \(i\); place all remaining particles at vertex \(j\).
\end{enumerate}
\end{observation}
\begin{proof}
        To see this,
    suppose that \(m_{i,j}\) labeled particles
    are currently on the vertices \(i\) and \(j\),
    and let \(Q\) be the number of particles sent to \(i\).
    For every
    \(0\le a\le m_{i,j}\), we have
    \[
        \Pr[Q=a]
        =
        \int_0^1 \binom{m_{i,j}}{a}\cdot  p_{i,j}^a \cdot (1-p_{i,j})^{m_{i,j}-a}\,\mathrm{d}p_{i,j}
        =
        \frac{1}{m_{i,j}+1} \enspace.
    \]
    Moreover, conditional on \(Q=a\),
    every \(a\)-subset of the \(m_{i,j}\) particles
    is equally likely to be chosen for vertex \(i\).
    This construction gives
    exactly the update rule of the labeled KMP chain in \cref{def:labeled-kmp}.
\end{proof}

Thus, we can equivalently describe the labeled KMP chain as follows.

\begin{definition}[Labeled KMP chain on a general graph with auxiliary randomness]
\label{def:lkmp-general}
For positive integers \(t\) and \(d\), let \(G=(V = [d],E)\) be a connected graph on \(d\) vertices.
The labeled KMP chain on $G$ with
$t$ particles is a Markov chain on
\(\Omega_G:=[d]^t\).
Given the current configuration \(x=(x_1,\dots,x_t)\in\Omega_G\),
the chain evolves as follows:
\begin{enumerate}
    \item sample a uniformly random edge \(e=\{i,j\}\in E\);
    \item sample \(p\sim\Unif[0,1]\), and for each particle currently located at \(i\) or \(j\),
    independently move it to \(i\) with probability \(p\)
    and to \(j\) with probability \(1-p\);
    \item leave all other particles unchanged.
\end{enumerate}
\end{definition}

The advantage of this formulation is that,
conditional on the entire sequence of selected edges and
the associated sequence of uniform random variables throughout the process,
the distribution of the labeled configuration at every time step has a product structure.

Let us fix an initial state \(x^{(0)}\in\Omega_G\),
the random edges $\{\{i_s, j_s\} \in E\}_{1\leq s\leq k}$
and the random variables $\{p_{s}\}_{1\leq s\leq k}$ used in the first \(k\) steps of the labeled chain.
Then, for each time $s$,
define a random row-stochastic matrix
\(A_s\in\R^{d\times d}\) :
\begin{equation}\label{eq:def-A-non-parallel}
    A_s(i,j) \coloneqq \begin{cases}
        p_s & \text{if } (i,j) = (i_s, i_s) \text{ or } (i,j) = (j_s, i_s) \\
        1-p_s & \text{if } (i,j) = (i_s, j_s) \text{ or } (i,j) = (j_s, j_s) \\
        1 & \text{if } i=j \notin \set{i_s,j_s} \\
        0 & \text{otherwise}
    \end{cases}\enspace.
\end{equation}
    Thus \(A_s\) is the one-step transition matrix for
    a single labeled particle at time \(s\),
    conditional on \(\{i_s, j_s\}\)
    and \(p_s\).
    Moreover, we define 
    \begin{equation}\label{eq:def-B-non-parallel}
        B_s := A_1 \cdot A_2 \cdots A_s \enspace.
    \end{equation}

\begin{remark}\label{rmk:transition-matrix-labeled-chain-non-parallel}
    Let $A$ be the random row-stochastic matrix defined as in \cref{eq:def-A-non-parallel} from a uniformly random edge \(e=\{i,j\}\in E\) and a random variable \(p\sim\Unif[0,1]\).
    The transition matrix of the labeled chain can be expressed as
    \[
        \lkmpchain^{(t,G)}
        =
        \E_{e,p}[A^{\otimes t}] \enspace.
    \]
    Indeed, conditioned on \( e \) and \(p\),
    the labeled particles evolve independently according to the transition matrix \(A\),
    so the conditional transition matrix of the labeled chain is \(A^{\otimes t}\).
    Averaging over the randomness of \( e \) and \(p\) yields the above expression.
\end{remark}

\begin{observation}\label{obs:conditional-product-structure}
    Conditional on \(\{\{i_s,j_s\}\}_{1\le s\le k}\)
    and \(\{p_s\}_{1\le s\le k}\),
    since the labeled particles evolve independently,
    a particle initially at vertex \(i\) has the following distribution after \(k\) steps:
    \begin{equation*}
        q_i^{(k)} = e_i\cdot B_k \enspace,
    \end{equation*}
    where \(e_i\) is the row vector with a \(1\)
    in coordinate \(i\) and \(0\) elsewhere.
    Therefore,
    the conditional distribution of the labeled configuration
    at time \(k\) admits a product structure:
    \begin{equation*}
        \mu_k
        \coloneqq
        \bigotimes_{\ell=1}^t q_{x_\ell^{(0)}}^{(k)} \enspace.
    \end{equation*}
    In particular, the unconditional distribution after \(k\) steps of
    the labeled chain can be obtained by averaging \(\mu_k\)
    over \(\{\{i_s,j_s\}\}_{1\le s\le k}\) and
    \(\{p_s\}_{1\le s\le k}\).
    That is, for every \(y\in\Omega_G\),
    \begin{equation*}
        \lkmpchain^k(x^{(0)},y)
        =
        \E_{
            \substack{\{\{i_s,j_s\}\}_{1\le s\le k},\,
            \{p_s\}_{1\le s\le k}}}
        [\mu_k(y)] \enspace.
    \end{equation*}
\end{observation}

Our third observation is that the stationary distribution
\(\lkmpdist\) of the labeled chain also admits a similar conditional
product structure.
Recall that the distribution $\mathrm{Dirichlet}(1,\dots,1)$
is the uniform distribution over
$\{ u \in \mathbb{R}_{\ge 0}^d : \sum_{i=1}^d u_i = 1 \}$,
and admits the
representation $u_i = G_i / \sum_j G_j$ for i.i.d.\ $G_i \sim \Gamma(1,1)$.

\begin{lemma}
\label{lem:labeled-kmp-stationary-product-structure}
Let \(k\in\N\), and \(u\) be a random vector drawn from the distribution
\(\Dirichlet(1,\dots,1)\).
For a sequence of random edges \(\{\{i_s,j_s\}\}_{1\leq s\leq k}\)
and a sequence of random variables \(\{p_s\}_{1\leq s\leq k}\)
independent of \(u\),
define \(A_s\) and \(B_s\) as in
\cref{eq:def-A-non-parallel,eq:def-B-non-parallel},
and let \(w^{(k)}=u\cdot B_k\).
Then, for every \(y\in\Omega_G\),
\[
    \lkmpdist(y)
    =
    \E_{
        \substack{\{\{i_s,j_s\}\},\,\{p_s\},\,u}
    }\Br{
        \prod_{i=1}^t w^{(k)}(y_i)
    } \enspace,
\]
or equivalently,
\[
    \lkmpdist
    =
    \E_{
        \substack{\{\{i_s,j_s\}\},\,\{p_s\},\,u}
    }\Br{\br{w^{(k)}}^{\otimes t}} \enspace.
\]
\end{lemma}

\begin{proof}
    We first show that the unconditional distribution of \(w^{(k)}\)
    is \(\Dirichlet(1,\dots,1)\).
    For this, it suffices to verify the following one-step invariance:
    for any fixed edge \(e = \{i,j\}\),
    let \(A\) be the random row-stochastic matrix defined as
    in \cref{eq:def-A-non-parallel} using a random variable
    \(p\sim\Unif[0,1]\),
    \[
        u\cdot A \sim \Dirichlet(1,\dots,1)\enspace.
    \]
    To see this, we can write \(u\) as a normalized vector of i.i.d. \(\Gamma(1,1)\) random variables:
    \[
        u_i = \frac{G_i}{G_1+\cdots+G_d}\enspace,
    \]
    where \(G_1,\dots,G_d\) are i.i.d. random variables drawn from
    distribution \(\Gamma(1,1)\) .
    Let $u' = u\cdot A$ and we have
    \[
        u_i' = p\cdot (u_i+u_j) = \frac{p\cdot(G_i+G_j)}{G_1+\cdots+G_d}\enspace,
        \quad
        u_j' = (1-p)\cdot (u_i+u_j) = \frac{(1-p)\cdot(G_i+G_j)}{G_1+\cdots+G_d}\enspace.
    \]
    Now let us denote $G'_i = p\cdot(G_i+G_j)$ and $G'_j = (1-p)\cdot(G_i+G_j)$,
    with \(G'_v=G_v\) for \(v\notin\{i,j\}\).
    We have
    \[
        u'_i = \frac{G'_i}{G'_1+\cdots+G'_d} \enspace,
        \quad
        u'_j = \frac{G'_j}{G'_1+\cdots+G'_d} \enspace.
    \]
    One can check that \(\{G'_i\}_i\) are independent \(\Gamma(1,1)\) random variables as well.
    Hence, \(u\cdot A\) has the distribution \(\Dirichlet(1,\dots,1)\).

    To prove the lemma,
    let
    \(
        \nu_k \coloneqq (w^{(k)})^{\otimes t}.
    \)
    It suffices to compute \(\E[\nu_k(y)]\) for each
    \(y\in\Omega_G\).
    Fix \(y=(y_1,\dots,y_t)\in\Omega_G\), and
    we have
    \[
        \nu_k(y)
        =
        \prod_{\ell=1}^t w^{(k)}(y_\ell)
        =
        \prod_{i=1}^d (w^{(k)}(i))^{\#_i(y)} \enspace,
    \]
    where \(\#_i(y)\) is the number of particles at vertex \(i\) in configuration \(y\).
    Then \(\E[\nu_k(y)]\) is exactly a moment of the Dirichlet distribution.
    Using the moment formula in \cref{eq:moment-diri},
    \[
        \E[\nu_k(y)]
        =
        \E\Br{\prod_{i=1}^d (w^{(k)}(i))^{\#_i(y)}}
        =
        \frac{(d-1)!\cdot\prod_{i=1}^d \#_i(y)!}{(t+d-1)!} \enspace.
    \]
    On the other hand, by definition of \(\lkmpdist\) in \cref{lem:labeled-kmp-chain-stationary}, we have
    \[
        \lkmpdist(y)
        =
        \frac{\prod_{i=1}^d \#_i(y)!}{t!\binom{t+d-1}{d-1}}
        =
        \frac{\prod_{i=1}^d \#_i(y)!}{t!}
        \cdot
        \frac{t!(d-1)!}{(t+d-1)!}
        =
        \frac{(d-1)!\cdot\prod_{i=1}^d \#_i(y)!}{(t+d-1)!} \enspace.
    \]
    Therefore, we have
    \[
        \lkmpdist
    =
    \E_{
        \substack{\{\{i_s,j_s\}\},\,\{p_s\},\,u}
    }\Br{ \br{w^{(k)}}^{\otimes t}} \enspace. \qedhere
    \]
\end{proof}

The conditional product structure helps reduce
the analysis of the labeled chain
to the behavior of a single particle.
Indeed, once conditioned on the auxiliary randomness,
the distribution at time \(k\) is a product measure.
By \cref{fact:tv-tensor},
the total variation distance is then bounded
by the sum of the total variation distances
of the single-particle marginals.
This insight underlies our mixing analysis in the next subsection.

\subsection{Mixing Time of the Labeled KMP Chain}
\label{sec:mixing-time-labeled-kmp-chain}

We now give the mixing bound for
the labeled KMP chain on general graphs,
which uses the second largest eigenvalue of the two-particle labeled chain \(\lkmpchain^{(2,G)}\)
as the contraction ratio.

\begin{theorem}
\label{thm:lkmp-mixing-time-general-graph}
For positive integers \(t\), \(d\) and $k$,
let \(G=(V=[d],E)\) be a connected graph on \(d\) vertices.
Let $\{x^{(s)}\}_{s\geq 0} $ be the labeled KMP chain on \(G\) with
transition matrix \(\lkmpchain\) as defined in
\cref{def:lkmp-general}.
Then, for any initial state \(x^{(0)}\in\Omega_G\), we have
\[
    \norm{\lkmpchain^k(x^{(0)},\cdot)-\lkmpdist}_{\mathrm{TV}}
    \le  2 t\cdot d^{5/4} \cdot \lambda_2^{k/2}\enspace,
\]
where \(\lkmpdist\) is the stationary distribution of \(\lkmpchain\),
and \(\lambda_2\) is the second largest eigenvalue in absolute value of the transition matrix \(\lkmpchain^{(2,G)}\).
\end{theorem}

\begin{proof}
    The case $d\leq 2$ is trivial, so we assume $d\geq 3$.
    Fix an initial state $x^{(0)}\in\Omega_G$.
    Recall the notations and discussions in
    \cref{obs:conditional-product-structure},
    conditional on \(\{\{i_s,j_s\}\}_{1\leq s\leq k}\) and
    \(\{p_s\}_{1\leq s\leq k}\),
    the distribution of the labeled configuration
    at time \(k\) is the product measure
    \[
        \mu_k
        =
        \bigotimes_{\ell=1}^t q_{x_\ell^{(0)}}^{(k)} \enspace,
    \]
    where
    \(
        q_i^{(k)} = e_i\cdot B_k .
    \)
    Moreover, for every \(y\in\Omega_G\),
    \begin{equation}
        \lkmpchain^k(x^{(0)},y)
        =
        \E_{
            \substack{
                \{\{i_s,j_s\}\},\ \{p_s\}
            }
        }
        [\mu_k(y)]
        \enspace.
        \label{eq:decomposition-distribution}
    \end{equation}
    On the other hand, by
    \cref{lem:labeled-kmp-stationary-product-structure},
    \begin{equation}
        \lkmpdist
        =
        \E_{
            \substack{\{\{i_s,j_s\}\},\,\{p_s\},\,u}
        }\Br{ \br{w^{(k)}}^{\otimes t}}
        \enspace,
        \label{eq:decomposition-stationary}
    \end{equation}
    where
    \(
        w^{(k)} = u\cdot B_k
    \)
    and
    \(
        u\sim \Dirichlet(1,\dots,1).
    \)
    Combining \cref{eq:decomposition-distribution} and \cref{eq:decomposition-stationary}, we have
    \[
        \norm{\lkmpchain^k(x^{(0)},\cdot)-\lkmpdist}_{\mathrm{TV}}
        \le
        \E_{
            \substack{\{\{i_s,j_s\}\},\,\{p_s\},\,u}
        }
        \Br{\norm{\mu_k-(w^{(k)})^{\otimes t}}_{\mathrm{TV}}} \enspace.
    \]
    Using the tensor product structure, by \cref{fact:tv-tensor} we can further bound the right-hand side by
    \begin{align}\label{eq:tv-bound-kmp-general-graph}
        \norm{\lkmpchain^k(x^{(0)},\cdot)-\lkmpdist}_{\mathrm{TV}}
        &\le
        \sum_{i=1}^d \#_i(x^{(0)}) \cdot \E_{\{\{i_s,j_s\}\},\,\{p_s\},\,u}
        \Br{\Vert q_i^{(k)}-w^{(k)} \Vert_{\mathrm{TV}}} \nonumber\\
        &\le
        \sum_{i=1}^d \#_i(x^{(0)}) \cdot \frac{\sqrt{d}}{2} \cdot
        \E_{\{\{i_s,j_s\}\},\,\{p_s\},\,u}
        \Br{\Vert q_i^{(k)}-w^{(k)} \Vert_{2}} \nonumber\\
        &=
        \sum_{i=1}^d \#_i(x^{(0)}) \cdot \frac{\sqrt{d}}{2} \cdot
        \E_{\{\{i_s,j_s\}\},\,\{p_s\},\,u}
        \Br{\norm{(e_i-u)\cdot B_k}_2} \enspace.
    \end{align}
 
    We now introduce the following linear map from
    \(d\times d\) matrices to length-\(d^2\) row vectors:
    \[
        \mathrm{vec}(e_i^{\intercal}e_j)=e_i\otimes e_j
        \qquad\text{for all } i,j\in[d] \enspace.
    \]
    Then, for any $d\times d$ matrices $X$, $M$ and $N$, we have
    \begin{align}
        \mathrm{vec}( M X N )
        = \mathrm{vec}(X) (M^{\intercal} \otimes N) \enspace.
    \end{align}
    And it is evident that $ \norm{X}_2 = \norm{ \mathrm{vec}(X) }_2 $.
    Using this notation, we have
    \begin{align*}
        \E_{\{\{i_s,j_s\}\},\,\{p_s\},\,u}
        \Br{\norm{(e_i-u)\cdot B_k}_2^2}
        &=
        \tr{ \E_{\{\{i_s,j_s\}\},\,\{p_s\},\,u}
        \Br{B_k^{\intercal}(e_i-u)^{\intercal}(e_i-u)B_k} } \\
        &\leq
        \sqrt{d}\cdot
        \norm{ \E_{\{\{i_s,j_s\}\},\,\{p_s\},\,u}
        \Br{B_k^{\intercal}(e_i-u)^{\intercal}(e_i-u)B_k} }_2\\
        &=\sqrt{d}\cdot
        \norm{ \E_{\{\{i_s,j_s\}\},\,\{p_s\},\,u}
        \Br{\mathrm{vec}\br{B_k^{\intercal}(e_i-u)^{\intercal}(e_i-u)B_k}} }_2 \\
        &=\sqrt{d}\cdot
        \norm{ \E_{\{\{i_s,j_s\}\},\,\{p_s\},\,u}
        \Br{\br{(e_i-u)\otimes(e_i-u)}\br{B_k\otimes B_k}}}_2
    \end{align*} 
Let $\pi^{(2)}$ be the stationary distribution of $\lkmpchain^{(2,G)}=\E_{e,p}[A\otimes A]$ (see \cref{rmk:transition-matrix-labeled-chain-non-parallel}).
Expanding
\[
    (e_i-u)\otimes(e_i-u)
    =
    e_i\otimes e_i
    - e_i\otimes u
    - u\otimes e_i
    + u\otimes u \enspace,
\]
and inserting \(\pi^{(2)}\) into each term, we obtain
\begin{align*}
    &\E_{\{\{i_s,j_s\}\},\,\{p_s\},\,u}
    \Br{\br{(e_i-u)\otimes(e_i-u)}(B_k\otimes B_k)} \\
    ={}&
    \E_{\{\{i_s,j_s\}\},\,\{p_s\}}
    \Br{\br{e_i\otimes e_i-\pi^{(2)}}(B_k\otimes B_k)}
    - \E_{\{\{i_s,j_s\}\},\,\{p_s\},\,u}
    \Br{\br{e_i\otimes u-\pi^{(2)}}(B_k\otimes B_k)} \\
    &\quad
    - \E_{\{\{i_s,j_s\}\},\,\{p_s\},\,u}
    \Br{\br{u\otimes e_i-\pi^{(2)}}(B_k\otimes B_k)} 
    + \E_{\{\{i_s,j_s\}\},\,\{p_s\},\,u}
    \Br{\br{u\otimes u-\pi^{(2)}}(B_k\otimes B_k)}\,.
\end{align*}
Applying the triangle inequality and then the bound
\(\norm{\cdot}_2 \le 2\norm{\cdot}_{\mathrm{TV}}\), we get
        \begin{align}
            \E_{\{\{i_s,j_s\}\},\,\{p_s\},\,u}
        \Br{\norm{(e_i-u)\cdot B_k}_2^2}\leq~&2\sqrt{d}\cdot
        \norm{\br{e_i\otimes e_i-\pi^{(2)}}\cdot\E_{\{\{i_s,j_s\}\},\,\{p_s\}}
        \Br{B_k\otimes B_k}}_{\mathrm{TV}}\nonumber\\
        &+2\sqrt{d}\cdot \E_{u} \Br{
        \norm{\br{e_i\otimes u-\pi^{(2)}}\cdot \E_{\{\{i_s,j_s\}\},\,\{p_s\}}
        \Br{B_k\otimes B_k}}_{\mathrm{TV}}}\nonumber\\
        &+2\sqrt{d}\cdot \E_{u} \Br{
        \norm{ \br{u\otimes e_i-\pi^{(2)}}\cdot \E_{\{\{i_s,j_s\}\},\,\{p_s\}}
        \Br{B_k\otimes B_k}}_{\mathrm{TV}}}\nonumber\\
        &+2\sqrt{d}\cdot \E_{u} \Br{
        \norm{ \br{u\otimes u-\pi^{(2)}}\cdot \E_{\{\{i_s,j_s\}\},\,\{p_s\}}
        \Br{B_k\otimes B_k}}_{\mathrm{TV}}} \nonumber
        \enspace.
        \end{align}
        Since
        $
            \E[ B_k\otimes B_k ] = \E_{e,p}[A\otimes A]^k = (\lkmpchain^{(2,G)})^k,
        $
        each of the four terms on the right-hand side is the total variation distance at time \(k\) from the stationary distribution \(\pi^{(2)}\) for the two-particle labeled chain, started respectively from \(e_i\otimes e_i\), \(e_i\otimes u\), \(u\otimes e_i\), and \(u\otimes u\).
          By convexity of total variation distance, the bound in \cref{fact:mixing_time} for point-mass initial states extends to arbitrary initial distributions. 
Therefore, by \cref{fact:mixing_time} with $\pi^{(2)}_{\mathrm{min}} = \frac{1}{d(d+1)}$, we have
    \begin{align}
    \E_{\{\{i_s,j_s\}\},\,\{p_s\},\,u}
    \Br{\norm{(e_i-u)\cdot B_k}_2^2}
    \le
    4\cdot d\sqrt{d+1}\cdot \lambda_2^k \enspace.
    \label{eq:contraction-under-B-kmp-general-graph}
\end{align}
    Combining \cref{eq:tv-bound-kmp-general-graph,eq:contraction-under-B-kmp-general-graph} and using Jensen's inequality, we have
    \[
    	\norm{\lkmpchain^k(x^{(0)},\cdot)-\lkmpdist}_{\mathrm{TV}}
        \leq  t\cdot d(d+1)^{1/4} \cdot \lambda_2^{k/2} 
        \leq 2 t\cdot d^{5/4} \cdot \lambda_2^{k/2}\enspace.\qedhere
    \]
\end{proof}

The theorem above controls the mixing time
through the second eigenvalue of the two-particle labeled chain.
The following comparison further relates this eigenvalue to that
of the single-particle KMP chain.
We defer its proof to \cref{sec:lambda-2}.

\begin{restatable}{proposition}{twotoone}
\label{cor:two-particle-labeled-gap-compared-to-single-particle}
For every connected graph \(G=(V=[d],E)\),
\[
    1-\lambda_2\!\br{\lkmpchain^{(2,G)}}
    \ge
    \frac{2}{3}\cdot
    \br{1-\lambda_2\!\br{\kmpchain^{(1,G)}}} \enspace.
\]
\end{restatable}

Let \(T^{\mathrm{mix},(t,G)}_{\scriptscriptstyle \mathsf{LKMP}}(\varepsilon)\) denote the mixing time of the labeled KMP chain on \(G\) with \(t\) particles.
We have the following mixing time bound, which is a corollary of \cref{thm:lkmp-mixing-time-general-graph} and \cref{cor:two-particle-labeled-gap-compared-to-single-particle}.

\begin{corollary}\label{cor:mixing-time-labeled-kmp-general-graph}
    For every \(\varepsilon\in(0,1)\) and positive integers \(t\) and \(d\),
    let \(G=(V=[d],E)\) be a connected graph on \(d\) vertices.
    Let $\gamma^{(1,G)}_{\scriptscriptstyle \mathsf{KMP}}$ 
    denote the spectral gap of the single-particle random walk $\kmpchain^{(1,G)}$. Then,
    \[
    T^{\mathrm{mix},(t,G)}_{\scriptscriptstyle \mathsf{LKMP}}(\varepsilon) = O\! \left(\frac{1}{\gamma^{(1,G)}_{\scriptscriptstyle \mathsf{KMP}}}\cdot\br{\log d + \log t + \log (1/\varepsilon)} \right).
    \]
\end{corollary}

\begin{proof}
    Let
    \(
        \delta(k)=\sup_{x\in\Omega_G}
\|\lkmpchain^k(x,\cdot)-\lkmpdist\|_{\mathrm{TV}}.
    \)
By
\cref{thm:lkmp-mixing-time-general-graph} and
\cref{cor:two-particle-labeled-gap-compared-to-single-particle},
\[
    \delta(k)
    \leq 2t d^{5/4} \cdot (\lambda_2(\lkmpchain^{(2,G)}))^{k/2}
    \leq 2t d^{5/4} \cdot \exp\!\left(-\bigl(1-\lambda_2(\lkmpchain^{(2,G)})\bigr)\cdot k/2\right)
    \leq 2t d^{5/4}\cdot \e^{-\gamma^{(1,G)}_{\scriptscriptstyle \mathsf{KMP}}\cdot k/3}\enspace.
\]
Thus, taking $k = O(\br{\log d + \log t + \log (1/\varepsilon)}/\gamma^{(1,G)}_{\scriptscriptstyle \mathsf{KMP}})$ makes the right-hand side at most \(\varepsilon\).
\end{proof}

\subsection{Exact Coupling for the KMP Chain}
\label{sec:kmp-exact-coupling}

The mixing time of the labeled KMP chain in \cref{cor:mixing-time-labeled-kmp-general-graph}
    already gives an \(O((\log d)/\gamma^{(1,G)}_{\scriptscriptstyle \mathsf{KMP}})\) mixing time for the unlabeled KMP chain
    when \(t\) is polynomially bounded in \(d\).
    To obtain a bound uniform in \(t\),
    this subsection develops a coupling method for larger particle numbers,
    which is inspired by the two-stage coupling approach in \cite{Smi14,PS17}.

    We will construct a coupling of two KMP chains
    that use the same random edges
    and agree with high probability at the final step.
    We begin with a single update on a shared edge,
    with \(B\) particles at its two vertices in the first chain
    and \(D\) in the second.
    The following lemma provides two ways to couple their updates.

\begin{lemma}
\label{lem:kmp-integer-split-couplings}
Let \(B,D\in\mathbb N\) be nonnegative integers.
\begin{enumerate}[label=(\roman*)]
    \item 
    Sample \(\xi\sim\Unif[0,1)\), and set
    \(
        Z=\left\lfloor(B+1)\xi\right\rfloor
    \)
    and
    \(
        \widetilde Z=\left\lfloor(D+1)\xi\right\rfloor.
    \)
    Then, we have that
    \(Z\sim\Unif\{0,\ldots,B\}\),
    \(\widetilde Z\sim\Unif\{0,\ldots,D\}\),
    and
    \begin{equation}
        |Z-\widetilde Z|
        ~+~|(B-Z)-(D-\widetilde Z)|
        ~=~|B-D| \enspace.
        \label{eq:kmp-quantile-integer-split}
    \end{equation}

    \item
    Let \(A,C\in\mathbb N\) and \(h,L\geq0\) satisfy
    \(|A-C|\leq h\), \(|B-D|\leq h\), and \(B,D\geq L\).
    There exists a coupling of
    \[
        Z\sim\Unif\{0,\ldots,B\}
        \quad\text{and}\quad
        \widetilde Z\sim\Unif\{0,\ldots,D\}
    \]
    such that \(\Pr[A+Z\neq C+\widetilde Z]\leq 3h/(L+1)\).
\end{enumerate}

\end{lemma}

\begin{proof}
For each \(z\in\{0,\ldots,B\}\), we have
\(Z=z\) if and only if
\(\xi\in[z/(B+1),(z+1)/(B+1))\), an interval of length
\(1/(B+1)\). Thus \(Z\) is uniform on \(\{0,\ldots,B\}\).
The same argument applies to \(\widetilde Z\).
Assume without loss of generality that \(B\geq D\).
Since \(B-D\) is a nonnegative integer and \(0\leq\xi<1\),
\(
    \lfloor(D+1)\xi\rfloor
    \leq \lfloor(B+1)\xi\rfloor
    \leq \lfloor(D+1)\xi+(B-D)\rfloor
    = \lfloor(D+1)\xi\rfloor+B-D.
\)
Thus \(0\leq Z-\widetilde Z\leq B-D\).
Thus, both \(Z-\widetilde Z\) and
\((B-Z)-(D-\widetilde Z)\) are nonnegative, and hence
\[
    |Z-\widetilde Z|+|(B-Z)-(D-\widetilde Z)|
    =(Z-\widetilde Z)+(B-D)-(Z-\widetilde Z)
    =B-D \enspace.
\]
This proves \cref{eq:kmp-quantile-integer-split}.

We now prove the second claim and assume \(B\geq D\) without loss of generality.
The shifted variables \(A+Z\) and \(C+\widetilde Z\) must be uniform on
\(I=\{A,A+1,\ldots,A+B\}\) and
\(J=\{C,C+1,\ldots,C+D\}\), respectively.
Since \(|A-C|\leq h\) and \(|B-D|\leq h\),
the left endpoints of \(I\) and \(J\) differ by at most \(h\),
and the right endpoints differ by at most \(2h\).
Hence \(|I\cap J|\geq B+1-3h\).
Since \(B\geq D\), the total variation distance between the
uniform distributions on \(I\) and \(J\) is
\[
    1-\frac{|I\cap J|}{B+1}
    \leq \frac{3h}{B+1}
    \leq \frac{3h}{L+1}\enspace.
\]
By maximal coupling, there exists a coupling of
\(A+Z\sim\Unif(I)\) and \(C+\widetilde Z\sim\Unif(J)\) such that
\(\Pr[A+Z\neq C+\widetilde Z] \leq 3h/(L+1)\).
Subtracting \(A\) and \(C\), respectively, gives the required
coupling of \(Z\sim\Unif\{0,\ldots,B\}\) and
\(\widetilde Z\sim\Unif\{0,\ldots,D\}\).
\end{proof}

We next combine these couplings along a fixed sequence of edge updates.
The following lemma shows that,
when the updated edges form a connected graph
and the initial multiplicity vectors are sufficiently close,
the two chains can be coupled so that their final states are equal with high probability. 

\begin{lemma}
\label{lem:kmp-scheduled-coalescence}
Let \(d\geq3\) and \(t>d^{1000}\), and let
\(G=(V=[d],E)\) be a connected graph.
Fix a sequence of edges \(e_1,\ldots,e_N\) on the vertex set
\([d]\) such that the graph \(([d],\{e_1,\ldots,e_N\})\) is connected.
Let \(X^{(0)}\) and \(\widetilde X^{(0)}\) be random vectors in \(\Delta_G\) with an arbitrary joint distribution such that \(\widetilde X^{(0)}\sim\Unif(\Delta_G)\), and let
\(
    \eta
    =
    \Pr [
        \|X^{(0)}-\widetilde X^{(0)}\|_1>td^{-19}
    ].
\)

Then there exists a coupling of two processes
\((X^{(s)})_{s=0}^N\) and \((\widetilde X^{(s)})_{s=0}^N\)
with the given initial joint distribution, each applying the
redistribution rule in \cref{def:kmp-general} to
the edge \(e_s\) at step \(s\), such that
\[
    \Pr\!\left[X^{(N)}\neq\widetilde X^{(N)}\right]
    ~\leq~ \eta+\frac{1}{d^2}\enspace.
\]
\end{lemma}

\begin{proof}

We begin by constructing a sequence of partitions
\(\{\mathcal Q_s\}_{s=0}^N\) of \([d]\)
backward in time, starting with the singleton partition
\(
    \mathcal Q_N = \{\{1\},\ldots,\{d\}\}
\).
For \(s\) decreasing from \(N-1\) to \(0\),
if the two vertices of \(e_{s+1}\) lie in different parts of \(\mathcal Q_{s+1}\),
we obtain \(\mathcal Q_s\) by merging those two parts.
Otherwise, set \(\mathcal Q_s=\mathcal Q_{s+1}\).

\vspace{-0.6em}
\paragraph{Constructing the coupling.}

We now construct the coupling forward in time.
Fix \(0\leq s<N\), write \(e_{s+1}=\{i,j\}\) and set
\begin{equation*}
    B_s=X_i^{(s)}+X_j^{(s)}
    \quad
    \text{and}
    \quad
    D_s=\widetilde X_i^{(s)}+\widetilde X_j^{(s)} \enspace.
\end{equation*}
We jointly sample \(Z_s\) and \(\widetilde Z_s\), the numbers
of particles at vertex \(i\) after the update in the two chains,
as follows:
\begin{itemize}
    \item If \(\mathcal Q_s=\mathcal Q_{s+1}\), sample
        \(\xi\sim\Unif[0,1)\) independently
        and set
        \[
            Z_s=\lfloor(B_s+1)\xi\rfloor
            \quad
            \text{and}
            \quad
            \widetilde Z_s=\lfloor(D_s+1)\xi\rfloor \enspace.
        \]
        This is the first coupling in \cref{lem:kmp-integer-split-couplings}.

    \item Otherwise, a set \(S\in\mathcal Q_s\) splits into two sets
        \(S_i,S_j\in\mathcal Q_{s+1}\) containing \(i\) and \(j\),
        respectively.
        Define
        \begin{equation*}
            A_s=\sum_{z\in S_i\setminus\{i\}}X_z^{(s)}
            \quad
            \text{and}
            \quad
            C_s=\sum_{z\in S_i\setminus\{i\}}\widetilde X_z^{(s)} \enspace,
        \end{equation*}
        which are the numbers of particles in \(S_i\setminus\{i\}\) in the two chains.
        Sample
        \[
            Z_s\sim\Unif\{0,\ldots,B_s\}
            \quad\text{and}\quad
            \widetilde Z_s\sim\Unif\{0,\ldots,D_s\}
        \]
        using the second coupling in
        \cref{lem:kmp-integer-split-couplings}, with
        \(A=A_s\), \(C=C_s\), \(B=B_s\), and \(D=D_s\).
\end{itemize}
In both cases, update the two coordinates by
\[
    (X_i^{(s+1)},X_j^{(s+1)}) \gets (Z_s,B_s-Z_s)
    \quad
    \text{and}
    \quad
    (\widetilde X_i^{(s+1)},\widetilde X_j^{(s+1)})
    \gets(\widetilde Z_s,D_s-\widetilde Z_s)\enspace,
\]
and leave all other coordinates unchanged.

\vspace{-0.6em}
\paragraph{Decomposing the event of failure.}
Let
\(
    \mathcal A
    =
    \{
        \|X^{(0)}-\widetilde X^{(0)}\|_1\leq td^{-19}
    \}.
\)
By assumption, \(\Pr[\mathcal A^{\mathsf c}]=\eta\).
Thus,
\begin{equation}
\begin{aligned}
    \Pr\!\left[X^{(N)}\neq\widetilde X^{(N)}\right]
    &~\leq~\eta~+~\Pr\!\left[
        X^{(N)}\neq\widetilde X^{(N)} \wedge \mathcal A
    \right] \enspace.
\end{aligned}
\label{eq:scheduled-initial-failure}
\end{equation}

By assumption, the graph with edges $e_1,\ldots,e_N$ is connected,
so $\mathcal Q_0=\{[d]\}.$
Let \(\mathcal E_s\) be the event that the two chains have
the same number of particles in each part of \(\mathcal Q_s\):
\[
    \mathcal E_s
    =
    \left\{
        \sum_{z\in S}X_z^{(s)}
        =
        \sum_{z\in S}\widetilde X_z^{(s)}
        \text{ for every }S\in\mathcal Q_s
    \right\} \enspace.
\]
Since both chains have \(t\) particles,
\(\mathcal E_0\) always holds.
Since \(\mathcal Q_N\) consists of singletons,
\(\mathcal E_N\) is precisely the event
\(X^{(N)}=\widetilde X^{(N)}\).
By construction, the two vertices of \(e_{s+1}\) belong to the same part of \(\mathcal Q_s\).
Thus, the update leaves the number of particles in each part of \(\mathcal Q_s\) unchanged.
Consequently,
we have
\begin{equation}\label{eq:backwardcontaining}
    \mathcal E_{s+1}\subseteq\mathcal E_s \enspace,
    \quad \text{and}\quad
    \mathcal E_{s+1}=\mathcal E_s
    \quad\text{if} \quad\mathcal Q_{s+1}=\mathcal Q_s \enspace.
\end{equation}

Let \(s_1<\cdots<s_{d-1}\) denote the split times \(s\)
at which \(\mathcal Q_s\neq\mathcal Q_{s+1}\).
There are exactly \(d-1\) such times, since each split increases
the number of parts by one, from one part to \(d\) singletons.
Since \(\mathcal E_0\) always holds, if \(\mathcal E_N\) fails,
there must be a first step at which \(\mathcal E_s\) holds
but \(\mathcal E_{s+1}\) fails.
This can occur only when
\(\mathcal Q_s\neq\mathcal Q_{s+1}\).
Therefore,
\[
\begin{aligned}
    \{X^{(N)}\neq\widetilde X^{(N)}\}
    ~=~\mathcal E_N^{\mathsf c}
    ~=~\bigvee_{s=0}^{N-1}
        (\mathcal E_s\wedge\mathcal E_{s+1}^{\mathsf c})
    ~=~\bigvee_{r=1}^{d-1}
        (\mathcal E_{s_r}\wedge\mathcal E_{s_r+1}^{\mathsf c}) \enspace.
\end{aligned}
\]
Since the events \(\mathcal E_{s_r}\wedge\mathcal E_{s_r+1}^{\mathsf c}\)
are pairwise disjoint, 
\[
    \{X^{(N)}\neq\widetilde X^{(N)}\}\wedge\mathcal A
    =\bigvee_{r=1}^{d-1}
        \left(\mathcal A\wedge\mathcal E_{s_r}
        \wedge\mathcal E_{s_r+1}^{\mathsf c}\right) \enspace.
\]
Let \(\mathcal B_r\) be the event that at least one of the
two vertices of \(e_{s_r+1}\) has fewer than \(td^{-12}\)
particles in \(\widetilde X^{(s_r)}\).
Writing $e_{s_r+1}=\{i_r,j_r\}$, this is
\[
    \mathcal B_r
    =\{\widetilde X_{i_r}^{(s_r)}<td^{-12}\}
     \vee\{\widetilde X_{j_r}^{(s_r)}<td^{-12}\} \enspace.
\]
Using a union bound and applying
\cref{eq:scheduled-initial-failure}, we obtain
\begin{align}
    \Pr\!\left[X^{(N)}\neq\widetilde X^{(N)}\right]
    &~\leq~\eta+\sum_{r=1}^{d-1}\Pr\!\left[
        \mathcal A\wedge\mathcal E_{s_r}
        \wedge\mathcal E_{s_r+1}^{\mathsf c}\right] \nonumber\\
    &~=~\eta+\sum_{r=1}^{d-1}\Pr\!\left[
        \mathcal A\wedge\mathcal E_{s_r}
        \wedge\mathcal E_{s_r+1}^{\mathsf c}\wedge\mathcal B_r\right]+\sum_{r=1}^{d-1}\Pr\!\left[
        \mathcal A\wedge\mathcal E_{s_r}
        \wedge\mathcal E_{s_r+1}^{\mathsf c}\wedge\mathcal B_r^{\mathsf c}\right] \nonumber\\
    &~\leq~\eta+\sum_{r=1}^{d-1}\Pr[\mathcal B_r] +\sum_{r=1}^{d-1}\Pr\!\left[
        \mathcal A\wedge\mathcal E_{s_r}
        \wedge\mathcal E_{s_r+1}^{\mathsf c}\wedge\mathcal B_r^{\mathsf c}\right] \enspace.
\label{eq:scheduled-failure-decomposition}
\end{align}
It suffices to give a good bound on \(\Pr[\mathcal B_r]\) and
\(\Pr[\mathcal A\wedge\mathcal E_{s_r}
\wedge\mathcal E_{s_r+1}^{\mathsf c}\wedge\mathcal B_r^{\mathsf c}]\).

\vspace{-0.6em}
\paragraph{Bounding $\Pr[\mathcal B_r]$.}

We first bound the second term in \cref{eq:scheduled-failure-decomposition}. 
Each fixed-edge update preserves \(\Unif(\Delta_G)\).
Indeed, for a vector uniformly distributed on \(\Delta_G\),
fixing the coordinates outside \(\{i,j\}\) fixes the
remaining particle total \(D_s\), and each pair
\((0,D_s),(1,D_s-1),\ldots,(D_s,0)\) is equally likely.
The update samples exactly this distribution.
Since \(\widetilde X^{(0)}\) is uniform on \(\Delta_G\)
and the constructed process follows this update rule,
\(\widetilde X^{(s)}\) is uniform on \(\Delta_G\)
at every time \(s\).

For \(W\sim\Unif(\Delta_G)\) and \(z\in\{0,\ldots,t\}\),
we have that
\begin{align*}
    \Pr[W_i<td^{-12}]
    &~\leq~\sum_{z=0}^{\lfloor td^{-12}\rfloor}\Pr[W_i=z]
    ~=~\sum_{z=0}^{\lfloor td^{-12}\rfloor}
    \frac{\binom{t-z+d-2}{d-2}}{\binom{t+d-1}{d-1}}
    ~\leq~ \sum_{z=0}^{\lfloor td^{-12}\rfloor}
    \frac{\binom{t+d-2}{d-2}}{\binom{t+d-1}{d-1}}\\
    &~\leq~(td^{-12}+1) \cdot \frac{d-1}{t+d-1}
    ~\leq~ (td^{-12}+1) \cdot \frac{d}{t}
    ~\leq~ d^{-11}+\frac{d}{t}
    ~\leq~ 2d^{-11}\enspace,
\end{align*}
where the last inequality follows from \(t>d^{1000}\).
Applying a union bound over the two vertices of \(e_{s_r+1}\)
gives
\begin{equation*}
    \Pr[\mathcal B_r]\leq4d^{-11} \enspace.
\end{equation*}

\vspace{-0.6em}
\paragraph{Bounding failure when \(\mathcal B_r\) does not occur.}

We now bound the third term in \cref{eq:scheduled-failure-decomposition}.

To this end, we first bound
$\|X^{(s)}-\widetilde X^{(s)}\|_1$
on the event $\mathcal A\wedge\mathcal E_s$.
We first claim that on \(\mathcal E_s\), every \(S\in\mathcal Q_s\) satisfies
\begin{equation}
    \sum_{z\in S}
    |X_z^{(s)}-\widetilde X_z^{(s)}|
    ~\leq~
    \|X^{(0)}-\widetilde X^{(0)}\|_1 \enspace.
    \label{eq:parallel-coupling-part-discrepancy}
\end{equation}
We prove this by induction on \(s\).
For \(s=0\), the only part of \(\mathcal Q_0\) is \([d]\),
so the claimed inequality holds with equality.
Now suppose the claim holds at times \(0,\ldots,s\), and consider
the update from time \(s\) to \(s+1\).
Suppose that coordinates \(i\) and \(j\) are updated.
\begin{itemize}
    \item If \(\mathcal Q_s=\mathcal Q_{s+1}\),
        \cref{eq:kmp-quantile-integer-split} gives
        \[
            |X_i^{(s+1)}-\widetilde X_i^{(s+1)}|
            +|X_j^{(s+1)}-\widetilde X_j^{(s+1)}|
            =|B_s-D_s|
            \leq
            |X_i^{(s)}-\widetilde X_i^{(s)}|
            +|X_j^{(s)}-\widetilde X_j^{(s)}| \enspace.
        \]
        Since \(i\) and \(j\) lie in the same part,
        the above inequality shows that the discrepancy within this part cannot increase.
        All other parts remain unchanged, so the induction hypothesis gives
        \cref{eq:parallel-coupling-part-discrepancy} at time \(s+1\).

    \item If \(\mathcal Q_s\neq\mathcal Q_{s+1}\),
    then a set \(S\in\mathcal Q_s\) splits into two sets \(S_i,S_j\in\mathcal Q_{s+1}\).
    On \(\mathcal E_{s+1}\), the two chains have the same number of particles in \(S_i\) after the update.
    Since only coordinate \(i\) changes within \(S_i\), we have
    \[
        X_i^{(s+1)}-\widetilde X_i^{(s+1)}
        =-\sum_{z\in S_i\setminus\{i\}}
        \bigl(X_z^{(s)}-\widetilde X_z^{(s)}\bigr)\enspace.
    \]
    Moreover, \(\mathcal E_{s+1}\subseteq\mathcal E_s\) in \cref{eq:backwardcontaining} implies
    \(\sum_{z\in S}(X_z^{(s)}-\widetilde X_z^{(s)})=0\).
    Consequently,
    \begin{align*}
        \sum_{z\in S_i}
        |X_z^{(s+1)}-\widetilde X_z^{(s+1)}|
        &=\sum_{z\in S_i\setminus\{i\}}
        |X_z^{(s)}-\widetilde X_z^{(s)}|
        +\left|\sum_{z\in S_i\setminus\{i\}}
        \bigl(X_z^{(s)}-\widetilde X_z^{(s)}\bigr)\right|\\
        &=\sum_{z\in S_i\setminus\{i\}}
        |X_z^{(s)}-\widetilde X_z^{(s)}|
        +\left|\sum_{z\in S\setminus(S_i\setminus\{i\})}
        \bigl(X_z^{(s)}-\widetilde X_z^{(s)}\bigr)\right|\\
        &\leq\sum_{z\in S}|X_z^{(s)}-\widetilde X_z^{(s)}| \leq\|X^{(0)}-\widetilde X^{(0)}\|_1\enspace,
    \end{align*}
    where the last inequality follows from the induction hypothesis.
    The same argument applies to \(S_j\).
    Since all other parts remain unchanged, this completes the induction.
\end{itemize}
Since \(\mathcal Q_s\) has at most \(d\) parts, summing
\cref{eq:parallel-coupling-part-discrepancy} over all parts gives,
on \(\mathcal A\wedge\mathcal E_s\),
\begin{equation}
    \|X^{(s)}-\widetilde X^{(s)}\|_1
    \leq d\,\|X^{(0)}-\widetilde X^{(0)}\|_1
    \leq td^{-18}\enspace.
    \label{eq:parallel-coupling-global-discrepancy}
\end{equation}

Now fix some split time \(s_r\) in \(\{s_1,\ldots,s_{d-1}\}\)
and we bound
$\Pr [\mathcal A\wedge\mathcal E_{s_r}\wedge\mathcal E_{s_r+1}^{\mathsf c}\wedge\mathcal B_r^{\mathsf c}]$.
Write \(e_{s_r+1}=\{i,j\}\).
Recall the notation \(A_{s_r},B_{s_r},C_{s_r},D_{s_r}\) introduced in the
coupling construction.
On \(\mathcal A\wedge\mathcal E_{s_r}\wedge\mathcal B_r^{\mathsf c}\),
\cref{eq:parallel-coupling-global-discrepancy} gives
\(|A_{s_r}-C_{s_r}|\leq td^{-18}\) and
\(|B_{s_r}-D_{s_r}|\leq td^{-18}\).
Moreover, \(\mathcal B_r^{\mathsf c}\) ensures that both \(i\) and \(j\)
have at least \(td^{-12}\) particles in \(\widetilde X^{(s_r)}\),
so \(D_{s_r}\geq2td^{-12}\) and
\(B_{s_r}\geq D_{s_r}-td^{-18}\geq td^{-12}\).
On \(\mathcal E_{s_r}\), the event \(\mathcal E_{s_r+1}^{\mathsf c}\) is
equivalent to \(A_{s_r}+Z_{s_r}\neq C_{s_r}+\widetilde Z_{s_r}\). Thus
\begin{equation}
\begin{aligned}
    \mathcal A\wedge\mathcal E_{s_r}\wedge\mathcal E_{s_r+1}^{\mathsf c}
        \wedge\mathcal B_r^{\mathsf c} ~=~\mathcal A\wedge\mathcal E_{s_r}\wedge\mathcal B_r^{\mathsf c}
        \wedge\{A_{s_r}+Z_{s_r}\neq C_{s_r}+\widetilde Z_{s_r}\} \enspace.
\end{aligned}
\label{eq:scheduled-good-history-failure}
\end{equation}
By \cref{lem:kmp-integer-split-couplings} with \(h=td^{-18}\) and \(L=td^{-12}\),
the probability of \(A_{s_r}+Z_{s_r}\neq C_{s_r}+\widetilde Z_{s_r}\) at the next update is at most
$3d^{-6}$.
Therefore, \cref{eq:scheduled-good-history-failure} gives
\begin{align*}
    \Pr[\mathcal A\wedge\mathcal E_{s_r}
        \wedge\mathcal E_{s_r+1}^{\mathsf c}\wedge\mathcal B_r^{\mathsf c}]
    ~\leq~ 3d^{-6} \enspace.
\end{align*}

\vspace{-0.6em}
\paragraph{Combining the probability estimates.}

Substituting the two bounds into \cref{eq:scheduled-failure-decomposition} gives
\begin{align*}
    \Pr\!\left[X^{(N)}\neq\widetilde X^{(N)}\right]
    &\leq\eta+\sum_{r=1}^{d-1}\Pr[\mathcal B_r] +\sum_{r=1}^{d-1}\Pr[\mathcal A\wedge\mathcal E_{s_r}\wedge\mathcal E_{s_r+1}^{\mathsf c}\wedge\mathcal B_r^{\mathsf c}]\\
    &\leq\eta+4(d-1)d^{-11}+3(d-1)d^{-6} \leq\eta+\frac{1}{d^2} \enspace. 
\end{align*}
This completes the proof.
\end{proof}

\subsection{Mixing Time of the KMP Chain}
\label{sec:kmp-uniform-mixing}

We are now ready to establish our main result,
a mixing time bound for the KMP chain on general graphs
that is independent of the number of particles.
Let \(T^{\mathrm{mix},(t,G)}_{\scriptscriptstyle \mathsf{KMP}}(\varepsilon)\) denote the mixing time of the KMP chain on \(G\) with \(t\) particles.

\begin{restatable}{theorem}{kmpgeneralgraph}
\label{thm:kmp-mixing-time-general-graph}
    For every \(\varepsilon\in(0,1)\), every positive integer \(t\),
    and every integer \(d\geq 2\),
    let \(G=(V=[d],E)\) be a connected graph on \(d\) vertices.
    Let $\gamma^{(1,G)}_{\scriptscriptstyle \mathsf{KMP}}$
    denote the spectral gap of the single-particle random walk $\kmpchain^{(1,G)}$. Then, it holds that
    \[
    T^{\mathrm{mix},(t,G)}_{\scriptscriptstyle \mathsf{KMP}}(\varepsilon) = O\! \left(\frac{1}{\gamma^{(1,G)}_{\scriptscriptstyle \mathsf{KMP}}}\cdot\br{\log d + \log (1/\varepsilon)} \right).
    \]
\end{restatable}

\begin{remark}\label{rmk:kmp-mixing-time-general-graph}
   From the standard relation between mixing time and spectral gap (see, e.g., \cite[Theorem 12.5]{LPW09}), \cref{thm:kmp-mixing-time-general-graph} implies that
   \[
	  T^{\mathrm{mix},(t,G)}_{\scriptscriptstyle \mathsf{KMP}}(\varepsilon)
    =
    O\!\left(
        T^{\mathrm{mix},(1,G)}_{\scriptscriptstyle \mathsf{KMP}}(\varepsilon)
        \cdot
        (\log d+\log (1/\varepsilon))
    \right)
    \enspace.
    \]
\end{remark}

For the rest of this subsection, write
$
    \gamma
    \coloneqq
    \gamma^{(1,G)}_{\scriptscriptstyle \mathsf{KMP}}
$.
We first show that, for large particle numbers, the total variation
distance to stationarity is small after
\(O(\gamma^{-1}\log d)\) steps.

\begin{lemma}
\label{lem:kmp-large-particle-one-block}
Assume that \(d\geq3\) and \(t>d^{1000}\).
Let \(k_0=100\gamma^{-1}\log d\), \(T=6\gamma^{-1}\log d\) and \(k_\star=k_0+T\).
Then, we have
\[
    \sup_{r\in\Delta_G}
    \norm{\kmpchain^{k_\star}(r,\cdot)-\kmpdist}_{\mathrm{TV}}
    \leq \frac{4}{d^2} \enspace.
\]
\end{lemma}

\begin{proof}
Fix \(r\in\Delta_G\). Consider two copies of the KMP chain, starting from \(R^{(0)}=r\) and
\(\widetilde R^{(0)}\sim\kmpdist\), respectively.
By \cref{lem:coupling_lemma}, it suffices to construct a coupling of $(R^{(s)})_{s=0}^{k_\star}$ and $(\widetilde R^{(s)})_{s=0}^{k_\star}$ such that
\[
    \Pr\!\left[R^{(k_\star)}\neq \widetilde R^{(k_\star)}\right] \leq \frac{4}{d^2} \enspace.
\]
We use a two-stage coupling construction.

\paragraph{The first stage.}
For the first \(k_0\) steps, we construct the coupling by
assigning labels to the particles in each copy and using
the update rule in \cref{def:lkmp-general}.
Specifically, we use the same random edges $\{\{i_s,j_s\}\}_{1\leq s\leq k_0}$ and random parameters \(\{p_s\}_{1\leq s\leq k_0}\) for the two labeled chains.

Recall from \cref{eq:def-A-non-parallel,eq:def-B-non-parallel} that
\[
    B_{k_0}=A_1\cdots A_{k_0} \enspace,
\]
where \(A_s\) is the single-particle transition matrix
determined by \(e_s\) and
\(p_s\).
Thus, a particle initially at
\(i\) is at \(j\) after \(k_0\) steps with probability
\(B_{k_0}(i,j)\). 
Therefore, for $j\in[d]$, the expected number of particles
at vertex \(j\) after \(k_0\) steps in the first chain,
conditioned on the selected edges and split parameters, is
\begin{equation}\label{eq:kmp-conditional-mean}
    \mathbb E[R_j^{(k_0)}]
    ~=~ \sum_{i=1}^d r_i \cdot B_{k_0}(i,j)
    ~=~ (rB_{k_0})_j \enspace,
\end{equation}
where the expectation is conditional on the selected edges and split parameters.
By \cref{lem:labeled-kmp-stationary-product-structure}, the second
chain can be initialized by sampling
\(u\sim\Dirichlet(1,\ldots,1)\) and then placing each particle
independently according to \(u\).
Conditioned on \(u\), the selected edges, and the split parameters,
each particle is at vertex \(j\) after \(k_0\) steps with
probability
\[
    \sum_{i=1}^d u_i \cdot B_{k_0}(i,j)=(uB_{k_0})_j\enspace.
\]
Summing over its \(t\) particles gives
\begin{equation}\label{eq:kmp-stationary-conditional-mean}
    \mathbb E[\widetilde R_j^{(k_0)}]
    =t \cdot (uB_{k_0})_j \enspace.
\end{equation}
Define \(a^{(k_0)}:=t^{-1}(rB_{k_0})\) and \(b^{(k_0)}:=uB_{k_0}\).
Then, we have \(\mathbb E[R^{(k_0)}/t]=a^{(k_0)}\) and
\(\mathbb E[\widetilde R^{(k_0)}/t]=b^{(k_0)}\),
where the expectations are conditioned on \(u\), the selected edges,
and the split parameters.
Since \(a^{(k_0)}-b^{(k_0)}
=\sum_{i=1}^d(r_i/t)(e_i-u)B_{k_0}\),
where \(r_i/t\geq0\) and \(\sum_{i=1}^d r_i/t=1\),
the convexity of the squared Euclidean norm gives
\begin{align}
    \mathop{\mathbb E}_{\substack{u, \{\{i_s,j_s\}\},\{p_s\}}}
    \!\left[\|a^{(k_0)}-b^{(k_0)}\|_2^2\right]
    ~&\leq~\sum_{i=1}^d\frac{r_i}{t}\cdot
    \mathop{\mathbb E}_{\substack{u, \{\{i_s,j_s\}\},\{p_s\}}}
    \!\left[\|(e_i-u)B_{k_0}\|_2^2\right] \nonumber\\
    ~&\leq~\sum_{i=1}^d\frac{r_i}{t}
    \cdot 4d\sqrt{d+1}\cdot\lambda_2(\lkmpchain^{(2,G)})^{k_0}\nonumber\\
    &~=4d\sqrt{d+1}\cdot\lambda_2(\lkmpchain^{(2,G)})^{k_0}\,,\label{eq:kmp-conditional-mean-squared}
\end{align}
where the second inequality follows from \cref{eq:contraction-under-B-kmp-general-graph}.
By \cref{cor:two-particle-labeled-gap-compared-to-single-particle}
and the choice of \(k_0\),
\[
    \lambda_2(\lkmpchain^{(2,G)})^{k_0}
    \leq\exp\!\left(-\bigl(1-\lambda_2(\lkmpchain^{(2,G)})\bigr)k_0\right)
    \leq\exp\!\left(-\frac23\gamma k_0\right)
    \leq d^{-200/3}.
\]
Therefore, by Markov's inequality, we obtain
\begin{align*}
    \Pr\!\left[
        \|a^{(k_0)}-b^{(k_0)}\|_1>\frac{1}{3d^{19}}
    \right] 
    &\leq36d^{40}\sqrt{d+1}\cdot \lambda_2(\lkmpchain^{(2,G)})^{k_0}
    \leq\frac{1}{2d^2} \enspace.
\end{align*}
By Hoeffding's inequality, for $j\in[d]$,
\[
\begin{aligned}
    \Pr\!\left[
        \left|\frac{R_j^{(k_0)}}t-a_j^{(k_0)}\right|
        >\frac{1}{3d^{20}}
    \right]
    &=
    \Pr\!\left[
        |R_j^{(k_0)}-t a_j^{(k_0)}|
        >\frac{t}{3d^{20}}
    \right] \leq
    2\exp\!\left(-\frac{2t}{9d^{40}}\right)\leq \frac{1}{d^3} \enspace.
\end{aligned}
\]
If
\(\|R^{(k_0)}/t-a^{(k_0)}\|_1>\frac{1}{3d^{19}}\),
then at least one coordinate \(j\) satisfies
$ | R_j^{(k_0)}/t-a_j^{(k_0)} | > \frac{1}{3d^{20}}$.
A union bound over the \(d\) coordinates thus gives
\[
    \Pr\!\left[
        \norm{R^{(k_0)}/t-a^{(k_0)}}_1>\frac{1}{3d^{19}}
    \right]
    \leq\frac{1}{d^2} \enspace.
\]
The same bound holds for \(\widetilde R^{(k_0)}/t-b^{(k_0)}\).
The triangle inequality and a union bound then give
\begin{equation}
    \Pr\!\left[
        \|R^{(k_0)}-\widetilde R^{(k_0)}\|_1>td^{-19}
    \right]
    ~\leq~\frac{5}{2d^2} \enspace.
    \label{eq:kmp-large-particle-first-stage}
\end{equation}

\paragraph{The second stage.}
We now describe the second stage of the coupling.
In the next \(T\) steps, we sample \(T\) uniform edges \(e_1,\ldots,e_T\) of \(G\) independently.

Let \(\mathcal C\) be the event that the graph \(([d],\{e_1,\ldots,e_T\})\) is connected.
For \(e_1,\ldots,e_T\) that satisfy \(\mathcal C\),
we apply the coupling in \cref{lem:kmp-scheduled-coalescence} to the two copies of the KMP chain.
Combining \cref{lem:kmp-scheduled-coalescence} with \cref{eq:kmp-large-particle-first-stage}, we have
\[
    \Pr\!\left[
        R^{(k_\star)}\neq\widetilde R^{(k_\star)}
        \,\middle|\,\mathcal C
    \right]
    ~\leq~\frac{7}{2d^2} \enspace.
\]
For the case that \(\mathcal C\) does not hold, we apply independent updates
to the two copies of the KMP chain.
Consequently,
\[
    \Pr\!\left[
        R^{(k_\star)}\neq\widetilde R^{(k_\star)}
    \right]
    ~\leq~\Pr[\mathcal C^{\mathsf c}]
      ~+~
       \Pr\!\left[
           R^{(k_\star)}\neq\widetilde R^{(k_\star)}
           \,\middle|\,\mathcal C
       \right] ~\leq~  \Pr[\mathcal C^{\mathsf c}] ~+~ \frac{7}{2d^2} \enspace.
\]
Therefore, it remains to show that \(\Pr[\mathcal C^{\mathsf c}] \leq 1/(2d^2)\).

Fix a nonempty set \(S\subseteq V\) with \(|S|\leq d/2\),
and let \(\partial S\) denote the set of edges crossing the cut \((S,S^{\mathsf c})\).
Since each edge is sampled uniformly,
the probability that none of the \(T\) sampled edges crosses
this cut is $\left(1-\frac{|\partial S|}{|E|}\right)^T$.
Therefore, a union bound gives
    \begin{align*}
        \Pr[\mathcal C^{\mathsf c}]
        &\leq
        \sum_{S\subseteq V,\,1\leq|S|\leq d/2}
        \left(1-\frac{|\partial S|}{|E|}\right)^T.
    \end{align*}
For each edge \(e=\{a,b\}\in E\), let
$L_e \coloneqq (e_a-e_b)^{\intercal}(e_a-e_b)$ and $\glap_G = \sum_{e\in E} L_e$.
By Cheeger's inequality,
    we have
    \[
        \frac{|\partial S|}{|E|}
        \geq
        \frac{\lambda'_2(\glap_G)}{2|E|}\cdot|S|
        = 
        \left(1-\lambda_2(\kmpchain^{(1,G)})\right)|S| = \gamma \cdot|S|\enspace,
    \]
where \(\lambda'_2(\glap_G)\) is the second smallest eigenvalue of \(\glap_G\),
and the first equality follows from the fact that $ \kmpchain^{(1,G)} = I-\frac{1}{2|E|}\glap_G $
as in \cref{eq:chain-1-average}.

Since \(T\geq6\gamma^{-1}\log d\), we obtain
\begin{align*}
    \Pr[\mathcal C^{\mathsf c}]
    &\leq\sum_{k=1}^{\lfloor d/2\rfloor}
        \binom dk\,\e^{-\gamma kT}
    \leq\sum_{k=1}^{d}\binom dk\,d^{-6k}
    =(1+d^{-6})^d-1 \leq\e^{d^{-5}}-1
    \leq2d^{-5}
    \leq1/(2d^2)\enspace,
\end{align*}
where we used \(1+x\leq\e^x\),
\(\e^x-1\leq2x\) for \(0\leq x\leq1\), and \(d\geq3\).
This proves the lemma.
\end{proof}

We now prove the main theorem of this section.

\begin{proof}[Proof of \cref{thm:kmp-mixing-time-general-graph}]
    If \(d\leq 2\), the chain reaches stationarity after a single step.
    So we assume \(d\geq3\) and write 
    \[
        \delta(k)=\sup_{r\in\Delta_G}
        \|\kmpchain^k(r,\cdot)-\kmpdist\|_{\mathrm{TV}} \enspace.
    \]
If \(t\leq d^{1000}\), applying
\cref{lem:labeled-dominates-lifted-general-kmp},
\cref{thm:lkmp-mixing-time-general-graph}, and
\cref{cor:two-particle-labeled-gap-compared-to-single-particle},
we have
\[
    \delta(k)
    \leq 2t \cdot d^{5/4} \cdot (\lambda_2(\lkmpchain^{(2,G)}))^{k/2}
    \leq 2d^{1002} \cdot \exp\!\left(-\bigl(1-\lambda_2(\lkmpchain^{(2,G)})\bigr)\cdot k/2\right)
    \leq 2d^{1002}\cdot \e^{-\gamma k/3}\enspace.
\]
Therefore, taking $k = O(\gamma^{-1}(\log d+\log(1/\varepsilon)))$ makes the right-hand side at most \(\varepsilon\).

It remains to consider the case where \(t>d^{1000}\).
Let \(k_0=100\gamma^{-1}\log d\), \(T=6\gamma^{-1}\log d\) and \(k_\star=k_0+T\)
as in \cref{lem:kmp-large-particle-one-block}.
Applying \cref{fact:tv-block-decay} with
\(\ell=k_\star\) and \(\varepsilon=4/d^2\)
and using \cref{lem:kmp-large-particle-one-block},
we have
\[
    \delta(mk_\star)
    \leq\left(\frac{8}{d^2}\right)^m.
\]
Taking
\(
    m=
    \left\lceil
        \frac{\log(1/\varepsilon)}{\log(d^2/8)}
    \right\rceil
    =
    O\!\left(1+\frac{\log(1/\varepsilon)}{\log d}\right)
\)
makes the right-hand side at most \(\varepsilon\).
As \(k_\star=O(\gamma^{-1}\log d)\), we have the desired
$O(\gamma^{-1}(\log d+\log(1/\varepsilon)))$ bound on the mixing time.
\end{proof}

\section{The Continuous KMP Chain}
\label{sec:continuous-kmp}
In this section, we introduce the continuous KMP chain
on a general connected graph and
analyze its relationship with the discrete KMP chain
defined in \cref{def:kmp-general}.
Using the standard moment duality and the intertwining relation between the two chains
from \cref{sec:continuous-kmp-duality},
we derive an upper bound on the mixing time of the continuous chain
in \cref{sec:continuous-kmp-mixing}.
We finish with an application to
coordinate hit-and-run on the simplex in \cref{sec:CHAR}.

\subsection{Moment Duality and Intertwining Relation}
\label{sec:continuous-kmp-duality}

We begin with the definition of the continuous KMP chain on a general connected graph.

\begin{definition}[Continuous KMP chain]
    \label{def:continuous-kmp-general}
    Fix a positive integer \(d\), and let \(G=(V=[d],E)\) be a connected graph on \(d\) vertices.
    The continuous KMP chain on \(G\) is the Markov chain on
    \[
        \Delta_G^{\mathrm{cont}}
        \coloneqq
        \set{v\in \R_{\ge 0}^d : \sum_{x\in V} v_x = 1} \enspace.
    \]
    Given a current state \(v\in \Delta_G^{\mathrm{cont}}\), the next state is obtained by the following procedure:
    \begin{enumerate}
        \item sample a uniformly random edge \(e=\{i,j\}\in E\);
        \item sample \(\xi\sim\Unif[0,1]\), and update
        \[
            v_i \leftarrow (v_i+v_j)\xi
            \qquad\text{and}\qquad
            v_j \leftarrow (v_i+v_j)(1-\xi),
        \]
        while keeping all other coordinates unchanged.
    \end{enumerate}
    We denote its transition kernel by \(\kernel^{(G)}\).
The corresponding Markov operator \(\contchain^{(G)}\) acts on
bounded measurable functions
\(h:\Delta_G^{\mathrm{cont}}\to\mathbb R\) by
\[
    \bigl(\contchain^{(G)}h\bigr)(v)
    =
    \int_{\Delta_G^{\mathrm{cont}}}
        h(w)\,\kernel^{(G)}(v,\mathrm dw)
    \qquad
    \text{for every }v\in\Delta_G^{\mathrm{cont}} \enspace.
\]
 We omit the superscripts when they are clear from the context.
\end{definition}

We relate the discrete KMP chain \(\kmpchain^{(t,G)}\) and the continuous KMP chain \(\contchain^{(G)}\) through the classical KMP moment duality \cite{KMP82}.
Fix positive integers \(t\) and \(d\), and a connected graph \(G=(V=[d],E)\). Recall that
\begin{equation}\label{eq:DiscreteStateSpace}
    \Delta_G
    =
    \set{r\in\N^d:\sum_{x\in V} r_x=t}
\end{equation}
is the state space of the discrete KMP chain, and
\begin{equation}\label{eq:ContStateSpace}
     \Delta_G^{\mathrm{cont}}
    =
    \set{v\in \R_{\ge 0}^d:\sum_{x\in V} v_x=1}
\end{equation}
is the state space of the continuous KMP chain.
For \(r\in\Delta_G\) and \(v\in\Delta_G^{\mathrm{cont}}\), define
\[
    F(r,v)
    \coloneqq
    \prod_{i=1}^d \frac{v_i^{r_i}}{r_i!} \enspace.
\]
For each fixed \(r\), the function \(F(r,\cdot)\) is a polynomial of total degree \(t\). 
The following lemma gives a moment duality between the discrete and continuous KMP chains through the function \(F\).

\begin{lemma}
\label{lem:kmp-duality-general-graph}
For every \(r\in\Delta_G\) and \(v\in\Delta_G^{\mathrm{cont}}\),
\[
    (\contchain F(r,\cdot))(v)
    =
    (\kmpchain F(\cdot,v))(r) \enspace.
\]
\end{lemma}

\begin{proof}
    We first compute the left-hand side. Fix \(r\in\Delta_G\) and \(v\in\Delta_G^{\mathrm{cont}}\). By definition of the continuous KMP chain,
    \begin{align*}
        (\contchain F(r,\cdot))(v)
        =&~
        \frac{1}{|E|}
        \sum_{\{i,j\}\in E}
        \expect{\xi}{
            \frac{((v_i+v_j)\xi)^{r_i}\cdot ((v_i+v_j)(1-\xi))^{r_j}}{r_i!r_j!}
        }
        \cdot
        \prod_{x\notin\set{i,j}} \frac{v_x^{r_x}}{r_x!} \\
        =&~
        \frac{1}{|E|}
        \sum_{\{i,j\}\in E}
        \frac{(v_i+v_j)^{r_i+r_j}}{(r_i+r_j+1)!}
        \cdot
        \prod_{x\notin\set{i,j}} \frac{v_x^{r_x}}{r_x!} \enspace,
    \end{align*}
    where the second equality follows from the Beta integral identity
    \[
        \int_0^1 x^a(1-x)^b\,\mathrm{d}x
        =
        \frac{a!\,b!}{(a+b+1)!}
        \qquad \forall a,b\in\N \enspace.
    \]
    On the other hand, by definition of the discrete KMP chain,
    \begin{align*}
        (\kmpchain F(\cdot,v))(r)
        =&~
        \frac{1}{|E|}
        \sum_{\{i,j\}\in E}
        \frac{1}{r_i+r_j+1}
        \sum_{a=0}^{r_i+r_j}
        \frac{v_i^a v_j^{r_i+r_j-a}}{a!(r_i+r_j-a)!}
        \cdot
        \prod_{x\notin\set{i,j}} \frac{v_x^{r_x}}{r_x!} \\
        =&~
        \frac{1}{|E|}
        \sum_{\{i,j\}\in E}
        \frac{1}{(r_i+r_j+1)!}
        \sum_{a=0}^{r_i+r_j}
        \binom{r_i+r_j}{a}
        v_i^a v_j^{r_i+r_j-a}
        \cdot
        \prod_{x\notin\set{i,j}} \frac{v_x^{r_x}}{r_x!} \\
        =&~
        \frac{1}{|E|}
        \sum_{\{i,j\}\in E}
        \frac{(v_i+v_j)^{r_i+r_j}}{(r_i+r_j+1)!}
        \cdot
        \prod_{x\notin\set{i,j}} \frac{v_x^{r_x}}{r_x!} \enspace,
    \end{align*}
    where the last equality follows from the binomial theorem. This proves the lemma.
\end{proof}

We next use the moment duality to obtain an intertwining relation
between the transition operators of the two chains.
To formulate this relation, we introduce the low-degree polynomial spaces
on which the continuous operator acts.
Let \(\mu_G=\Dirichlet(1,\ldots,1)\) denote the uniform
stationary distribution of the continuous KMP chain
on \(\Delta_G^{\mathrm{cont}}\).
For each integer \(k\ge 0\), define
\begin{equation}
\label{eq:def-Pk-continuous-kmp}
    \PP_k
    \coloneqq
    \set{
        H:\Delta_G^{\mathrm{cont}}\to\R :
        H \text{ is a polynomial of total degree at most } k
    } \enspace.
\end{equation}
Since the update rule of the continuous KMP chain preserves polynomial degree, each space \(\PP_k\) is invariant under \(\contchain^{(G)}\). We therefore write
\begin{equation}
\label{eq:def-gap-k-continuous-kmp}
    \mathrm{gap}_k\!\br{\contchain^{(G)}}
    \coloneqq
    \text{the spectral gap of }
    \contchain^{(G)}\big|_{\PP_k\cap L^2(\mu_G)}
\end{equation}
for the spectral gap of the restriction of \(\contchain^{(G)}\) to
\(\PP_k\cap L^2(\mu_G)\).
We record the following comparison between the low-degree
restricted spectral gaps, which will be used in \cref{sec:lambda-2}.
It follows from \cite[Theorem 1.1]{kim2025spectral}
by letting all $\alpha_x = 1$.

\begin{lemma}
\label{lem:continuous-gap-comparison-kim}
    For a graph \(G = (V = [d],E)\) and any integer \(k\geq 1\), we have
    \[
        \frac{2(d+1)}{3d} \cdot \mathrm{gap}_1\!\br{\contchain^{(G)}}
        \leq \mathrm{gap}_k\!\br{\contchain^{(G)}} \enspace.
    \]
\end{lemma}

The following lemma turns the moment duality in
\cref{lem:kmp-duality-general-graph} into an intertwining relation
between the discrete and continuous transition operators.
This relation will allow us to transfer mixing bounds
from the discrete chain to the continuous chain
in \cref{sec:continuous-kmp-mixing}.

\begin{lemma}\label{lem:isomorphism}
For each $t\ge 1$, define the linear map
\(
T_t : \mathbb{R}^{\Delta_G} \to \PP_t
\)
by
\[
(T_t g)(v)
:=
\sum_{r\in \Delta_G} g(r)\, F(r,v) \enspace,
\qquad
F(r,v):=\prod_{i=1}^d \frac{v_i^{r_i}}{r_i!}\enspace.
\]
Then $T_t$ is a linear isomorphism, and it intertwines the discrete and continuous KMP
operators:
\[
N^{(G)}|_{\PP_t} T_t = T_t \kmpchain^{(t,G)}\enspace.
\]
Consequently, $N^{(G)}|_{\PP_t}$ and $\kmpchain^{(t,G)}$ have the same eigenvalues.
\end{lemma}

\begin{proof}
The intertwining relation follows directly from \cref{lem:kmp-duality-general-graph}. Indeed, for every $g\in \mathbb{R}^{\Delta_G}$
and every $v\in \Delta_G^{\mathrm{cont}}$,
\[
(N^{(G)} T_t g)(v)
=
\sum_{r\in \Delta_G} g(r)\, (N^{(G)}F(r,\cdot))(v)
=
\sum_{r\in \Delta_G} g(r)\, (\kmpchain^{(t,G)} F(\cdot,v))(r)
=
(T_t \kmpchain^{(t,G)} g)(v)\enspace.
\]

It remains to show that $T_t$ is an isomorphism onto $\PP_t$. First, note that
\[
|\Delta_G| = \binom{t+d-1}{d-1}\enspace.
\]
On the other hand, $\PP_t$ has the same dimension, since its standard monomial basis is indexed by
multi-indices $\alpha\in \N^d$ with $\sum_i \alpha_i \le t$, whose number is also
$\binom{t+d-1}{d-1}$.

We next show that the family $\{F(r,\cdot): r\in \Delta_G\}$ spans $\PP_t$. Let
\[
m_\alpha(v):=\prod_{i=1}^d v_i^{\alpha_i}
\]
be a monomial of total degree $|\alpha|:=\sum_i \alpha_i \le t$. Since $\sum_{i=1}^d v_i = 1$ on
$\Delta_G^{\mathrm{cont}}$, we have
\[
m_\alpha(v)
=
m_\alpha(v)\Bigl(\sum_{i=1}^d v_i\Bigr)^{t-|\alpha|}
\]
as functions on $\Delta_G^{\mathrm{cont}}$. Expanding the right-hand side by the multinomial
theorem shows that $m_\alpha$ lies in the span of $\{F(r,\cdot): r\in \Delta_G\}$. Hence
$\{F(r,\cdot): r\in \Delta_G\}$ spans $\PP_t$.

Since the number of these functions equals $\dim \PP_t$, they form a basis of $\PP_t$.
So $T_t$ is a linear isomorphism.
The spectral equality follows immediately from the intertwining relation.
\end{proof}

\subsection{Mixing Time of the Continuous KMP Chain on a General Graph}
\label{sec:continuous-kmp-mixing}
We now use the intertwining relation to transfer the particle-number-independent mixing bound in
\cref{thm:kmp-mixing-time-general-graph} for the discrete KMP chain to the continuous KMP chain.
\begin{theorem}
\label{lem:continuous-kmp-mixing}
Let \(G=(V,E)\) be a connected graph on \(d\geq2\)
vertices.
Let \(\kernel^{(G)}\) be the transition kernel of
the continuous KMP chain
and let \(\mu_G\) be its uniform stationary distribution.
There is an absolute constant \(C>0\) such that,
for every \(\varepsilon\in(0,1)\), the total-variation mixing time is at most
\[
        \frac{C}{
            \gamma_{\scriptscriptstyle\mathsf{KMP}}^{(1,G)}
        }
        \bigl(\log d+\log(1/\varepsilon)\bigr),
\]
where
\(\gamma_{\scriptscriptstyle\mathsf{KMP}}^{(1,G)}\)
is the spectral gap of the single-particle KMP chain
on \(G\).
\end{theorem}

\begin{proof}
We will approximate functions on the continuous state
space by functions obtained from randomly placing
\(t\) particles, and
then let \(t\to\infty\).

Let \(\pi_t\) be the uniform stationary distribution on $\Delta_G$ defined in \cref{eq:DiscreteStateSpace}.
By \cref{thm:kmp-mixing-time-general-graph}, we can choose
\(
    k=\left\lceil
        C/
            \gamma_{\scriptscriptstyle\mathsf{KMP}}^{(1,G)}
        \bigl(\log d+\log(1/\varepsilon)\bigr)
    \right\rceil
\)
with a sufficiently large absolute constant \(C\),
such that for every positive integer \(t\)
and every \(g:\Delta_G\to[0,1]\),
\begin{equation}\label{eq:discreteclose}
     \left|
        \left(
            \bigl(\kmpchain^{(t,G)}\bigr)^k g
        \right)\!(r)-\E_{R\sim\pi_t}[g(R)]
    \right|
    \leq\varepsilon
    \qquad\text{for every }r\in\Delta_G \enspace.
\end{equation}

Our goal is to obtain the corresponding estimate
for every continuous
\(f:\Delta_G^{\mathrm{cont}}\to[0,1]\).
Fix such a function \(f\), and define
\(g_t:\Delta_G\to[0,1]\) by \(g_t(r)=f(r/t)\).
To relate this function to the continuous chain,
fix \(v\in\Delta_G^{\mathrm{cont}}\) and place
\(t\) particles independently on the vertices,
each at vertex \(i\) with probability \(v_i\).
Let \(R=(R_i)_{i\in V}\) be the resulting particle counts.
Then
\[
    \Pr[R=r]
    =
    \frac{t!}{\prod_{i\in V}r_i!}
    \prod_{i\in V}v_i^{r_i} \enspace.
\]
Define \(f_t:\Delta_G^{\mathrm{cont}}\to[0,1]\) by
\[
    f_t(v)
    =\mathbb E[f(R/t)]
    =
    \sum_{r\in\Delta_G}
        g_t(r)\frac{t!}{\prod_{i\in V}r_i!}
        \prod_{i\in V}v_i^{r_i} \enspace.
\]
As \(R/t\) approaches \(v\),
we expect \(f_t(v)\) to approach \(f(v)\).
We will verify this below.

First, we show that the discrete mixing estimate
gives the desired continuous estimate for \(f_t\).
For a real-valued function \(g\) on \(\Delta_G\),
the map \(T_t\) from \cref{lem:isomorphism}
produces the polynomial function
\[
    (T_tg)(v)
    =
    \sum_{r\in\Delta_G}
        g(r)\prod_{i\in V}\frac{v_i^{r_i}}{r_i!},
    \qquad v\in\Delta_G^{\mathrm{cont}} \enspace.
\]
In particular, \(f_t=t!T_tg_t\).
The intertwining identity in that lemma gives
\[
    \contchain^{(G)}T_t
    =T_t\kmpchain^{(t,G)},
\]
and hence, by induction,
\[
    \bigl(\contchain^{(G)}\bigr)^k f_t
    =t!T_t\bigl(\kmpchain^{(t,G)}\bigr)^k g_t \enspace.
\]
Thus
\[
    \left(
        \bigl(\contchain^{(G)}\bigr)^k f_t
    \right)(v)
    =
    \sum_{r\in\Delta_G}
        \frac{t!}{\prod_{i\in V}r_i!}
        \prod_{i\in V}v_i^{r_i}
        \left(
            \bigl(\kmpchain^{(t,G)}\bigr)^k g_t
        \right)(r) \enspace.
\]
This identity relates two expectations:
the left-hand side evolves the continuous chain
from \(v\) and evaluates \(f_t\);
the right-hand side first samples the particle
counts from \(v\), then evolves the discrete chain
and evaluates \(g_t\).
Since
\[
    \sum_{r\in\Delta_G}
        \frac{t!}{\prod_{i\in V}r_i!}
        \prod_{i\in V}v_i^{r_i}
    =
    \left(\sum_{i\in V}v_i\right)^t
    =1 \enspace,
\]
applying \cref{eq:discreteclose} to each term,
we obtain
\[
    \left|
        \left(
            \bigl(\contchain^{(G)}\bigr)^k f_t
        \right)(v)-\E_{R\sim\pi_t}[g_t(R)]
    \right|
    \leq\varepsilon \enspace.
\]
Recall that
\(\mu_G=\Dirichlet(1,\ldots,1)\).
For every \(r\in\Delta_G\),
\[
    \int
        \prod_{i\in V}v_i^{r_i}\,\mathrm d\mu_G(v)
    =
    \frac{(d-1)!\prod_{i\in V}r_i!}{(t+d-1)!} \enspace.
\]
Since
\[
    |\Delta_G|=\binom{t+d-1}{d-1} \enspace,
\]
it follows that
\[
    \int
        \frac{t!}{\prod_{i\in V}r_i!}
        \prod_{i\in V}v_i^{r_i}\,\mathrm d\mu_G(v)
    =
    \frac{t!(d-1)!}{(t+d-1)!}
    =
    \frac1{|\Delta_G|} \enspace.
\]
Substituting into the definition of \(f_t\), we get
\[
    \int f_t\,\mathrm d\mu_G
    =
    \frac1{|\Delta_G|}
    \sum_{r\in\Delta_G}g_t(r)
    =
    \E_{R\sim\pi_t}[g_t(R)] \enspace.
\]
Therefore,
\[
    \left|
        \left(
            \bigl(\contchain^{(G)}\bigr)^k f_t
        \right)(v)
        -\int f_t\,\mathrm d\mu_G
    \right|
    \leq\varepsilon \enspace.
\]

It remains to replace \(f_t\) by \(f\).
We prove that \(f_t\to f\) uniformly on
\(\Delta_G^{\mathrm{cont}}\).
For the particle counts sampled from \(v\), we have
\[
    \mathbb E[\|R/t-v\|_2^2]
    =
    \sum_{i\in V}\frac{v_i(1-v_i)}{t}
    \leq\frac1t \enspace.
\]
Since \(f\) is continuous on the compact simplex,
it is uniformly continuous.
Given \(\delta>0\), choose \(a>0\) such that
\[
    |f(x)-f(y)|\leq\delta
    \qquad\text{whenever }\|x-y\|_2\leq a \enspace.
\]
Using \(0\leq f\leq1\) and Markov's inequality,
we obtain
\[
\begin{aligned}
    |f_t(v)-f(v)|
    &\leq\mathbb E|f(R/t)-f(v)|
    \leq\delta+\Pr[\|R/t-v\|_2>a]
    \leq\delta+\frac1{ta^2} \enspace.
\end{aligned}
\]
The bound is independent of \(v\), so
\[
    \|f_t-f\|_\infty
    :=
    \sup_{v\in\Delta_G^{\mathrm{cont}}}
        |f_t(v)-f(v)|
    \longrightarrow0 \enspace.
\]
A Markov operator does not increase the supremum norm.
Consequently,
\[
\begin{aligned}
    \left|
        \left(
            \bigl(\contchain^{(G)}\bigr)^k f
        \right)(v)
        -\int f\,\mathrm d\mu_G
    \right|
    &\leq
    \left|
        \left(
            \bigl(\contchain^{(G)}\bigr)^k f_t
        \right)(v)
        -\int f_t\,\mathrm d\mu_G
    \right|
    +2\|f_t-f\|_\infty\\
    &\leq\varepsilon+2\|f_t-f\|_\infty \enspace.
\end{aligned}
\]
The choice of \(k\) is independent of \(t\).
Letting \(t\to\infty\) therefore gives
\[
    \left|
        \left(
            \bigl(\contchain^{(G)}\bigr)^k f
        \right)(v)
        -\int f\,\mathrm d\mu_G
    \right|
    \leq\varepsilon \enspace.
\]
By \cref{fact:tv-continuous-functions},
total variation distance can be computed by taking
the supremum over continuous functions with values
in \([0,1]\).
Using
\[
    \left(
        \bigl(\contchain^{(G)}\bigr)^k f
    \right)(v)
    =
    \int f(w)\,
        \bigl(\kernel^{(G)}\bigr)^k(v,\mathrm dw) \enspace,
\]
we conclude that
\[
    \sup_{v\in\Delta_G^{\mathrm{cont}}}
    \left\|
        \bigl(\kernel^{(G)}\bigr)^k(v,\cdot)-\mu_G
    \right\|_{\mathrm{TV}}
    \leq\varepsilon \enspace. \qedhere
\]
\end{proof}

\subsection{Coordinate Hit-and-Run on the Simplex}
\label{sec:CHAR}

Kook and Vempala~\cite{KV26} established spectral-gap
and mixing-time bounds for coordinate hit-and-run
on general convex bodies.
As an application of
\cref{lem:continuous-kmp-mixing}, we obtain
a sharper total variation mixing-time bound for
the standard simplex
\[
    S_n=\left\{x\in\mathbb R_{\geq0}^n:
        \sum_{i=1}^n x_i\leq1\right\},
\]
with the proof given below.
At each step, coordinate hit-and-run chooses
$i\in[n]$ uniformly and replaces $x_i$ by a uniform
sample from
\(
    [0,1-\sum_{j\neq i}x_j]
\),
leaving all other coordinates unchanged.

\begin{corollary}
\label{lem:char-simplex-mixing}
For every integer $n\geq2$ and every
$\varepsilon\in(0,1)$, the total-variation mixing time for coordinate hit-and-run on \(S_n\) is at most
\(
    O\!\left(
        n\bigl(\log n+\log(1/\varepsilon)\bigr)
    \right).
\)
\end{corollary}

\begin{proof}
For \(x\in S_n\), introduce the additional coordinate
\(
    x_0=1-\sum_{i=1}^n x_i.
\)
When coordinate \(i\) is selected, the update can be
written as
\[
    x_i'=U(x_0+x_i) \enspace,
    \qquad
    x_0'=(1-U)(x_0+x_i) \enspace,
\]
where \(U\) is uniform on \([0,1]\).
Thus, coordinate hit-and-run on \(S_n\) is equivalent
to the continuous KMP chain on the star graph
\(G=K_{1,n}\), with center \(0\) and leaves
\(1,\ldots,n\).
Each step selects one of its \(n\) edges uniformly.

The single-particle KMP chain on \(G\) has
transition matrix
\[
    \kmpchain^{(1,G)}
    =
    \begin{pmatrix}
        \frac12 & \frac{1}{2n} & \frac{1}{2n}
            & \cdots & \frac{1}{2n}\\
        \frac{1}{2n} & 1-\frac{1}{2n} & 0
            & \cdots & 0\\
        \frac{1}{2n} & 0 & 1-\frac{1}{2n}
            & \cdots & 0\\
        \vdots & \vdots & \vdots & \ddots & \vdots\\
        \frac{1}{2n} & 0 & 0
            & \cdots & 1-\frac{1}{2n}
    \end{pmatrix} \enspace.
\]
Its eigenvalues are \(1\),
\(1-\frac{1}{2n}\) with multiplicity \(n-1\), and
\(\frac{n-1}{2n}\).
Since \(n\geq2\), its spectral gap is
\(
    \gamma_{\scriptscriptstyle\mathsf{KMP}}^{(1,G)}
    =\frac1{2n} .
\)
Applying \cref{lem:continuous-kmp-mixing}
with \(d=n+1\) gives a mixing-time bound of
\[
    O\!\left(
        n\bigl(\log n+\log(1/\varepsilon)\bigr)
    \right)
\]
for the continuous KMP chain on \(G\).
The bijection
\[
    (x_1,\ldots,x_n)
    \longmapsto
    \left(1-\sum_{i=1}^n x_i,x_1,\ldots,x_n\right)
\]
maps the uniform distribution on \(S_n\) to
\(\mu_G\) and preserves total variation distance.
Hence the same mixing-time bound holds for
coordinate hit-and-run on \(S_n\).
\end{proof}

\begin{remark}
For comparison, applying
\cite[Corollary~1.2]{KV26} to the standard simplex
gives a total variation mixing-time bound of
\(
    O\!\left(
        n^3\log n\log\frac{\chi^2(\mu_0\|\pi_n)}{\varepsilon}
    \right)
\)
for an initial distribution \(\mu_0\), where \(\pi_n\)
denotes the uniform distribution on \(S_n\).
For this simplex, \cref{lem:char-simplex-mixing} improves the bound to
\(
    O\!\left(n\bigl(\log n+\log(1/\varepsilon)\bigr)\right)
\)
and applies uniformly over all initial points.
\end{remark}

\section{Parallel Kac's Walk on the Symmetric Subspace}
\label{sec:scalable-pru}
This section now connects the KMP chain to \emph{Kac's walk}~\cite{Kac56}.
We show that the twirling channel induced by the parallel
Kac's walk~\cite{LQSY+24}
converges to the Haar average on the symmetric
subspace at a rate governed by a parallel KMP chain,
bringing the mixing techniques developed above into the quantum setting.

We begin by recalling the \emph{parallel Kac’s walk}
on \(\ugroup{d}\) introduced in~\cite{LQSY+24}.
Its one-step distribution \(\kacdist\) is given by the random unitary
\(K=P_\tau W\), where \(P_\tau\) is a random permutation unitary
and \(W\) is a random block-diagonal unitary:
\begin{align}\label{eq:kac-random-unitary}
    K \coloneqq P_\tau \cdot W \enspace,
\end{align}
where
\(
    W = \sum_{i=1}^{d/2} \ketbra{i}\otimes W_i
\)
is a block-diagonal unitary, whose blocks $W_i$ are independent Haar random unitaries in $\ugroup{2}$, and
\(
    P_\tau = \sum_{x=1}^{d} \ketbratwo{\tau(x)}{x}
\)
is the permutation unitary induced by a uniformly random permutation $\tau \in \S_d$.
Thus, each step consists of first applying independent Haar-random unitaries within each two-dimensional block, and then applying a uniformly random permutation of the computational basis states.

For each positive integer $t$, we define $\kacchannel{t}$ to be the
$t$-fold twirling channel associated with the distribution $\kacdist$.
That is, $\kacchannel{t}$ describes the average action of the random unitary
$K^{\otimes t}$ on an input operator $\sigma$:
\[
    \kacchannel{t}(\sigma)
    \coloneqq
    \expect{K\gets\kacdist}{K^{\otimes t}\cdot \sigma \cdot K^{\otimes t,\dagger}} \enspace.
\]
It will be convenient to separately consider the two sources of randomness
appearing in the definition of $K$, namely the block-diagonal Haar random unitary $W$
and the random permutation $P_\tau$.
To this end, we introduce the corresponding $t$-fold twirling channels
\begin{equation}\label{eq:rand-perm-block-haar-channels}
    \blockhaarchannel{t}(\sigma) \coloneqq \expect{W_1,\cdots,W_{d/2}}{W^{\otimes t} \cdot \sigma \cdot W^{\otimes t,\dagger}} \enspace, \enspace\enspace
        \randpermchannel{t}(\sigma) \coloneqq \expect{\tau}{P_\tau^{\otimes t} \cdot \sigma \cdot P_\tau^{\otimes t,\dagger}} \enspace.
\end{equation}
Since $K = P_\tau W$, the channel $\kacchannel{t}$ can therefore be written as
\[
    \kacchannel{t} = \randpermchannel{t} \circ \blockhaarchannel{t} \enspace.
\]
In other words, one first averages over the block-diagonal Haar randomness
and then over the random permutation.

Finally, because our construction is obtained by iterating this random transformation,
we write $\repkacchannel{t}{k}$ for the $k$-time repetition of $\kacchannel{t}$:
\begin{equation}\label{eq:kac-k-fold-channel}
    \repkacchannel{t}{k}(\sigma) \coloneqq
    \underbrace{\kacchannel{t}\circ\cdots\circ\kacchannel{t}}_{k\text{ times}}(\sigma) \enspace.
\end{equation}
Equivalently, $\repkacchannel{t}{k}$ represents the average action of applying
$k$ independent steps of the parallel Kac's walk on the input.

Our goal is to show that, after sufficiently many independent repetitions,
the channel $\repkacchannel{t}{k}$ becomes close to the Haar twirling channel
$\haarchannel{t}$ on the symmetric subspace.
The channel $\haarchannel{t}$ is the $t$-fold twirling channel associated
with the Haar distribution $\haardist$ over $\ugroup{d}$:
\begin{equation}\label{eq:haar-twirl-channel}
    \haarchannel{t}(\sigma) = \expect{U\gets\haardist}{U^{\otimes t} \cdot \sigma \cdot U^{\otimes t,\dagger}} \enspace.
\end{equation}
The following theorem gives our main quantitative bound.

\begin{restatable}{theorem}{maintheorem}\label{thm:main}
There is an absolute constant \(c>0\) such that the following holds.
Let $n,t$ be positive integers, let $k\in\mathbb N$, and let $d=2^n$.
Let $\repkacchannel{t}{k}$ and $\haarchannel{t}$ be the quantum channels
defined in \cref{eq:kac-k-fold-channel} and \cref{eq:haar-twirl-channel}, respectively.
Let $\symsubspaceproj$ be the projection onto the symmetric subspace $\symsubspace{t}{d}$.
For any state $\sigma$ such that $\tr{(\symsubspaceproj\otimes\id)\cdot\sigma} = 1$,
we have
\[
    \norm{(\repkacchannel{t}{k+1}\otimes \id)(\sigma)
    -(\haarchannel{t}\otimes \id)(\sigma)}_1
    \leq2d^2\e^{-ck}.
\]
\end{restatable}

By choosing an appropriate $k$, we immediately obtain the following corollary.

\begin{corollary}\label{thm:twirling}
Under the assumptions of \cref{thm:main},
for any \(\varepsilon\in(0,1)\),
the following holds whenever \(k\ge c^{-1}\log(2d^2/\varepsilon)\):

 \[
    \norm{(\repkacchannel{t}{k+1}\otimes \id)(\sigma) - (\haarchannel{t}\otimes \id)(\sigma)}_1 \leq \varepsilon
\]
where \(c\) is the constant in \cref{thm:main}.
Thus, \(O(\log d+\log(1/\varepsilon))\) repetitions suffice
uniformly over \(t\).
\end{corollary}

To prove \cref{thm:main}, we first reduce the action of
\(\kacchannel{t}\) on the symmetric subspace
to a classical Markov chain on partitions of \(t\)
in \cref{sec:classical-interpretation}.
In \cref{sec:convergence}, we lift this chain to the parallel KMP chain
and state the resulting mixing bound.
We then combine these ingredients in \cref{sec:proof-of-main} to prove \cref{thm:main}.

\subsection{From Parallel Kac's Walk to Partition Resampling Chain}
\label{sec:classical-interpretation}
We first show in \cref{sec:one-step} that, on the symmetric subspace, the output of the parallel Kac channel \(\kacchannel{t}\)
can always be written as a linear combination of the states \(\{\rho_\lambda\}_{\lambda\vdash t}\).
This leads to a transition matrix indexed by partitions,
which is explained in \cref{sec:matrix-interpretation}.
In \cref{sec:markov-chain-interpretation}, we show that this matrix
is exactly the transition matrix of a classical Markov chain on $\partition{t}{d}$,
which we call the \emph{partition resampling chain}.
This viewpoint allows us to study the convergence of $\repkacchannel{t}{k}$
via the mixing time of the corresponding classical chain,
as stated in \cref{sec:channel-mixing}.

\subsubsection{Decomposition of Symmetric States under \texorpdfstring{$\kacchannel{t}$}{Kac Twirling Channel}}
\label{sec:one-step}

Since we are interested in the action of $\kacchannel{t}$
on the symmetric subspace, we begin by analyzing its two components,
namely $\blockhaarchannel{t}$ and $\randpermchannel{t}$
in \cref{eq:rand-perm-block-haar-channels},
when restricted to this subspace.

We first study $\randpermchannel{t}$ on generalized Dicke states.
The key point is that it forgets the specific multiplicity vector
and depends only on the type.
As a result, all generalized Dicke states of the same type
are mixed into the maximally mixed state on the corresponding type subspace.

\begin{lemma}\label{lem:rand-perm-channel}
    For any multiplicity vector $r \in \mathbb{N}^d$ with type
    $\lambda \in \partition{t}{d}$, we have
    \[
        \randpermchannel{t}(\ketbra{D_r}) = \rho_{\lambda} \enspace,	
    \]
    where $\rho_{\lambda}$ is defined in \cref{eq:projector-and-state-lambda}.
\end{lemma}

\begin{proof}
    Suppose that the multiplicity vector \(r\) of \(\ket{D_r}\)
    has type \(\lambda = (\lambda_1,\dots,\lambda_s)\).
    We first describe the action of \(P_\tau^{\otimes t}\)
    on generalized Dicke states.
    Since \(P_\tau\) permutes the computational basis,
    applying \(P_\tau^{\otimes t}\) to a basis vector
    \(\ket{x_1,\ldots,x_t}\) replaces each symbol \(x_i\) by \(\tau(x_i)\).
    Thus, the number of occurrences of the symbol \(a\) now becomes the number of occurrences of \(\tau(a)\).
    In terms of multiplicity vectors,
    this corresponds to the natural action of \(\S_d\) on \(\mathbb{N}^d\) defined by
    \(
        (\tau\cdot r)_j := r_{\tau^{-1}(j)} .
    \)
    Thus, for every permutation \(\tau\in \S_d\),
    \(
        P_\tau^{\otimes t}\ket{D_r}=\ket{D_{\tau\cdot r}}.
    \)
    Applying $\randpermchannel{t}$ to \(\ketbra{D_r}\) hence gives
    \[
    \randpermchannel{t}(\ketbra{D_r})
    =
    \frac{1}{d!}\sum_{\tau\in \S_d}\ketbra{D_{\tau\cdot r}} .
    \]

    Now observe that the orbit of \(r\) under the action of \(\S_d\) consists
    precisely of all multiplicity vectors whose type is \(\lambda\), namely
    \(
    \{\tau\cdot r : \tau\in \S_d\} = \Xi_{\lambda}
    \),
    where $\Xi_{\lambda}$ is defined in \cref{eq:xi-lambda}.
    Denote the stabilizer of \(r\) by
    \(
    \stab(r) := \{\tau\in \S_d : \tau\cdot r = r\} .
    \)
    By the orbit-stabilizer theorem,
    \[
        |\Xi_{\lambda}| = \frac{d!}{|\stab(r)|} \enspace.
    \]
    Moreover, for each \(s\in\Xi_{\lambda}\), the number of permutations
    \(\tau\in \S_d\) such that \(\tau\cdot r = s\) is exactly
    \(|\stab(r)|\). Hence we can regroup the above sum as
    \[
        \randpermchannel{t}(\ketbra{D_r})
        =
        \frac{|\stab(r)|}{d!}
        \sum_{s\in\Xi_{\lambda}} \ketbra{D_s}
        =
        \frac{1}{|\Xi_{\lambda}|}
        \sum_{s\in\Xi_{\lambda}} \ketbra{D_s}
        =
        \rho_{\lambda} \enspace.
    \]
    This completes the proof.
\end{proof}

We next turn to the action of $\blockhaarchannel{t}$ and show that
$\blockhaarchannel{t}$ kills all off-diagonal terms
between generalized Dicke states.
Intuitively, this happens because independent random phases
are introduced on each basis inside every block,
and averaging over these phases eliminates
every cross term between distinct multiplicity vectors.

\begin{lemma}\label{lem:block-haar-kills-off-diagonal}
    For any state $\sigma$ in the symmetric subspace $\symsubspace{t}{d}$,
    the state $\blockhaarchannel{t}(\sigma)$ is diagonal
    in the generalized Dicke basis $\{\ket{D_r}\}_r$.
    In particular, it admits the decomposition
    \[
        \blockhaarchannel{t}(\sigma)
        =
        \sum_r c_\sigma(r)\cdot\ketbra{D_r},
    \]
    for some coefficients $c_\sigma(r)$ depending on $\sigma$.
\end{lemma}

\begin{proof}
    Recall the definition of $\blockhaarchannel{t}$:
    \[
    	\blockhaarchannel{t}(\sigma) = \expect{W_1,\cdots,W_{d/2}}{W^{\otimes t} \cdot \sigma \cdot W^{\otimes t,\dagger}} \enspace,
    \]
    where $W$ is a block-diagonal unitary with independent Haar random blocks \(W_i\in\ugroup{2}\).
    First, note that the distribution of $W$ is identical to
    that of $F \cdot W$,
    where $F$ is the random phase unitary
    \[
        F=\sum_{l=1}^{d/2}\ketbra{l}\otimes
        \begin{bmatrix}
            \e^{\i\theta_{2l-1}} & 0\\
            0 & \e^{\i\theta_{2l}}
        \end{bmatrix}
        \enspace,
    \]
    and $\theta_1,\dots,\theta_d$ are independent random variables
    uniformly distributed over $[0, 2\pi)$.
    Therefore,
    \begin{align*}
        \blockhaarchannel{t}(\sigma)
        = \expect{W}{W^{\otimes t} \cdot \sigma \cdot W^{\otimes t,\dagger}}
        &= \expect{F,W}{(F\cdot W)^{\otimes t} \cdot \sigma \cdot (F\cdot W)^{\otimes t,\dagger}} \\
        &= \expect{F}{F^{\otimes t} \cdot \blockhaarchannel{t}(\sigma) \cdot F^{\otimes t,\dagger}} \enspace.
    \end{align*}  
    Since $\blockhaarchannel{t}(\sigma)$ remains
    supported on the symmetric subspace,
    it can be written as a linear combination of
    generalized Dicke states:
    \[
        \blockhaarchannel{t}(\sigma) =
        \sum_{r,r'} c_\sigma(r,r') \cdot \ketbratwo{D_r}{D_{r'}} \enspace.
    \]
    Let \(\ket{D_r}\) and \(\ket{D_{r'}}\) be two distinct generalized Dicke states appearing in the above decomposition with multiplicity vectors
    \(r=(r_1,\ldots,r_d)\) and \(r'=(r'_1,\ldots,r'_d)\), respectively.
    Since \(r\neq r'\),
    there exists some \(j\in[d]\) such that \(r_j\neq r'_j\).
    When \(F^{\otimes t}\) acts on \(\ket{D_r}\), each occurrence of
    \(\ket{j}\) contributes a phase factor \(\e^{\i\theta_{j}}\).
    Thus, $F^{\otimes t}$ contributes a phase factor of
    $\e^{\i \theta_{j} (r_j - r'_j)}$ to
    $\ketbratwo{D_r}{D_{r'}}$.
    Since $\theta_{j}$ is uniformly distributed over $[0, 2\pi)$,
    the expectation of this phase factor is zero, and thus
    \[
        \expect{F}{F^{\otimes t} \cdot \ketbratwo{D_r}{D_{r'}} \cdot F^{\otimes t,\dagger}} = 0 \enspace.
    \]
    Hence all off-diagonal terms vanish after averaging over $F$, and
    we conclude that
    \[
        \blockhaarchannel{t}(\sigma) =
        \sum_r c_\sigma(r) \cdot \ketbra{D_r} \enspace,
    \]
    where \(c_\sigma(r) := c_\sigma(r,r)\).
\end{proof}

Combining the previous two lemmas,
we obtain a structural description of the action of $\kacchannel{t}$
on the symmetric subspace.
Namely, after $\blockhaarchannel{t}$ diagonalizes
the state in the generalized Dicke basis,
$\randpermchannel{t}$ merges all basis states of the same type
into the corresponding maximally mixed state $\rho_\lambda$.
Therefore, the output state of $\kacchannel{t}$
on a symmetric state can be expressed as a linear combination
of the maximally mixed states $\st{\rho_\lambda}_{\lambda\vdash t}$.

\begin{proposition}\label{prop:decomp-after-kac-channel}
    For any state \(\sigma\) in the symmetric subspace $\symsubspace{t}{d}$,
    we have
    \[
        \kacchannel{t}(\sigma) =
        \sum_{\lambda\in\partition{t}{d}} p_\sigma(\lambda) \cdot \rho_\lambda \enspace,
    \]
    for some coefficients \(p_\sigma(\lambda)\) depending on $\sigma$, where $\st{\rho_\lambda}_{\lambda\vdash t}$ is defined in \cref{eq:projector-and-state-lambda}.
\end{proposition}

\begin{proof}
    Applying \(\kacchannel{t} = \randpermchannel{t} \circ \blockhaarchannel{t}\) to \(\sigma\) gives
    \[
        \kacchannel{t}(\sigma) =
        \randpermchannel{t}\left(\blockhaarchannel{t}(\sigma)\right) \enspace.
    \]
    By \cref{lem:block-haar-kills-off-diagonal},
    \[
        \blockhaarchannel{t}(\sigma) =
        \sum_{r} c_\sigma(r)\cdot \ketbra{D_r} \enspace,
    \]
    for some coefficients \(c_\sigma(r)\).
    Applying \cref{lem:rand-perm-channel} to each term \(\ketbra{D_r}\) gives
    \[
        \randpermchannel{t}\left(\blockhaarchannel{t}(\sigma)\right) =
        \sum_{r} c_\sigma(r) \cdot \rho_{r^{\darrow}} \enspace.
    \]
    Finally, regrouping the above sum according to the types of the multiplicity vectors gives
    \[     
        \kacchannel{t}(\sigma) =
        \sum_{\lambda \in\partition{t}{d} } p_\sigma(\lambda) \cdot \rho_\lambda \enspace,
    \]
    where \(p_\sigma(\lambda) := \sum_{r : r^{\darrow} = \lambda} c_\sigma(r)\).  
\end{proof}

\subsubsection{Transition Matrix Interpretation}
\label{sec:matrix-interpretation}
From \cref{prop:decomp-after-kac-channel}, we know that
any output state of $\kacchannel{t}$ can be written as
a linear combination of states $\{\rho_\lambda\}_{\lambda \vdash t}$
defined in \cref{eq:projector-and-state-lambda}.
Therefore, the action of $\kacchannel{t}$ can be viewed as
updating the coefficients in this decomposition.
This observation naturally leads to a transition matrix indexed by partitions.

\begin{definition}[Transition matrix corresponding to $\kacchannel{t}$]
\label{def:transition-matrix-kac-channel}
We define the transition matrix $M^{(t,d)} \in \mathbb{R}^{\partition{t}{d} \times \partition{t}{d}}$ corresponding to $\kacchannel{t}$
to be the matrix satisfying
\[
\kacchannel{t}(\rho_\lambda)
=
\sum_{\mu \in\partition{t}{d}}
M^{(t,d)}(\lambda,\mu) \cdot \rho_\mu
\qquad \forall\ \lambda\in\partition{t}{d} \enspace.
\]
We often omit the superscripts and write $M$
for the transition matrix when they are clear from the context.
\end{definition}

To make it explicit, we provide a formula for the
entries of $M$.

\begin{lemma}\label{lem:element-in-transition-matrix}
The entries of the transition matrix $M^{(t,d)}$ are given by
\[
	M^{(t,d)}(\lambda,\mu) = \tr{ \Pi_\mu \cdot \blockhaarchannel{t}(\rho_\lambda) }
    \qquad \forall\ \lambda,\mu \in\partition{t}{d} \enspace,
\]
where $\Pi_\mu$ is defined in \cref{eq:projector-and-state-lambda}.
Moreover,
$M$ is a stochastic matrix:
    all entries are non-negative, and each row sums to $1$.
\end{lemma}

\begin{proof}
    From \cref{def:transition-matrix-kac-channel} we have
    \[
    \kacchannel{t}(\rho_\lambda)
    =
    \sum_{\mu' \in\partition{t}{d}} M(\lambda,\mu')\,\rho_{\mu'} \enspace.
    \]
    Taking the trace with $\Pi_\mu$ and using
    \cref{fact:orthogonality}, we obtain
    \[
    	\tr{ \Pi_\mu \cdot \kacchannel{t}(\rho_\lambda) }
        = \sum_{\mu' \in\partition{t}{d}} M(\lambda,\mu') \cdot \tr{ \Pi_\mu \cdot \rho_{\mu'} } = M(\lambda,\mu)\enspace.
    \]
    On the other hand,
    \begin{align*}
        \tr{ \Pi_\mu \cdot \kacchannel{t}(\rho_\lambda) } 
        = \tr{ \Pi_\mu \cdot \randpermchannel{t}(\blockhaarchannel{t}(\rho_\lambda) ) }
        &=  \tr{ \randpermchannel{t} ( \Pi_\mu ) \cdot
                \blockhaarchannel{t}(\rho_\lambda) } \\
        &= \tr{ \Pi_\mu \cdot \blockhaarchannel{t}(\rho_\lambda) } \enspace,
    \end{align*}
    where the last step follows from \cref{lem:rand-perm-channel}, since
    \[
    	\randpermchannel{t} ( \Pi_\mu )
        = \sum_{r\in \Xi_{\mu}} \randpermchannel{t} (\ketbra{D_r})
        = |\Xi_\mu| \cdot \rho_\mu
        = \Pi_\mu
        \enspace.
    \]
    Combining the two expressions proves the formula.

    Finally, the formula for the matrix entries implies that
    $M^{(t,d)}(\lambda,\mu) \geq 0$,
    since both $\Pi_\mu$ and
    $\blockhaarchannel{t}(\rho_\lambda)$ are positive semidefinite.
    Moreover,
    \[
        \sum_{\mu \in \partition{t}{d}}
        M^{(t,d)}(\lambda,\mu)
        =
        \tr{
            \left(
                \sum_{\mu \in \partition{t}{d}} \Pi_\mu
            \right)
            \cdot \blockhaarchannel{t}(\rho_\lambda)
        }
        =
        \tr{
            \symsubspaceproj \cdot \blockhaarchannel{t}(\rho_\lambda)
        }
        =
        1
        \enspace ,
    \]
    where we used
    $\sum_{\mu \in \partition{t}{d}} \Pi_\mu = \symsubspaceproj$
    on the symmetric subspace.
    Hence $M^{(t,d)}$ is stochastic.
\end{proof}

Having defined the transition matrix $M$
corresponding to $\kacchannel{t}$,
we can now view the action of $\kacchannel{t}$
as applying $M$ to the coefficients in the decomposition of the input state.

\begin{lemma}\label{lem:transition-matrix-description-kac-channel}
    Let $\rho = \sum_{\lambda\in\partition{t}{d}} p(\lambda) \cdot \rho_\lambda$
    be any input state in the symmetric subspace $\symsubspace{t}{d}$.
    Then the state $\kacchannel{t}(\rho)$ can be written as
    \[
        \kacchannel{t}(\rho) = \sum_{\mu \in\partition{t}{d}} q(\mu)\cdot \rho_\mu \enspace,
    \]
    where the new coefficients $q(\mu)$ are given by
    \(
        q(\mu)
        =
        \sum_{\lambda\in\partition{t}{d}} p(\lambda)\cdot M^{(t,d)}(\lambda,\mu) .
    \)
\end{lemma}

\begin{proof}
    This follows directly from the definition:
    \begin{align*}
        \kacchannel{t}(\rho) =
        \sum_{\lambda\in\partition{t}{d}} p(\lambda) \cdot \kacchannel{t}(\rho_\lambda) 
        &=
        \sum_{\lambda\in\partition{t}{d}} p(\lambda) \cdot \left( \sum_{\mu \in\partition{t}{d}} M(\lambda,\mu) \cdot \rho_\mu \right) \\
        &=
        \sum_{\mu \in\partition{t}{d}} \br{ \sum_{\lambda \in\partition{t}{d}} p(\lambda)\cdot M^{(t,d)}(\lambda,\mu)} \cdot \rho_\mu \enspace. \qedhere
    \end{align*}
\end{proof}

\subsubsection{Markov Chain Interpretation}
\label{sec:markov-chain-interpretation}

The matrix $M$ defined in \cref{def:transition-matrix-kac-channel}
can be regarded as the transition matrix of
a classical Markov chain on the set $\partition{t}{d}$ containing all partitions of $t$ with length at most $d$.
At a high level, the chain lifts a partition to a multiplicity vector of that type,
randomly mixes its coordinates in pairs while preserving pairwise sums,
and then projects back to a partition by taking the type.

\begin{definition}[Partition resampling chain]
\label{def:partition-resampling-chain}
The partition resampling chain is a Markov chain
on the state space $\partition{t}{d}$.
The chain updates a current partition
$\lambda = (\lambda_1,\cdots,\lambda_s)$ in the following manner:
\begin{enumerate}
    \item Let $ R \in \N^d$ be the vector obtained by concatenating $\lambda$ with $d-s$ zeros if $s<d$, i.e.,
    \[
    	R = (\lambda_1,\dots,\lambda_s, \underbrace{0,\dots,0}_{d-s \text { zeros}} ) \enspace.
    \]
    Otherwise, if $s=d$, we simply let $R=\lambda$.
    \item Sample a uniformly random perfect matching of \([d]\), denoted by
        \[
            \mathcal{M} = \{\{i_1,j_1\},\dots,\{i_m,j_m\}\}, \qquad m = d/2\enspace.
        \]
    \item For each pair $\{i,j\} \in \match$, let $m_{i,j} = R_i + R_j$
    and independently and uniformly resample $(R_i,R_j)$ from the set
    \[
        \{ (x,y) \in \N^2 : x+y = m_{i,j} \}	\enspace.
    \]
    Let $R'$ denote the resulting vector.
    \item The updated partition $\lambda'$ is the type of $R'$, namely $\lambda' = {R'}^{\darrow}$.
\end{enumerate}
\end{definition}

\begin{proposition}\label{prop:markov-chain-transition-matrix-correspondence}
    The transition matrix of the above Markov chain is $M^{(t,d)}$.
\end{proposition} 

\begin{proof}
    Let $b = d/2$.
    We consider an equivalent version of the above chain.
    For a current partition $\lambda\vdash t$
    with length at most $d$,
    we can equivalently describe the update process as follows:
    \begin{enumerate}[
        labelsep=0.5em,
        topsep=4pt,
        itemsep=4pt,
        parsep=0pt
    ]
        \item Sample a vector $R \in \Xi_{\lambda,d}$ uniformly at random,
        where $\Xi_{\lambda,d}$ is defined in \cref{eq:xi-lambda}.
        \item Generate a new vector as follows: for each $j\in[b]$,
        let $m_j(R) = R_{2j-1}+R_{2j}$ and independently and uniformly 
        resample $(R_{2j-1},R_{2j})$ from the set
        \(
        \{ (x,y) \in \N^2 : x+y = m_j(R) \}	.
        \)
        Let $R'$ denote the new vector.
        \item Set the updated partition $\lambda'$ to be the type of $R'$,
        namely $\lambda' = {R'}^{\darrow}$.
    \end{enumerate}
    The two descriptions are equivalent,
    since randomly matching fixed coordinates and then resampling is
    distributionally identical to first uniformly shuffling the coordinates,
    followed by resampling on fixed pairs.
    The transition probability from $\lambda$ to $\lambda'$ is
    \begin{align*}
        \Pr\Br{\lambda\to\lambda'}
        =~& \frac{1}{\abs{\Xi_{\lambda}}} \sum_{r\in\Xi_{\lambda}} \Pr[ R'\in\Xi_{\lambda'} | R = r] 
        =~ \frac{1}{\abs{\Xi_{\lambda}}} \sum_{r\in\Xi_{\lambda}} \sum_{r'\in\Xi_{\lambda'}} \Pr[ R' = r' | R = r] \enspace.
    \end{align*}
    Note that from \cref{lem:element-in-transition-matrix}
    and \cref{eq:projector-and-state-lambda}, we have
    \[
    	M(\lambda,\lambda')
        = \tr{ \Pi_{\lambda'} \cdot \blockhaarchannel{t}(\rho_\lambda) }
        = \frac{1}{|\Xi_\lambda|} \sum_{r\in\Xi_\lambda} \sum_{r'\in\Xi_{\lambda'}}
        \tr{ \ketbra{D_{r'}} \cdot \blockhaarchannel{t}\br{\ketbra{D_r}} } \enspace.
    \]
    Therefore, it suffices to show that for any $r\in\Xi_\lambda$ and $r'\in\Xi_{\lambda'}$,
    \[
        \Pr[ R' = r' | R = r ] =
        \tr{ \ketbra{D_{r'}} \cdot \blockhaarchannel{t}\br{\ketbra{D_r}} } \enspace.
    \]

    Notice that the probability $\Pr[ R' = r' | R = r ]$ is nonzero
    only if $m_j(r) = m_j(r')$ for all $j\in[b]$. In this case, we have
    \[
        \Pr[ R' = r' | R = r ] = \prod_{j=1}^b \frac{1}{m_j(r)+1} \enspace.
    \]
    On the other hand, by the definition of $\blockhaarchannel{t}$
    in \cref{eq:rand-perm-block-haar-channels}, we have
    \begin{align*}
        (\star)\coloneqq
        \tr{ \ketbra{D_{r'}} \cdot \blockhaarchannel{t}\br{\ketbra{D_r}} }
        =~& 
        \tr{ 
            \ketbra{D_{r'}} \cdot
            \expect{W}{ W^{\otimes t} \ketbra{D_r} W^{\otimes t,\dagger} }
         }\\
         =~& \expect{W}{ \abs{\bra{D_{r'}} W^{\otimes t} \ket{D_r}}^2  } \enspace.
    \end{align*}
    Since $W$ is block-diagonal with blocks indexed by $j\in[b]$,
    its matrix element $W_{y,x}$ is zero
    whenever $x$ and $y$ are not in the same block.
    Consequently, if $m_j(r) \neq m_j(r')$ for some $j\in[b]$,
    then $\bra{D_{r'}} W^{\otimes t} \ket{D_r}=0$ for all such $W$,
    and hence the above expectation equals $0$.
    
    Now suppose that $m_j(r) = m_j(r')$ for all $j\in[b]$.    
    Let $m_j = m_j(r) = m_j(r')$.
    For $j\in[b]$ and some set $I_j\subseteq [t]$ with size $m_j = r_{2j-1}+r_{2j}$,
    define the \emph{local} Dicke state as
    \[
    	\ket{L_{r_{2j-1},r_{2j}}}_{I_j} \coloneq \frac{1}{\sqrt{\binom{m_j}{r_{2j}}}}
        \sum_{\substack{S\subseteq I_j \\ |S|=r_{2j}}}
        \ket{2j-1}^{\otimes r_{2j-1}}_{I_j\setminus S} \otimes \ket{2j}^{\otimes r_{2j}}_{S} \enspace.
    \]
    The state $\ket{L_{r_{2j-1},r_{2j}}}_{I_j}$
    is the uniform superposition
    over all ways of placing $\ket{2j}$ on exactly
    $r_{2j}$ sites of $I_j$ and $\ket{2j-1}$ on the
    remaining $r_{2j-1}$ sites.
    Notice that we can rewrite $\ket{D_r}$ as
    \[
    	\ket{D_r } = \sqrt{\frac{\prod_{j=1}^{b} m_j! }{t!}}
        \sum_{
            \substack{(I_1,\dots,I_b):\\I_1\sqcup\cdots\sqcup I_b=[t] \\
            \forall j\in[b],\ |I_j|=m_j}
        }
        \bigotimes_{j=1}^b \ket{L_{r_{2j-1},r_{2j}}}_{I_j} \enspace,
    \]
    where $\sqcup$ denotes the disjoint union.
    The intuition is that choosing a basis state
    of multiplicity vector $r$ is equivalent to
    first choosing, for each $j\in[b]$, the set $I_j$ of 
    sites occupied by the pair $\{2j-1,2j\}$,
    and then, within each $I_j$,
    choosing which $r_{2j}$ sites contain $2j$
    and which $r_{2j-1}$ sites contain $2j-1$.
    Therefore, we have
    \begin{align*}
        (\star)
        =~ & \br{\frac{\prod_{j=1}^{b} m_j!}{t!}}^2
        \expect{W}{
            \abs{
                \sum_{\substack{(I_1,\ldots,I_b)\\(I'_1,\ldots,I'_b)}}
                \br{\bigotimes_{j=1}^b \bra{L_{r'_{2j-1},r'_{2j}}}_{I'_j}}
                W^{\otimes t}
                \br{\bigotimes_{j=1}^b\ket{L_{r_{2j-1},r_{2j}}}_{I_j}}
            }^2 
        } \enspace.
    \end{align*}
    It is evident that the expectation is nonzero
    only if $I_j = I'_j$ for all $j\in[b]$,
    since $W$ is block-diagonal.
    In this case, we have
    \begin{align*}
        (\star)
        =~ &\br{\frac{\prod_{j=1}^{b} m_j!}{t!}}^2
        \underbrace{\expect{W}{
            \abs{
                \sum_{(I_1,\ldots,I_b)} \prod_{j=1}^b
                \bra{L_{r'_{2j-1},r'_{2j}}}_{I_j} W^{\otimes m_j} \ket{L_{r_{2j-1},r_{2j}}}_{I_j}
            }^2
        }}_{(\triangle)} \enspace.
    \end{align*}
If we identify $\ket{2j-1}$ and $\ket{2j}$ with
    the qubit basis states $\ket{0}$ and $\ket{1}$,
    then $\ket{L_{r_{2j-1},r_{2j}}}_{I_j}$ can be viewed as
    the following standard Dicke state $\ket{D^{(2)}_{m_j, r_{2j}}}$ of Hamming weight $r_{2j}$ on
    $m_j$ qubits:
    \[
    	\ket{D^{(2)}_{m_j, r_{2j}}} \coloneq \frac{1}{\sqrt{\binom{m_j}{r_{2j}}}}
        \sum_{\substack{x\in\{0,1\}^{m_j} \\ |x|=r_{2j}}} \ket{x}
    \]
    supported on the sites belonging to $I_j$.
    Recall that \(
    W = \sum_{j=1}^{b} \ketbra{j}\otimes W_j
\)
is a block-diagonal unitary, where $W_j$ are independent Haar random unitaries in $\ugroup{2}$. For each $j\in[b]$, using standard Dicke states we can write
\[\bra{L_{r'_{2j-1},r'_{2j}}}_{I_j} W^{\otimes m_j} \ket{L_{r_{2j-1},r_{2j}}}_{I_j}=\bra{D^{(2)}_{m_j, r'_{2j}}} W_j^{\otimes m_j} \ket{D^{(2)}_{m_j, r_{2j}}}\enspace.\]
    Then we expand $(\triangle)$ as
    \begin{align*}
            (\triangle) =~ & \sum_{\substack{(I_1,\ldots,I_b) \\ (I'_1,\ldots,I'_b)}}
        \expect{W}{
            \prod_{j=1}^b
            \bra{L_{r'_{2j-1},r'_{2j}}}_{I_j} W^{\otimes m_j} \ket{L_{r_{2j-1},r_{2j}}}_{I_j}
            \cdot
            \overline{\bra{L_{r'_{2j-1},r'_{2j}}}_{I'_j} W^{\otimes m_j} \ket{L_{r_{2j-1},r_{2j}}}_{I'_j}}
        } \\
        =~ & \sum_{\substack{(I_1,\ldots,I_b) \\ (I'_1,\ldots,I'_b)}}
        \prod_{j=1}^b
        \expect{W_j\sim\ugroup{2}}{
            \bra{D^{(2)}_{m_j, r'_{2j}}} W_j^{\otimes m_j} \ket{D^{(2)}_{m_j, r_{2j}}}
            \cdot
            \overline{\bra{D^{(2)}_{m_j, r'_{2j}}} W_j^{\otimes m_j} \ket{D^{(2)}_{m_j, r_{2j}}}}
        } \\
        =~& \sum_{\substack{(I_1,\ldots,I_b) \\ (I'_1,\ldots,I'_b)}}
        \prod_{j=1}^b
        \Tr \br{ \ketbra{D^{(2)}_{m_j, r'_{2j}}} \cdot \expect{W_j\sim\ugroup{2}}{ W_j^{\otimes m_j} \ketbra{D^{(2)}_{m_j, r_{2j}}} W_j^{\otimes m_j,\dagger} } } \\
        =~& \sum_{\substack{(I_1,\ldots,I_b) \\ (I'_1,\ldots,I'_b)}}
        \prod_{j=1}^b
        \Tr \br{ \ketbra{D^{(2)}_{m_j, r'_{2j}}} \cdot\rho_{\sym}^{m_j,2}}\enspace,
    \end{align*}
    where in the second equality $W_j\sim \ugroup{2}$ means
    that $W_j$ is drawn according to
    the Haar measure on $\ugroup{2}$ independently, and the last equality follows from \cref{eq:haar-averaging-symmetric-subspace}.
    Since
    \(
    	\{  \ket{D^{(2)}_{m_j, l}}: l\in\{0,1,\dots,m_j\}\} 
    \)
    is an orthonormal basis in $\vee^{m_j}\mathbb C^2$, we have
    \[
    \rho_{\sym}^{m_j,2}=
    \frac{1}{m_j+1}\sum_{l=0}^{m_j} \ketbra{D^{(2)}_{m_j, l}} \enspace.
    \]
    Thus, 
    \begin{align*}
        (\triangle)=~ & \sum_{\substack{(I_1,\ldots,I_b) \\ (I'_1,\ldots,I'_b)}}
        \prod_{j=1}^b \frac{1}{m_j+1} 
        = \br{\frac{t!}{\prod_{j=1}^b m_j!}}^2 \prod_{j=1}^b \frac{1}{m_j+1} \enspace.
    \end{align*}
    Therefore, we conclude that
    \[
    	(\star) = \prod_{j=1}^b \frac{1}{m_j+1} \enspace,
    \]
    which matches $\Pr[ R' = r' | R = r ]$.
\end{proof}

We summarize some basic properties of the partition resampling chain below.

\begin{lemma}\label{lem:property-partition-resampling-chain}
    The partition resampling chain is irreducible, aperiodic and reversible, with a unique stationary distribution $\pi$ given by
    \[
    \pi(\lambda) = \binom{t + d - 1}{d - 1}^{-1} \cdot |\Xi_{\lambda,d}| \qquad \forall\ \lambda \in \partition{t}{d} \enspace.
    \]
\end{lemma}

\begin{proof}
    Starting from any \(\lambda\in\partition{t}{d}\),
    there exists a finite sequence of updates,
    each occurring with positive probability,
    that moves all mass to a single coordinate
    and hence reaches the partition \((t)\).
    Reversing this procedure,
    one sees that \((t)\) can reach any
    \(\lambda'\in\partition{t}{d}\) with positive probability.
    Hence the chain is irreducible.
    Moreover, every state has a positive self-loop probability
    (by resampling each matched pair to its original values),
    so the chain is aperiodic.

    To prove reversibility, it suffices to verify the detailed balance condition
    \[
    \pi(\lambda)\cdot M(\lambda,\lambda')
    =\pi(\lambda') \cdot M(\lambda',\lambda)
    \qquad \forall \lambda,\lambda'\in\partition{t}{d} \enspace.
    \]
    By the equivalent description of the chain in the proof of \cref{prop:markov-chain-transition-matrix-correspondence},
    \[
    M(\lambda,\lambda')
    =\frac{1}{|\Xi_\lambda|}
    \sum_{r\in\Xi_\lambda}\sum_{r'\in\Xi_{\lambda'}}
    \Pr[R'=r'\mid R=r] \enspace.
    \]
    Hence,
    \[
    \pi(\lambda)\cdot M(\lambda,\lambda')
    =
    \binom{t+d-1}{d-1}^{-1}
    \sum_{r\in\Xi_\lambda}\sum_{r'\in\Xi_{\lambda'}}
    \Pr[R'=r'\mid R=r] \enspace.
    \]
    Now fix \(r,r'\in\mathbb N^d\), and write \(b=d/2\).
    If there exists \(j\in[b]\) such that
    \(m_j(r)\neq m_j(r')\), where \(m_j(r)=r_{2j-1}+r_{2j}\), then
    \[
    \Pr[R'=r'\mid R=r]=\Pr[R'=r\mid R=r']=0 \enspace.
    \]
    Otherwise \(m_j(r)=m_j(r')\) for all \(j\in[b]\), and therefore
    \[
    \Pr[R'=r'\mid R=r]
    =
    \prod_{j=1}^b \frac{1}{m_j(r)+1}
    =
    \prod_{j=1}^b \frac{1}{m_j(r')+1}
    =
    \Pr[R'=r\mid R=r'] \enspace.
    \]
    Thus
    \[
    \pi(\lambda)\cdot M(\lambda,\lambda')
    =
    \binom{t+d-1}{d-1}^{-1}
    \sum_{r'\in\Xi_{\lambda'}}\sum_{r\in\Xi_\lambda}
    \Pr[R'=r\mid R=r']
    =
    \pi(\lambda')\cdot M(\lambda',\lambda) \enspace,
    \]
    which proves that the chain is reversible w.r.t. \(\pi\).
    Finally, $\pi$ is a probability distribution since
    \[
        \sum_{\lambda\in\partition{t}{d}} |\Xi_\lambda|
        =
        \binom{t+d-1}{d-1} \enspace.
    \]
    By the detailed balance condition, $\pi$ is a stationary distribution of the chain. Since the chain is irreducible and aperiodic on a finite state space, the stationary distribution is unique. 
\end{proof}

\subsubsection{From Chain Mixing to Channel Mixing}
\label{sec:channel-mixing}

We now explain how the mixing time of the partition resampling chain
governs the convergence of $\repkacchannel{t}{k}$
to the maximally mixed state on the symmetric subspace.

\begin{proposition}\label{prop:spectral-gap-implies-channel-mixing}
Let $M$ be the transition matrix of the partition resampling chain
defined in \cref{def:partition-resampling-chain},
and let $\pi$ be its stationary distribution given in \cref{lem:property-partition-resampling-chain}.
Suppose that for some integer $k\in\N$, we have that
there exists $\epsilon > 0$ such that 
for any $\lambda\in\partition{t}{d}$,
\(
    \norm{M^k(\lambda,\cdot)-\pi}_{\mathrm{TV}}
    \le
    \epsilon.
\)
Then for any input state $\sigma$ in the symmetric subspace $\symsubspace{t}{d}$,
\[
    \norm{\repkacchannel{t}{k+1}(\sigma) - \rho_{\sym}}_1 \le 2\epsilon \enspace,
\]
where
\(
\rho_{\sym}
\)
is the maximally mixed state on the symmetric subspace $\vee^t \mathbb{C}^d$.
\end{proposition}

\begin{proof}
    Let $\sigma$ be any state in  $\symsubspace{t}{d}$.
    By \cref{prop:decomp-after-kac-channel}, we can write $\kacchannel{t}(\sigma)$ as
    \[
        \kacchannel{t}(\sigma) =
        \sum_{\lambda\in\partition{t}{d}} p_\sigma(\lambda) \cdot \rho_\lambda \enspace.
    \]
    Applying $\kacchannel{t}$ for $k$ times, by
    \cref{lem:transition-matrix-description-kac-channel},
    we have
    \[
        \repkacchannel{t}{k+1}(\sigma) =
        \sum_{\mu\in\partition{t}{d}} q_\sigma(\mu) \cdot \rho_\mu \enspace,
    \]
    where
    \(
        q_\sigma(\mu) = \sum_{\lambda\in\partition{t}{d}} p_\sigma(\lambda) \cdot M^k(\lambda,\mu).
    \)
    On the other hand, by \cref{fact:decomposition-maximally-mixed-state},
    $\rho_\sym$ can be written as
    \[
        \rho_{\sym} = \sum_{\mu\in\partition{t}{d}} \pi(\mu) \cdot \rho_\mu \enspace.
    \]
    Therefore,
    \begin{align*}
        \norm{\repkacchannel{t}{k+1}(\sigma) - \rho_{\sym}}_1
        =~ 2\cdot\norm{q_\sigma - \pi}_{\mathrm{TV}}
        =~ 2\cdot\norm{p_\sigma \cdot M^{k}  - \pi}_{\mathrm{TV}}
        \le
        2\epsilon \enspace,
    \end{align*}
    where the inequality follows from the assumed mixing bound
and the triangle inequality.
\end{proof} 
 
\subsection{From Partition Resampling Chain to Parallel KMP Chain}
\label{sec:convergence}
The reduction established in the preceding subsection allows us to
focus on the mixing time of the partition resampling chain.
Our key observation is that this chain admits a natural lifting
to a parallel version of the KMP chain~\cite{KMP82},
which records the full multiplicity vector rather than only its type.
We call this lifted chain the \emph{parallel KMP chain}
and define it below.

\begin{definition}[Parallel KMP chain]
\label{def:discrete-parallel-kmp}
    The parallel KMP chain is a Markov chain on $\Delta$ defined as
    \[
    	\Delta = \st{ R \in \N^d : \sum_{j=1}^d R_j = t } \enspace.
    \]
    Its update rule follows the same procedure as the partition resampling chain, except that we operate directly on the vector $R$ rather than its induced partition. More precisely, the chain evolves as follows:
    \begin{enumerate}
        \item Sample a uniformly random perfect matching of \([d]\), denoted by
        \[
            \mathcal{M} = \{\{i_1,j_1\},\dots,\{i_m,j_m\}\}, \qquad m = d/2\enspace.
        \]
        \item For each pair $\{i,j\} \in \match$, let $m_{i,j} = R_i + R_j$
        and independently resample $(R_i, R_j)$ uniformly from the set
        \(
            \{ (x,y) \in \N^2 : x+y = m_{i,j} \}.
        \)
        Let $R'$ denote the resulting vector.
        \item Set the next state of the chain to be $R'$.
    \end{enumerate}
    We denote the transition matrix of this lifted chain by $\liftedchain^{(t,d)}$. We may omit the superscripts when they are clear from the context.
\end{definition}

This Markov chain can be viewed as a parallel version of the KMP chain on the complete graph:
each step performs independent KMP updates simultaneously along all edges
of a uniformly random perfect matching.
Similarly to \cref{lem:discrete-kmp-stationary}, we have the following
properties of the parallel KMP chain.
We omit the proof as it follows the same argument.

\begin{lemma}
    \label{lem:discrete-parallel-kmp-stationary}
    For every positive integer $t$ and even integer $d\geq2$,
    the transition matrix $\liftedchain^{(t,d)}$ is symmetric.
    In particular, the parallel chain is irreducible, aperiodic and reversible with respect to
    the uniform distribution on $\Delta$, denoted by \( \lifteddist \),
    which is its unique stationary distribution.
\end{lemma}

Recall that the partition resampling chain 
in \cref{sec:markov-chain-interpretation} keeps track only of the sorted type,
whereas the parallel KMP chain records the full multiplicity vector in \(\Delta\).
Thus, the former is obtained from the latter by lumping all vectors
that share the same type.
Thus, the partition resampling chain mixes
at least as fast as the parallel KMP chain.
Adapting the techniques from \cref{sec:mixing-kmp} yields
the following mixing bound for the parallel KMP chain,
and hence for the partition resampling chain.
We formalize their relationship and prove the theorem
in \cref{sec:parallel-kmp}.

\begin{restatable}{theorem}{mixingtimeboundpkmp}
\label{thm:partition-resampling-mixing}
For some absolute constant \(c>0\), all even \(d\ge2\), and all integers \(t,k\ge1\),
\begin{align*}
    \sup_{r\in\Delta}
    \|\liftedchain^k(r,\cdot)-\lifteddist\|_{\mathrm{TV}}
    ~\leq~ d^2\cdot \e^{-ck} \quad \text{and} \quad
    \sup_{\lambda\in\partition{t}{d}}
    \|M^k(\lambda,\cdot)-\pi\|_{\mathrm{TV}}
    ~\leq~ d^2\cdot\e^{-ck} \enspace.
\end{align*}
Here, \(\liftedchain\) is the transition matrix of the parallel
KMP chain on \(\Delta\) defined in
\cref{def:discrete-parallel-kmp}, and \(M\) is the transition
matrix of the partition resampling chain on \(\partition{t}{d}\)
defined in \cref{def:partition-resampling-chain}.
Their stationary distributions \(\lifteddist\) and \(\pi\)
are given in \cref{lem:discrete-parallel-kmp-stationary}
and \cref{lem:property-partition-resampling-chain}, respectively.
\end{restatable}

Using the above theorem, we immediately obtain the following corollary on the mixing times of the parallel KMP chain and the partition resampling chain.
\begin{corollary}
\label{cor:patition-resampling-chain-mixing-time}
For every \(\varepsilon\in(0,1)\), the mixing times of the parallel KMP chain
and the partition resampling chain to total variation distance \(\varepsilon\)
are both
$O\!\left(\log d+\log(1/\varepsilon)\right)$
for all positive integers \(t\).
\end{corollary}

\subsection{Convergence of the Parallel Kac Channel}
\label{sec:proof-of-main}
In this subsection, we prove \cref{thm:main}.
We restate it here for the reader's convenience.

\maintheorem*

\begin{proof}
    Fix a state $\sigma$ such that $\tr{(\symsubspaceproj\otimes\id)\cdot\sigma} = 1$.
    We can decompose the state as
    \[
        \sigma =
        \sum_{ r, r' } \ketbratwo{D_r}{D_{r'}} \otimes \sigma_{r,r'} \enspace,
    \]
    where $\ket{D_r}$ is the generalized Dicke state with multiplicity vector $r$,
    and
    \[
    	\sigma_{r,r'} = (\bra{D_r}\otimes \id)\,\sigma\,(\ket{D_{r'}}\otimes \id)
    \]
    are operators on the corresponding system.
    Applying \(\kacchannel{t}\otimes \id\) to \(\sigma\), we obtain
    \[
        (\kacchannel{t}\otimes \id)(\sigma) =
        \sum_{ r, r' } \kacchannel{t}(\ketbratwo{D_r}{D_{r'}}) \otimes \sigma_{r,r'}\enspace.
    \]
    Recall the random phase unitary $F$ in the proof of \cref{lem:block-haar-kills-off-diagonal}.
    As observed in the proof of \cref{lem:block-haar-kills-off-diagonal},
    the channel $\blockhaarchannel{t}$ is invariant if we insert $F$ before or after the channel, and for distinct multiplicity vectors $r$ and $r'$,
    \[
        \blockhaarchannel{t}(\ketbratwo{D_r}{D_{r'}}) = \blockhaarchannel{t}\Br{\E_F[F^{\otimes t}\cdot \ketbratwo{D_r}{D_{r'}}\cdot F^{\otimes t,\dagger}]} = 0\enspace.
    \]
    Therefore, the only nonzero terms in the above sum are those with $r = r'$:
    \begin{equation*}
        (\kacchannel{t}\otimes \id)(\sigma) =
        \sum_{ r } \kacchannel{t}(\ketbratwo{D_r}{D_{r}}) \otimes \sigma_{r,r} \enspace.
    \end{equation*}
    Hence, after applying $\kacchannel{t}$ for another $k$ times,
    \begin{equation}\label{eq:decomp-after-kac-channel-anc}
        (\repkacchannel{t}{k+1}\otimes \id)(\sigma) =
        \sum_{ r } \repkacchannel{t}{k+1}(\ketbratwo{D_r}{D_{r}}) \otimes \sigma_{r,r} \enspace.
    \end{equation}
    As for the Haar twirling channel $\haarchannel{t}$,
    we have that by \cref{eq:haar-averaging-symmetric-subspace},
    \begin{equation}\label{eq:decomp-after-haar-channel-anc}
        (\haarchannel{t}\otimes \id)(\sigma) =
        \sum_{ r } \haarchannel{t}(\ketbratwo{D_r}{D_{r}}) \otimes \sigma_{r,r}
        = \sum_{ r } \rho_\sym \otimes \sigma_{r,r} \enspace,
    \end{equation}
    where $\rho_\sym$ is the maximally mixed state on the symmetric subspace $\symsubspace{t}{d}$.

    Combining \cref{eq:decomp-after-kac-channel-anc} and
    \cref{eq:decomp-after-haar-channel-anc}, we have
    \begin{align}\label{eq:final-bound-anc}
        \norm{(\repkacchannel{t}{k+1}\otimes \id)(\sigma) - (\haarchannel{t}\otimes \id)(\sigma)}_1
        &=~
        \norm{\sum_{ r } \left(\repkacchannel{t}{k+1}(\ketbratwo{D_r}{D_{r}}) - \rho_\sym\right) \otimes \sigma_{r,r}}_1 \nonumber \\
        &\leq~
        \sum_{ r } \norm{\repkacchannel{t}{k+1}(\ketbratwo{D_r}{D_{r}}) - \rho_\sym}_1 \cdot \norm{\sigma_{r,r}}_1 \enspace.
    \end{align}
    By \cref{prop:spectral-gap-implies-channel-mixing} and
    \cref{thm:partition-resampling-mixing}, we have
    \begin{equation} \label{eq:final-bound1-anc}
        \norm{\repkacchannel{t}{k+1}(\ketbratwo{D_r}{D_{r}}) - \rho_\sym}_1 \leq
        2d^2\e^{-ck} \enspace.
    \end{equation}
    Moreover, note that
    \begin{equation}\label{eq:final-bound2-anc}
         \sum_{ r } \norm{\sigma_{r,r}}_1 = \sum_{ r } \tr{\sigma_{r,r}} = \tr{(\symsubspaceproj\otimes \id)\cdot\sigma} = 1 \enspace.
    \end{equation}
    Combining \cref{eq:final-bound-anc,eq:final-bound1-anc,eq:final-bound2-anc} gives the desired bound.
\end{proof} 

\section{Scalable Non-Adaptive \texorpdfstring{$\pru$}{PRU} in the Symmetric Subspace}
\label{sec:non-adap-pru}
In this section, we present an application of the parallel Kac's walk to the construction of non-adaptive scalable pseudorandom unitaries ($\pru$s) in the symmetric subspace.
We begin with the definition of an (adaptive) scalable \pru\ family, and then introduce the corresponding non-adaptive notion restricted to the symmetric subspace. Let $\secpar$ denote the security parameter.

\begin{definition}[Scalable $\pru$]
    \label{def:adap-pru-sym}
    For each \(n,\secpar\in\mathbb N\),
    let $\mathcal{K}_{n,\secpar}$ be a key space,
    and let
    \[
    \mathcal{U}_{n,\secpar}
    \;:=\;
    \{\,U^{n,\secpar}_{k}\in \ugroup{2^n}\,\}_{k\in\mathcal{K}_{n,\secpar}}
    \]
    be a family of $n$-qubit unitaries.
    We say that the ensemble
    $\mathcal{U}\coloneq
\{\mathcal{U}_{n,\secpar}\}_{n,\secpar\in\mathbb{N}}$
    is an (adaptive) scalable pseudorandom unitary
    (\pru)
    if the following hold:
    \begin{itemize}
        \item \textbf{Efficient computation.}
        There exists a $\poly{n,\secpar}$-time quantum algorithm
        such that, on input $(n,\secpar,k)$,
        it outputs a classical description of a quantum circuit
        implementing $U^{n,\secpar}_k$.
        \item \textbf{Pseudorandomness.}
        For any $\poly{n,\secpar}$-time quantum oracle adversary $\A$ without knowing the description of the sampled unitary, we have
        \begin{align*}
            &\left| \Pr_{ k\gets\mathcal{K}_{n,\secpar} }
            \Br{\A^{U^{n,\secpar}_{k}}\!(1^{n}, 1^{\secpar}) = 1}
            -
            \Pr_{V\gets\haardist}
            \Br{\A^{V}\!(1^{n}, 1^{\secpar})= 1} \right| = \negl{\secpar} \enspace,
        \end{align*}
        where $d=2^n$ and $\haardist$ is the Haar distribution over $\ugroup{d}$.
    \end{itemize}
\end{definition}

The  notion of scalable \pru s\ extends the original notion of
\pru s\ \cite{JLS18}
by treating the security parameter and the number of qubits
as independent parameters, and requiring the pseudorandomness error to be negligible in the security parameter alone, independent of the system size. Constructing a scalable $\pru$ appears to be challenging. While a number of constructions of $\pru$s are now known~\cite{MPSY24,MH25,LQSY+25,science.adv8590,SMLBH25,MLSH25,CSMHB25},
to the best of our knowledge, none of them are scalable in this sense.
In this work, we consider a restricted variant, namely
\emph{non-adaptive scalable \pru s on the symmetric subspace}.

\begin{definition}[Non-adaptive scalable $\pru$ on the symmetric subspace]
    \label{def:non-adap-pru-sym-subspace}
    For each \(n,\secpar\in\mathbb N\),
    let $\mathcal{K}_{n,\secpar}$ be a key space,
    and let
    \(
    \mathcal{U}_{n,\secpar}
    \;:=\;
    \{\,U^{n,\secpar}_{k}\in \ugroup{2^n}\,\}_{k\in\mathcal{K}_{n,\secpar}}
    \)
    be a family of $n$-qubit unitaries.
    We say that the ensemble
    $\mathcal{U}\coloneq\{\mathcal{U}_{n,\secpar}\}_{n,\secpar\in\mathbb{N}}$
    is a non-adaptive scalable pseudorandom unitary
    (\pru) on the symmetric subspace
    if the following hold:
    \begin{itemize}
        \item \textbf{Efficient computation.}
        There exists a $\poly{n,\secpar}$-time quantum algorithm
        such that, on input $(n,\secpar,k)$,
        it outputs a classical description of a quantum circuit
        implementing $U^{n,\secpar}_k$.
        \item \textbf{Pseudorandomness.}
        For any $t = \poly{n,\secpar}$,
        let $\sigma$ be any efficiently preparable quantum state
        whose first $tn$ qubits are supported
        on the symmetric subspace $\symsubspace{t}{d}$ where $d = 2^n$, i.e.,
        \[
        	\tr{(\symsubspaceproj^{t,d}\otimes\id)\cdot\sigma} = 1 \enspace.
        \]
        Then, for any $\poly{n,\secpar}$-time quantum adversary $\A$,
        \begin{align*}
            &\left| \Pr_{ k\gets\mathcal{K}_{n,\secpar} }\Br{\A\br{\br{(U^{n,\secpar}_{k})^{\otimes t}\otimes I}\cdot \sigma \cdot \br{(U^{n,\secpar}_{k})^{\otimes t}\otimes I}^{\dagger} } = 1} \right. \\
            &\qquad\qquad\qquad\qquad\qquad
            \left.  - \Pr_{V\gets\haardist}\Br{\A\br{(V^{\otimes t}\otimes I)\cdot \sigma \cdot(V^{\otimes t}\otimes I)^{\dagger} } = 1} \right| = \negl{\secpar} \enspace,
        \end{align*}
        where $d=2^n$ and $\haardist$ is the Haar distribution over $\ugroup{d}$.
    \end{itemize}
\end{definition}

With these definitions, we are now ready to state the quantum pseudorandomness consequence of our convergence result.
This is admittedly a restricted notion of pseudorandom unitaries. Nevertheless, we view it as positive evidence that the parallel Kac's walk
may ultimately lead to a scalable \pru\ construction in greater generality.


\begin{theorem}\label{thm:pru}
Assuming quantum-secure one-way functions exist,
there exists a non-adaptive scalable \pru\ family
$\mathcal{U}
=\{\mathcal{U}_{n,\secpar}\}_{n,\secpar\in\mathbb{N}}$
on the symmetric subspace.
\end{theorem}

\begin{proof}
    Fix $n,\secpar\in\mathbb{N}$, and let $d=2^n$.
    We first show that the random unitary $K$ defined in
    \cref{eq:kac-random-unitary} can be implemented by a
    polynomial-size quantum circuit using random functions
    and random permutations as oracles.
    Indeed, by \cite[Lemma 19]{LQSY+24},
    the block-diagonal $W$ can be approximately implemented by
    a $\poly{n,\secpar}$-size quantum circuit using random functions
    with spectral-norm error at most $2^{-(n+\secpar)}$, by choosing
    the implementation precision polynomially in $n+\secpar$.\footnote{In the notation of \cite{LQSY+24}, $W$ is the matrix $\widetilde{Q}$, and its approximation is denoted by $Q$. Also, each block is a $\su(2)$ matrix there, which can be extended to $\ugroup{2}$ simply by adding two random phases using the random function.}
    The random permutation $P_\tau$ can be realized by querying the random permutation $\tau$ and $\tau^{-1}$:
    \[
    	\ket{x}\ket{0} \xrightarrow{O_\tau} \ket{x}\ket{\tau(x)}\xrightarrow{\mathrm{SWAP}}\ket{\tau(x)}\ket{x} \xrightarrow{O_{\tau^{-1}}} \ket{\tau(x)}\ket{0} \enspace.
    \]
 
    Let $r=\lceil C(n+\log^2(\secpar+2))\rceil$ for a sufficiently
    large absolute constant $C$.
    Denote by $\widetilde{K}$ the resulting approximate implementation of $K$,
    and let $\appkacchannel{t}$ and $\repappkacchannel{t}{r}$ be the corresponding channels defined analogously to $\kacchannel{t}$ and $\repkacchannel{t}{r}$.
    Since
    \(
        \|\widetilde{K}-K\|_{\infty}\le 2^{-(n+\secpar)},
    \)
    we have that for any $t=\poly{n,\secpar}$ and any state $\sigma$,
    \[
        \norm{
            (\repappkacchannel{t}{r+1}\otimes \id)(\sigma)
            -
            (\repkacchannel{t}{r+1}\otimes \id)(\sigma)
        }_1
        \le 2(r+1)\cdot t\cdot \|\widetilde{K}-K\|_{\infty}
        = \negl{\secpar} \enspace.
    \]
    Combining this approximation bound with \cref{thm:main} and the triangle inequality, we obtain that for any $t=\poly{n,\secpar}$ and any state $\sigma$ such that
    \(
        \tr{(\symsubspaceproj^{t,d}\otimes\id)\cdot \sigma}=1,
    \)
    \begin{equation}\label{eq:final-bound-anc-app}
        \norm{
            (\repappkacchannel{t}{r+1}\otimes \id)(\sigma)
            -
            (\haarchannel{t}\otimes \id)(\sigma)
        }_1
        \le
        \negl{\secpar}
        +
        2d^2\e^{-cr}
        = \negl{\secpar} \enspace.
    \end{equation}
    Let $\widetilde{U}$ denote the random unitary obtained by composing
    $r+1$ independent copies of $\widetilde K$,
    where each copy uses fresh random functions and random permutations.
    By \cref{eq:final-bound-anc-app}, for any $t=\poly{n,\secpar}$ and any state $\sigma$ such that
    \(
        \tr{(\symsubspaceproj^{t,d}\otimes\id)\cdot\sigma}=1,
    \)
    we have
    \begin{align}\label{eq:realistic}
        \norm{
            \E_{\widetilde U}\Br{
                (\widetilde U^{\otimes t}\otimes I)\cdot \sigma \cdot (\widetilde U^{\otimes t}\otimes I)^\dagger
            }
            -
            \E_{V\gets\haardist}\Br{
                (V^{\otimes t}\otimes I)\cdot \sigma \cdot (V^{\otimes t}\otimes I)^\dagger
            }
        }_1
        = \negl{\secpar}\enspace.
    \end{align}

    We now define the keyed family $\mathcal U_{n,\secpar}$
    by replacing all random functions and random permutations
    appearing in the above construction
    with quantum-secure pseudorandom functions
    and pseudorandom permutations, using internal security parameter
    $n+\secpar$.
    Let $\mathcal K_{n,\secpar}$ be the set of all keys specifying
    all $\prf$ and $\prp$ instances used across
    the $r+1$ layers of the construction.
    For each $k\in\mathcal K_{n,\secpar}$, let $U_k^{n,\secpar}$ be the resulting $n$-qubit unitary, and define
    \[
        \mathcal U_{n,\secpar}
        :=
        \{\,U_k^{n,\secpar}\in \ugroup{d}\,\}_{k\in\mathcal K_{n,\secpar}} \enspace.
    \]
    By construction, since $\prf$s and $\prp$s are efficiently computable,
    each $U_k^{n,\secpar}$ is efficiently computable.
    Moreover, by the security of the underlying
    pseudorandom primitives and a standard hybrid argument,
    no $\poly{n,\secpar}$-time quantum adversary can distinguish
    the state generated using $(U_k^{n,\secpar})^{\otimes t}$
    for uniformly random $k\gets\mathcal K_{n,\secpar}$ from
    the corresponding state generated using $\widetilde U^{\otimes t}$,
    except with negligible probability.
    Combining this with the bound in \cref{eq:realistic}
    proves the pseudorandomness condition in \cref{def:non-adap-pru-sym-subspace}.
\end{proof}


\endgroup

\newpage
\bibliographystyle{alpha}
\bibliography{ref}

@inproceedings{BY26,
  author        = {Zvika Brakerski and Henry Yuen},
  title         = {On Scalable Pseudorandom Unitaries and the
                   Unitary Synthesis Problem},
  editor        = {Nadia Heninger and Mike Rosulek},
  booktitle     = {Advances in Cryptology -- {CRYPTO} 2026, Part V},
  series        = {Lecture Notes in Computer Science},
  volume        = {16804},
  pages         = {372--401},
  publisher     = {Springer Nature Switzerland},
  year          = {2026},
  doi           = {10.1007/978-3-032-35409-9_13},
  eprint        = {2605.09957},
  archivePrefix = {arXiv},
  url           = {https://doi.org/10.1007/978-3-032-35409-9_13}
}

@article{HRW26,
  title={Random-Oracle Unitary Synthesis is Impossible},
  author={Andrew Huang and Akshar Ramkumar and John Wright},
  journal={arXiv:2610.06700},
  year={2026}
}

@misc{OpenAI26,
  author        = {{OpenAI}},
  title         = {{Polynomial-Time Unitary Synthesis from a Boolean Oracle}},
  howpublished  = {OpenAI Math Release preprint},
  year          = {2026},
  note          = {October 5, 2026},
  url           = {https://github.com/openai/math/blob/main/preprints/Polynomial-Time-Unitary-Synthesis-from-a-Boolean-Oracle-October-5-2026/paper.pdf}
}

@article{QS23,
  title={Mixing of the averaging process and its discrete dual on finite-dimensional geometries},
  author={Quattropani, Matteo and Sau, Federico},
  journal={The Annals of Applied Probability},
  volume={33},
  number={2},
  pages={1136--1171},
  year={2023},
  doi={10.1214/22-AAP1838}
}

@article{caputo2026,
  title={Aldous' spectral gap phenomena in stochastic exchange models},
  author={Pietro Caputo and Matteo Quattropani and Federico Sau},
  journal={arXiv:2609.10450},
  year={2026}
}

@article{Jiang2026,
  title={Resolvent characteristics and quadratic mixing of Kac's walk on $SO(n)$},
  author={Yunjiang Jiang},
  journal={arXiv:2609.08018},
  year={2026}
}

@article{HMH+23,
  title={Efficient Unitary Designs with a System-Size Independent Number of Non-{Clifford} Gates},
  author={Jonas Haferkamp and Felipe Montealegre-Mora and Markus Heinrich and Jens Eisert and David Gross and Ingo Roth},
  journal={Communications in Mathematical Physics},
  volume={397},
  number={3},
  pages={995--1041},
  year={2023},
  publisher={Springer},
  doi={10.1007/s00220-022-04507-6},
  url={https://doi.org/10.1007/s00220-022-04507-6},
}

@inproceedings{CHHLMT25,
  title={Incompressibility and spectral gaps of random circuits},
  author={Chen, Chi-Fang and Haah, Jeongwan and Haferkamp, Jonas and Liu, Yunchao and Metger, Tony and Tan, Xinyu},
  booktitle={IEEE 66th Annual Symposium on Foundations of Computer Science -- FOCS 2025},
  pages={1304--1312},
  year={2025},
  organization={IEEE},
  doi={10.1109/FOCS63196.2025.00069},
  url={https://doi.org/10.1109/FOCS63196.2025.00069},
}

@Article{Caputo2008,
author={Pietro Caputo},
title={On the spectral gap of the Kac walk and other binary collision processes},
journal={ALEA. Latin American Journal of Probability and Mathematical Statistics},
year={2008},
volume={4},
pages={205–222}
}

@article{morris2006mixing,
  title={The Mixing Time for Simple Exclusion},
  author={Morris, Ben},
  journal={The Annals of Applied Probability},
  pages={615--635},
  year={2006},
  publisher={JSTOR}
}

@article{caputo2010proof,
  title={Proof of {Aldous'} spectral gap conjecture},
  author={Caputo, Pietro and Liggett, Thomas and Richthammer, Thomas},
  journal={Journal of the American Mathematical Society},
  volume={23},
  number={3},
  pages={831--851},
  year={2010}
}

@inproceedings{anari2022entropic,
  title={Entropic independence: optimal mixing of down-up random walks},
  author={Anari, Nima and Jain, Vishesh and Koehler, Frederic and Pham, Huy Tuan and Vuong, Thuy-Duong},
  booktitle={Proceedings of the 54th Annual ACM SIGACT Symposium on Theory of Computing -- STOC 2022},
  pages={1418--1430},
  year={2022}
}

@inproceedings{anari2018log,
  title={Log-concave polynomials, entropy, and a deterministic approximation algorithm for counting bases of matroids},
  author={Anari, Nima and Gharan, Shayan Oveis and Vinzant, Cynthia},
  booktitle={IEEE 59th Annual Symposium on Foundations of Computer Science -- FOCS 2018},
  pages={35--46},
  year={2018},
  organization={IEEE}
}

@inproceedings{anari2019log,
  title={Log-concave polynomials II: high-dimensional walks and an FPRAS for counting bases of a matroid},
  author={Anari, Nima and Liu, Kuikui and Gharan, Shayan Oveis and Vinzant, Cynthia},
  booktitle={Proceedings of the 51st Annual ACM SIGACT Symposium on Theory of Computing},
  pages={1--12},
  year={2019}
}

@article{anari2024log,
  title={Log-concave polynomials III: Mason’s ultra-log-concavity conjecture for independent sets of matroids},
  author={Anari, Nima and Liu, Kuikui and Oveis Gharan, Shayan and Vinzant, Cynthia},
  journal={Proceedings of the American Mathematical Society},
  volume={152},
  number={05},
  pages={1969--1981},
  year={2024}
}

@inproceedings{anari2021log,
  title={Log-concave polynomials IV: approximate exchange, tight mixing times, and near-optimal sampling of forests},
  author={Anari, Nima and Liu, Kuikui and Gharan, Shayan Oveis and Vinzant, Cynthia and Vuong, Thuy-Duong},
  booktitle={Proceedings of the 53rd Annual ACM SIGACT Symposium on Theory of Computing},
  pages={408--420},
  year={2021}
}

@article{branden2020lorentzian,
  title={Lorentzian polynomials},
  author={Br{\"a}nd{\'e}n, Petter and Huh, June},
  journal={Annals of Mathematics},
  volume={192},
  number={3},
  pages={821--891},
  year={2020},
  publisher={Department of Mathematics, Princeton University Princeton, New Jersey, USA}
}

@inproceedings{narayanan2023sampling,
  title={Sampling from convex sets with a cold start using multiscale decompositions},
  author={Narayanan, Hariharan and Rajaraman, Amit and Srivastava, Piyush},
  booktitle={Proceedings of the 55th Annual ACM Symposium on Theory of Computing},
  pages={117--130},
  year={2023}
}

@article{laddha2023convergence,
  title={Convergence of {Gibbs} sampling: Coordinate Hit-and-Run mixes fast},
  author={Laddha, Aditi and Vempala, Santosh S},
  journal={Discrete \& Computational Geometry},
  volume={70},
  number={2},
  pages={406--425},
  year={2023},
  publisher={Springer}
}

@article{narayanan2022mixing,
  title={On the mixing time of coordinate hit-and-run},
  author={Narayanan, Hariharan and Srivastava, Piyush},
  journal={Combinatorics, Probability and Computing},
  volume={31},
  number={2},
  pages={320--332},
  year={2022},
  publisher={Cambridge University Press}
}

@inproceedings{alev2020improved,
  title={Improved analysis of higher order random walks and applications},
  author={Alev, Vedat Levi and Lau, Lap Chi},
  booktitle={Proceedings of the 52nd annual ACM SIGACT symposium on theory of computing -- STOC2020},
  pages={1198--1211},
  year={2020}
}

@article{anari2021spectral,
  title={Spectral independence in high-dimensional expanders and applications to the hardcore model},
  author={Anari, Nima and Liu, Kuikui and Gharan, Shayan Oveis},
  journal={SIAM Journal on Computing},
  volume={53},
  number={6},
  pages={FOCS20--1},
  year={2021},
  publisher={SIAM}
}

@inproceedings{chen2021optimal,
  title={Optimal mixing of Glauber dynamics: Entropy factorization via high-dimensional expansion},
  author={Chen, Zongchen and Liu, Kuikui and Vigoda, Eric},
  booktitle={Proceedings of the 53rd Annual ACM SIGACT Symposium on Theory of Computing},
  pages={1537--1550},
  year={2021}
}

@inproceedings{chen2022optimal,
  title={Optimal mixing for two-state anti-ferromagnetic spin systems},
  author={Chen, Xiaoyu and Feng, Weiming and Yin, Yitong and Zhang, Xinyuan},
  booktitle={2022 IEEE 63rd Annual Symposium on Foundations of Computer Science (FOCS)},
  pages={588--599},
  year={2022},
  organization={IEEE}
}

@article{chen2024rapid,
  title={Rapid mixing of Glauber dynamics via spectral independence for all degrees},
  author={Chen, Xiaoyu and Feng, Weiming and Yin, Yitong and Zhang, Xinyuan},
  journal={SIAM Journal on Computing},
  pages={FOCS21--224},
  year={2024},
  publisher={SIAM}
}

@article{chen2024spectral,
  title={Spectral independence via stability and applications to {Holant}-type problems},
  author={Chen, Zongchen and Liu, Kuikui and Vigoda, Eric},
  journal={TheoretiCS},
  volume={3},
  year={2024},
  publisher={Episciences. org}
}

@article{feng2022rapid,
  title={Rapid mixing from spectral independence beyond the Boolean domain},
  author={Feng, Weiming and Guo, Heng and Yin, Yitong and Zhang, Chihao},
  journal={ACM Transactions on Algorithms (TALG)},
  volume={18},
  number={3},
  pages={1--32},
  year={2022},
  publisher={ACM New York, NY}
}

@inproceedings{cryan2019modified,
  title={Modified log-Sobolev inequalities for strongly log-concave distributions},
  author={Cryan, Mary and Guo, Heng and Mousa, Giorgos},
  booktitle={IEEE 60th Annual Symposium on Foundations of Computer Science -- FOCS 2019},
  pages={1358--1370},
  year={2019},
  organization={IEEE}
}

@article{chen2023rapid,
  title={Rapid mixing of Glauber dynamics up to uniqueness via contraction},
  author={Chen, Zongchen and Liu, Kuikui and Vigoda, Eric},
  journal={SIAM Journal on Computing},
  volume={52},
  number={1},
  pages={196--237},
  year={2023},
  publisher={SIAM}
}

@inproceedings{chen2022localization,
  title={Localization schemes: A framework for proving mixing bounds for Markov chains},
  author={Chen, Yuansi and Eldan, Ronen},
  booktitle={2022 IEEE 63rd Annual symposium on foundations of computer science (FOCS)},
  pages={110--122},
  year={2022},
  organization={IEEE}
}

@article{kim2025spectral,
  title={Spectral gap of the {KMP} and other stochastic exchange models on arbitrary graphs},
  author={Kim, Seonwoo and Quattropani, Matteo and Sau, Federico},
  journal={arXiv:2505.02400},
  year={2025}
}

@article{oliveira2013mixing,
  title={Mixing of the symmetric exclusion processes in terms of the corresponding single-particle random walk},
  author={Oliveira, Roberto Imbuzeiro},
  journal={Annals of probability: An official journal of the Institute of Mathematical Statistics},
  volume={41},
  number={2},
  pages={871--913},
  year={2013},
  publisher={Institute of Mathematical Statistics}
}

@InProceedings{Sasada2013,
author="Sasada, Makiko",
title="On the Spectral Gap of the {Kac} Walk and Other Binary Collision Processes on d-Dimensional Lattice",
booktitle="Symmetries, Integrable Systems and Representations",
year="2013",
publisher="Springer",
pages="543--560",
}

@Article{Carlen2003,
author={Eric A. Carlen
and Mario C. Carvalho
and Michael Loss},
title={Determination of the spectral gap for {Kac's} master equation and related stochastic evolution},
journal={Acta Mathematica},
year={2003},
month={Mar},
day={01},
volume={191},
number={1},
pages={1-54},
doi={10.1007/BF02392695},
url={https://doi.org/10.1007/BF02392695}
}

@article{Jan2001,
 ISSN = {00911798, 2168894X},
 URL = {http://www.jstor.org/stable/2652922},
 author = {Elise Janvresse},
 journal = {The Annals of Probability},
 number = {1},
 pages = {288--304},
 publisher = {Institute of Mathematical Statistics},
 title = {Spectral Gap for Kac's Model of Boltzmann Equation},
 urldate = {2026-03-31},
 volume = {29},
 year = {2001}
}

@article{Haferkamp2022randomquantum,
  doi = {10.22331/q-2022-09-08-795},
  url = {https://doi.org/10.22331/q-2022-09-08-795},
  title = {Random quantum circuits are approximate unitary {$t$}-designs in depth {$O(nt^{5+o(1)})$}},
  author = {Haferkamp, Jonas},
  journal = {{Quantum}},
  issn = {2521-327X},
  publisher = {{Verein zur F{\"{o}}rderung des Open Access Publizierens in den Quantenwissenschaften}},
  volume = {6},
  pages = {795},
  month = sep,
  year = {2022}
}

@article{PhysRevA.104.022417,
  title = {Improved spectral gaps for random quantum circuits: Large local dimensions and all-to-all interactions},
  author = {Haferkamp, Jonas and Hunter-Jones, Nicholas},
  journal = {Phys. Rev. A},
  volume = {104},
  issue = {2},
  pages = {022417},
  numpages = {18},
  year = {2021},
  month = {Aug},
  publisher = {American Physical Society},
  doi = {10.1103/PhysRevA.104.022417},
  url = {https://link.aps.org/doi/10.1103/PhysRevA.104.022417}
}

@INPROCEEDINGS{OSRP:2023,
  author={O’Donnell, Ryan and Servedio, Rocco A. and Paredes, Pedro},
  booktitle={IEEE 64th Annual Symposium on Foundations of Computer Science -- FOCS 2023}, 
  title={Explicit orthogonal and unitary designs}, 
  year={2023},
  pages={1240-1260},
  doi={10.1109/FOCS57990.2023.00073}}

@article{
science.adv8590,
author = {Thomas Schuster  and Jonas Haferkamp  and Hsin-Yuan Huang },
title = {Random unitaries in extremely low depth},
journal = {Science},
volume = {389},
number = {6755},
pages = {92-96},
year = {2025},
doi = {10.1126/science.adv8590},
URL = {https://www.science.org/doi/abs/10.1126/science.adv8590},
eprint = {https://www.science.org/doi/pdf/10.1126/science.adv8590}}

@article{LQSY+24,
author = {Lu, Chuhan and Qin, Minglong and Song, Fang and Yao, Penghui and Zhao, Mingnan},
title = {Quantum Pseudorandom Scramblers},
journal = {SIAM Journal on Computing},
volume = {55},
number = {4},
pages = {787-850},
year = {2026},
doi = {10.1137/24M170990X}
}

@INPROCEEDINGS{MPSY24,
  author={Metger, Tony and Poremba, Alexander and Sinha, Makrand and Yuen, Henry},
  booktitle={IEEE 65th Annual Symposium on Foundations of Computer Science -- FOCS 2024}, 
  title={Simple Constructions of Linear-Depth t-Designs and Pseudorandom Unitaries}, 
  year={2024},
  pages={485-492},
  doi={10.1109/FOCS61266.2024.00038}
}

@article{Hastings70,
  title={{Monte Carlo} sampling methods using {Markov} chains and their applications},
  author={Wilfred Keith Hastings},
  journal={Biometrika},
  volume={57},
  number={1},
  pages={97--109},
  year={1970},
  doi={10.1093/biomet/57.1.97},  
}

@book{LPW09,
  title     = {Markov chains and mixing times},
  author    = {Levin, David A and Peres, Yuval},
  year      = {2017},
  publisher = {American Mathematical Society},
}

@article{PS17,
  title={Kac's walk on $ n $-sphere mixes in $ n\log n $ steps},
  author={Pillai, Natesh S and Smith, Aaron},
  journal={The Annals of Applied Probability},
  volume={27},
  number={1},
  pages={631--650},
  year={2017},
  publisher={Institute of Mathematical Statistics},
  doi={10.1214/16-AAP1214}
}

@article{JPSSS22,
author = {Vishesh Jain and Natesh S. Pillai and Ashwin Sah and Mehtaab Sawhney and Aaron Smith},
title = {{Fast and memory-optimal dimension reduction using {Kac}’s walk}},
volume = {32},
journal = {The Annals of Applied Probability},
number = {5},
publisher = {Institute of Mathematical Statistics},
pages = {4038 -- 4064},
year = {2022},
doi = {10.1214/22-AAP1784}
}

@article{BHH16,
  title={Local random quantum circuits are approximate polynomial-designs},
  author={Brand{\~a}o, Fernando G.S.L. and Harrow, Aram W. and
                  Horodecki, Micha{\l}},
  journal={Communications in Mathematical Physics},
  volume={346},
  pages={397--434},
  year={2016},
  publisher={Springer},
  doi={10.1007/s00220-016-2706-8}
}

@article{Oliveira09,
  title={On The Convergence to Equilibrium of {K}ac's Random Walk on Matrices},
  author={Oliveira, Roberto Imbuzeiro},
  journal={The Annals of Applied Probability},
  volume = {19},
  number={3},
  pages={1200--1231},
  year={2009},
  publisher={JSTOR},
  url={https://www.jstor.org/stable/30243617}
}

@InProceedings{JLS18,
  author={Ji, Zhengfeng and Liu, Yi-Kai and Song, Fang},
  title={Pseudorandom Quantum States},
  booktitle={Advances in Cryptology -- CRYPTO 2018},
  year={2018},
  publisher={Springer},
  pages={126--152},
  doi={10.1007/978-3-319-96878-0_5}
}

@article{SMLBH25,
  title={Strong random unitaries and fast scrambling},
  author={Schuster, Thomas and Ma, Fermi and Lombardi, Alex and Brandao, Fernando and Huang, Hsin-Yuan},
  journal={arXiv:2509.26310},
  year={2025}
}

@article{MLSH25,
  title={Random unitaries that conserve energy},
  author={Mao, Liang and Cui, Laura and Schuster, Thomas and Huang, Hsin-Yuan},
  journal={arXiv:2510.08448},
  year={2025}
}

@article{CSMHB25,
  title={Random unitaries from Hamiltonian dynamics},
  author={Cui, Laura and Schuster, Thomas and Mao, Liang and Huang, Hsin-Yuan and Brandao, Fernando},
  journal={arXiv:2510.08434},
  year={2025}
}

@Article{Diaconis2000,
author={Diaconis, Persi
and Saloff-Coste, Laurent},
title={Bounds for Kac's Master Equation},
journal={Communications in Mathematical Physics},
year={2000},
month={Feb},
day={01},
volume={209},
number={3},
pages={729-755},
issn={1432-0916},
doi={10.1007/s002200050036},
url={https://doi.org/10.1007/s002200050036}
}

@inproceedings{Kac56,
  title={Foundations of kinetic theory},
  author={Kac, Mark},
  booktitle={Third Berkeley symposium on mathematical statistics and
                  probability},
  volume={3},
  pages={171--197},
  year={1956},
}

@article{PS18,
  title={On the mixing time of {Kac’s} walk and other high-dimensional
                  {Gibbs} samplers with constraints},
  author={Pillai, Natesh S and Smith, Aaron},
  journal={The Annals of Probability},
  volume={46},
  number={4},
  pages={2345--2399},
  year={2018},
  publisher={JSTOR},
  url={https://www.jstor.org/stable/26506605}
}

@inproceedings{MH25,
  title={How to construct random unitaries},
  author={Ma, Fermi and Huang, Hsin-Yuan},
  booktitle={Proceedings of the 57th Annual ACM Symposium on Theory of Computing -- STOC 2025},
  pages={806--809},
  year={2025},
  url={https://doi.org/10.1145/3717823.3718254},
  doi={10.1145/3717823.3718254}
}

@article{KMP82,
  title={Heat flow in an exactly solvable model},
  author={Kipnis, Claude and Marchioro, Carlo and Presutti, Errico},
  journal={Journal of Statistical Physics},
  volume={27},
  number={1},
  pages={65--74},
  year={1982},
  publisher={Springer},
  doi={10.1007/BF01011791}
}

@book{BN04,
  title={A primer on statistical distributions},
  author={Balakrishnan, Narayanaswamy and Nevzorov, Valery B},
  year={2004},
  publisher={John Wiley \& Sons}
}

@article{CCE2018,
	title = {Entropy production inequalities for the {Kac} {Walk}},
	volume = {11},
	issn = {1937-5093},
	doi = {10.3934/krm.2018012},
	language = {en},
	number = {2},
	urldate = {2026-03-31},
	journal = {Kinetic and Related Models},
	publisher = {Kinetic and Related Models},
	author = {Carlen, Eric A. and Carvalho, Maria C. and Einav, Amit},
	month = apr,
	year = {2018},
	pages = {219--238},
}

@article{CCL+10,
	title = {Entropy and chaos in the {Kac} model},
	volume = {3},
	copyright = {http://creativecommons.org/licenses/by/3.0/},
	issn = {1937-5093},
	doi = {10.3934/krm.2010.3.85},
	language = {en},
	number = {1},
	urldate = {2026-03-31},
	journal = {Kinetic and Related Models},
	publisher = {Kinetic and Related Models},
	author = {Carlen, Eric A. and Carvalho, Maria C. and Roux, Jonathan Le and Loss, Michael and Villani, Cédric},
	month = jan,
	year = {2010},
	pages = {85--122},
}

@article{LQSY+25,
  title={Parallel {Kac's} Walk Generates {PRU}},
  author={Chuhan Lu and Minglong Qin and Fang Song and Penghui Yao and Mingnan Zhao},
  journal={arXiv:2504.14957},
  year={2025}
}

@article{KV26,
  title={Spectral Gaps of {Hit-and-Run} and
                   {Coordinate Hit-and-Run}},
  author={Yunbum Kook and Santosh S. Vempala},
  journal={arXiv:2608.16878},
  year={2026}
}

@book{Meckes2019,
  author    = {Meckes, Elizabeth S.},
  title     = {The Random Matrix Theory of the Classical Compact Groups},
  publisher = {Cambridge University Press},
  year      = {2019},
  doi       = {10.1017/9781108318220}
}

@article{PS26,
  title={Kac's walk on rotation matrices mixes in $n^2 \log n$ steps},
  author={Natesh S. Pillai and Aaron Smith},
  journal={arXiv:2604.23828},
  year={2026}
}

@article{Smi14,
 author = {Aaron Smith},
 journal = {The Annals of Applied Probability},
 number = {1},
 pages = {114--130},
 publisher = {Institute of Mathematical Statistics},
 title = {A {Gibbs} SAMPLER ON THE n-SIMPLEX},
 urldate = {2026-11-03},
 volume = {24},
 year = {2014}
}

@misc{PSV26,
      title={On the Pseudo-Mixing of {Kac's} Walk},
      author={Natesh S. Pillai and Aaron Smith and Vinod Vaikuntanathan},
      year={2026},
      howpublished={arXiv:2608.17374},
      url={https://arxiv.org/abs/2608.17374},
}

@inproceedings{VZ26,
  title={Improving algorithmic efficiency using cryptography: Trapdoored matrices and applications},
  author={Vaikuntanathan, Vinod and Zamir, Or},
  booktitle={Proceedings of the 2026 Annual ACM-SIAM Symposium on Discrete Algorithms -- SODA 2026},
  pages={2554--2574},
  year={2026},
  organization={SIAM}
}

@article{PR26,
  author  = {Pymar, Richard and Rivera, Nicol{\'a}s},
  title   = {Mixing of the symmetric beta-binomial splitting process on arbitrary graphs},
  journal = {Annales de l'Institut Henri Poincar{\'e}, Probabilit{\'e}s et Statistiques},
  volume  = {62},
  number  = {1},
  pages   = {110--145},
  year    = {2026},
  doi     = {10.1214/24-AIHP1530}
}

\begingroup
\setstretch{1.1}
\newpage
\appendix

\section{Spectral Reduction from the Two-Particle to the Single-Particle}
\label[appendix]{sec:lambda-2}
In this section, we relate the spectrum of \(\lkmpchain^{(2,G)}\)
to that of two simpler chains
and then use this relation to prove \cref{cor:two-particle-labeled-gap-compared-to-single-particle}.
The key observation is that the function space \(\R^{V\times V}\) admits a decomposition into symmetric and antisymmetric subspaces.
On the symmetric subspace, the labeled chain reduces to the unlabeled KMP chain with two particles.
On the antisymmetric subspace, it is governed by
a Markov chain induced by the single-particle random walk on graphs.
This decomposition allows us to compare the two-particle labeled spectrum with the unlabeled spectrum.
Let \(t=2\) and fix a graph $G$ throughout this section.

\subsection{Decomposition of the Two-Particle Labeled KMP Chain}

Let $\R^{V\times V} \coloneqq \{ f: V\times V\to\R \}$
be the space of real-valued functions on \(V\times V\).
This space admits a natural direct sum decomposition
into \emph{symmetric} and \emph{antisymmetric subspaces}.
The symmetric subspace consists of functions that remain
invariant after swapping the two variables, denoted by
\[
    \mathcal H_{\sym}:=\set{f \in \R^{V\times V}: f(y,z)=f(z,y)\text{ for all }y,z\in V} \enspace,
\]
while the antisymmetric subspace consists of functions that change
sign after swapping the two variables, denoted by
\[
    \mathcal H_{\asym}:=\set{f \in \R^{V\times V}: f(y,z)=-f(z,y)\text{ for all }y,z\in V} \enspace.
\]
Therefore,
\(
\R^{V\times V}=\mathcal H_{\sym}\oplus \mathcal H_{\asym}
\),
since every function \(f\) can be uniquely decomposed as
\[
    f(y,z)
    =
    \frac{f(y,z)+f(z,y)}{2}+\frac{f(y,z)-f(z,y)}{2}\enspace,
\]
where the first term lies in \(\mathcal H_{\sym}\), and the second lies in \(\mathcal H_{\asym}\).

To analyze the spectrum of \(\lkmpchain^{(2,G)}\), we begin by describing its action along a fixed edge.
Let \(e=\{a,b\}\in E\) be a fixed edge, and let \(A_{e,p}\)
denote the one-step transition matrix
conditioned on choosing the edge \(e\) and the parameter \(p\in[0,1]\), as in \cref{eq:def-A-non-parallel}:
\begin{equation}\label{eq:def-A-e-p-non-parallel}
    A_{e,p}(i,j) \coloneqq \begin{cases}
        p & \text{if } (i,j) = (a,a) \text{ or } (i,j) = (b,a) \\
        1-p & \text{if } (i,j) = (a,b) \text{ or } (i,j) = (b,b) \\
        1 & \text{if } i=j \notin \set{a,b} \\
        0 & \text{otherwise}
    \end{cases} \enspace.
\end{equation}
After identifying a function \(f:V\times V\to\R\) with the corresponding column vector
indexed by \(V\times V\), the conditional one-step action on two labeled particles
is given by \(A_{e,p}\otimes A_{e,p}\).
Thus,
from \cref{rmk:transition-matrix-labeled-chain-non-parallel}, the transition matrix of the two-particle labeled KMP chain is given by
\begin{equation}\label{eq:lkmpchain-two-particles}
    \lkmpchain^{(2,G)}
    =
    \E_{e,p}\Br{A_{e,p} \otimes A_{e,p} }
    =
    \E_e\Br{\E_p\Br{A_{e,p} \otimes A_{e,p}}} \enspace,
\end{equation}
where $e$ is a uniformly random edge and $p\sim\Unif[0,1]$ is independent of $e$.
We consider the action of \(A_e:=\E_p\Br{A_{e,p} \otimes A_{e,p}}\)
for a fixed edge \(e=\{a,b\}\).

The next two lemmas describe the action of \(A_e\) and its consequence for the symmetric and antisymmetric decomposition of \(\R^{V\times V}\).

\begin{lemma}\label{lem:action-of-Ae}
For every \(f:V\times V\to\R\) and every \((y,z)\in V\times V\), we have
\[
    (A_ef)(y,z)
    =
    \begin{cases}
        f(y,z), & \text{if } y,z\notin\set{a,b},\\[0.6em]
        \dfrac{f(a,z)+f(b,z)}{2}, & \text{if } y\in\set{a,b}\text{ and } z\notin\set{a,b},\\[1em]
        \dfrac{f(y,a)+f(y,b)}{2}, & \text{if } y\notin\set{a,b}\text{ and } z\in\set{a,b},\\[1em]
        \dfrac{1}{3}f(a,a)+\dfrac{1}{6}f(a,b)+\dfrac{1}{6}f(b,a)+\dfrac{1}{3}f(b,b),
        & \text{if } y,z\in\set{a,b}.
    \end{cases}
\]
\end{lemma}

\begin{proof}
    For every \(f:V\times V\to\R\), the value of \(A_ef\) at \((y,z)\)
    is given as follows.
    \begin{itemize}
        \item If \(y,z\notin\set{a,b}\), then neither particle is affected, so
        \[
            (A_ef)(y,z)=f(y,z) \enspace.
        \]
        \item If \(y\in\set{a,b}\) and \(z\notin\set{a,b}\), then only the first particle is updated.
        Conditioned on \(p\), this particle moves to \(a\) with probability \(p\) and to \(b\) with probability \(1-p\), while the second particle stays at \(z\).
        Therefore
        \[
            (A_{e,p}\otimes A_{e,p})f(y,z)
            =
            p\,f(a,z)+(1-p)\,f(b,z) \enspace,
        \]
        and averaging over \(p\sim\Unif[0,1]\) gives
        \[
            (A_ef)(y,z)
            =
            \frac{f(a,z)+f(b,z)}{2} \enspace.
        \]
        \item Similarly, if \(y\notin\set{a,b}\) and \(z\in\set{a,b}\), then
        \[
            (A_ef)(y,z)
            =
            \frac{f(y,a)+f(y,b)}{2} \enspace.
        \]
        \item If \(y,z\in\set{a,b}\), then both particles are updated.
        Conditioned on \(p\), the two particles are reassigned independently inside the edge,
        so the four outcomes \((a,a)\), \((a,b)\), \((b,a)\), and \((b,b)\)
        occur with probabilities \(p^2\), \(p(1-p)\), \(p(1-p)\), and \((1-p)^2\), respectively.
        Thus
        \[
            (A_{e,p}\otimes A_{e,p})f(y,z)
            =
            p^2 f(a,a)+p(1-p)f(a,b)+p(1-p)f(b,a)+(1-p)^2f(b,b) \enspace.
        \]
        Since \(\E_p[p^2]=\E_p[(1-p)^2]=1/3\) and \(\E_p[p(1-p)]=1/6\),
        averaging over \(p\) yields
        \[
            (A_ef)(y,z)
            =
            \frac{1}{3}f(a,a)+\frac{1}{6}f(a,b)+\frac{1}{6}f(b,a)+\frac{1}{3}f(b,b)
            \qquad \text{for all }\, y,z\in\set{a,b} \enspace.\qedhere
        \]
    \end{itemize}
\end{proof}

\begin{lemma}
\label{lem:lkmp-commutes-swap}
Both \(\mathcal H_{\sym}\) and \(\mathcal H_{\asym}\)
are invariant subspaces of \(\lkmpchain^{(2,G)}\).
\end{lemma}

\begin{proof}
We identify $\mathbb{R}^{V\times V}$ with $\mathbb{R}^V \otimes \mathbb{R}^V$ via
 $f(y,z) \leftrightarrow \sum_{y,z} f(y,z)\, e_y \otimes e_z$.
Let $S:\mathbb{R}^{V\times V}\to \mathbb{R}^{V\times V}$ be the swap operator defined by
\[
S(e_y\otimes e_z)=e_z\otimes e_y \qquad \forall\, y,z\in V \enspace.
\]
Then $\mathcal{H}_{\sym}$ and $\mathcal{H}_{\asym}$ are exactly the $+1$ and $-1$
eigenspaces of $S$, respectively. For every edge $e$ and every $p\in[0,1]$, we have
\[
S\cdot\bigl(A_{e,p}\otimes A_{e,p}\bigr)=\bigl(A_{e,p}\otimes A_{e,p}\bigr)\cdot S\enspace.
\]
 Averaging over $p$ gives
$S\cdot A_e=A_e \cdot S$, and averaging further over $e$ gives
\[
S\cdot \lkmpchain^{(2,G)} = \lkmpchain^{(2,G)} \cdot S\enspace.
\]
Therefore both eigenspaces of $S$ are invariant under $\lkmpchain^{(2,G)}$.
 \end{proof}

This allows us to analyze the spectrum of \(\lkmpchain^{(2,G)}\) by analyzing its restrictions to \(\mathcal H_{\sym}\) and \(\mathcal H_{\asym}\) separately.
We first analyze its restriction to \(\mathcal H_{\sym}\), denoted by \(\lkmpchain^{(2,G)}|_{\mathcal H_{\sym}}\).
Recall that the unlabeled KMP chain \(\kmpchain^{(2,G)}\) is a Markov chain on \(\Delta_G\) defined in \cref{def:kmp-general}.

\begin{lemma}
\label{lem:symmetric-subspace-kmp}
The restricted matrix \(\lkmpchain^{(2,G)}|_{\mathcal H_{\sym}}\)
has the same eigenvalues as \(\kmpchain^{(2,G)}\).
\end{lemma}

\begin{proof}
We define a linear map
\[
T_{\sym}:\mathcal H_{\sym}\to \mathbb{R}^{\Delta_G}
\]
by
\[
(T_{\sym}f)(e_i+e_j):=f(i,j)\qquad \forall\, i,j\in V.
\]
This is well defined because $f$ is symmetric.

Fix a state $x=(y,z)\in V\times V$, and let $(Y',Z')$ denote the next state of the labeled chain
when started from $x$. To see the relation with the unlabeled chain, condition on a fixed edge
$e=\{a,b\}\in E$. Then the distribution of the new multiplicity vector $e_{Y'}+e_{Z'}$ is as follows.

\begin{itemize}
    \item If $y,z\notin\{a,b\}$, then neither particle is affected, and
\[
e_{Y'}+e_{Z'} = e_y+e_z \enspace.
\]
    \item If $y\in\{a,b\}$ and $z\notin\{a,b\}$, then only the first particle is updated. After averaging over
$p\sim \mathrm{Unif}[0,1]$, the new multiplicity vector $e_{Y'}+e_{Z'}$ is
\[
e_a+e_z \quad \text{or} \quad e_b+e_z
\]
with probabilities $1/2$ and $1/2$, respectively. The case $y\notin\{a,b\}$ and $z\in\{a,b\}$ is
similar.
\item If $y,z\in\{a,b\}$, then both particles are redistributed inside the edge. After averaging over $p$,
the four labeled outcomes $(a,a)$, $(a,b)$, $(b,a)$, and $(b,b)$ occur with probabilities
$1/3$, $1/6$, $1/6$, and $1/3$, respectively. Hence the multiplicity vector $e_{Y'}+e_{Z'}$ is equal
to
\[
2e_a,\qquad e_a+e_b,\qquad 2e_b
\]
with probabilities $1/3$, $1/3$, and $1/3$, respectively.
\end{itemize}

These are exactly the transition rules of the unlabeled KMP chain started from $e_y+e_z$.
Therefore, for every $f\in \mathcal H_{\sym}$,
\[
(T_{\sym}A_e f)(e_y+e_z)
=
(A_e f)(y,z)
=
(K_e T_{\sym}f)(e_y+e_z),
\]
where $K_e$ denotes the edge-conditioned transition operator of the unlabeled two-particle KMP
chain. Since this holds for every $y,z\in V$, we obtain
\[
T_{\sym}A_e = K_e T_{\sym}.
\]
Averaging over $e$ gives
\[
T_{\sym}
\Bigl(\lkmpchain^{(2,G)}|_{\mathcal H_{\sym}}\Bigr)
=
\kmpchain^{(2,G)}\, T_{\sym}.
\]
It remains to note that $T_{\sym}$ is an isomorphism: indeed, both spaces have dimension
$\binom{d+1}{2}$, and if $T_{\sym}f=0$, then $f(i,j)=0$ for all $i,j$, so $f=0$. Hence the
two operators are similar, and therefore have the same eigenvalues.
\end{proof}

We next analyze the antisymmetric subspace.
Let us denote the restriction of \(\lkmpchain^{(2,G)}\) to \(\mathcal H_{\asym}\) by \(\lkmpchain^{(2,G)}|_{\mathcal H_{\asym}}\).
Unlike the symmetric subspace, the antisymmetric subspace does not collapse to an unlabeled multiplicity profile.
Instead, the two-particle update still remembers the two coordinates separately,
and the relevant comparison object is the two-particle chain obtained from the single-particle random walk on \(G\).
For each fixed edge \(e=\{a,b\}\in E\), let \(C_e\) denote the transition matrix
of the random walk conditioned on choosing \(e\), namely
\[
    C_e(y,z)
    =
    \begin{cases}
        \frac{1}{2}, & \text{if } y\in\set{a,b} \text{ and } z\in\set{a,b},\\[0.3em]
        1, & \text{if } y=z \notin\set{a,b},\\[0.3em]
        0, & \text{otherwise}
    \end{cases} \enspace.
\]
Equivalently, for every \(h:V\to\R\),
\[
    (C_e h)(y)
    =
    \begin{cases}
        h(y), & \text{if } y\notin\set{a,b},\\[0.3em]
        \frac{h(a)+h(b)}{2}, & \text{if } y\in\set{a,b}.
    \end{cases}
\]
Now define
\[
    C
    :=
    \frac{1}{\abs{E}}\sum_{e\in E} C_e\otimes C_e \enspace.
\]
Intuitively, \(C\) describes the following two-particle process:
first choose a uniformly random edge \(e=\{a,b\}\),
and then let the two particles move independently according to the random walk rule on that edge,
so that each particle currently at \(a\) or \(b\) moves to \(a\) and \(b\) with probabilities \(1/2\) and \(1/2\), respectively,
while particles outside the edge remain fixed.

\begin{lemma}
\label{lem:antisymmetric-subspace-rw}
The restriction of \(\lkmpchain^{(2,G)}\) to \(\mathcal H_{\asym}\)
coincides with \(C\).
Namely,
\[
    \lkmpchain^{(2,G)}|_{\mathcal H_{\asym}}
    =
    C|_{\mathcal H_{\asym}} \enspace.
\]
Consequently, \(\lkmpchain^{(2,G)}|_{\mathcal H_{\asym}}\) and \(C|_{\mathcal H_{\asym}}\)
have the same eigenvalues.
\end{lemma}

\begin{proof}
Recall that \(A_e=\E_p\Br{A_{e,p} \otimes A_{e,p}}\), where \(A_{e,p}\) is the single-particle update matrix on the fixed edge \(e\) as defined in \cref{eq:def-A-e-p-non-parallel}.
By \cref{eq:lkmpchain-two-particles} and the definition of \(C\), we have
\[
    \lkmpchain^{(2,G)}=\frac{1}{\abs{E}}\sum_{e\in E}A_e,
    \qquad
    C=\frac{1}{\abs{E}}\sum_{e\in E}C_e\otimes C_e \enspace.
\]
Therefore, it suffices to prove that
\(
A_e|_{\mathcal H_{\asym}}=(C_e\otimes C_e)|_{\mathcal H_{\asym}}
\)
for every edge \(e\).

Fix \(e=\{a,b\}\), and let \(f\in\mathcal H_{\asym}\).
We show that \(A_ef=(C_e\otimes C_e)f\) pointwise on \(V\times V\).

\begin{itemize}
    \item If $y,z\notin\{a,b\}$, then neither particle is affected by either dynamics. Hence
\[
(A_e f)(y,z)=f(y,z)=((C_e\otimes C_e)f)(y,z)\enspace.
\]
\item If $y\in\{a,b\}$ and $z\notin\{a,b\}$, then by \cref{lem:action-of-Ae},
\[
(A_e f)(y,z)=\frac{f(a,z)+f(b,z)}{2}.
\]
On the other hand, under $C_e\otimes C_e$, the first coordinate is averaged according to $C_e$,
while the second coordinate remains fixed, so
\[
((C_e\otimes C_e)f)(y,z)=\frac{f(a,z)+f(b,z)}{2}.
\]
The case $y\notin\{a,b\}$ and $z\in\{a,b\}$ is identical.
\item 
Finally, if $y,z\in\{a,b\}$, then by the antisymmetric property,
\[
f(a,a)=f(b,b)=0,\qquad f(b,a)=-f(a,b).
\]
Hence \cref{lem:action-of-Ae} gives
\[
(A_e f)(y,z)
=
\frac13 f(a,a)+\frac16 f(a,b)+\frac16 f(b,a)+\frac13 f(b,b)
=
0.
\]
On the other hand,
\[
((C_e\otimes C_e)f)(y,z)
=
\frac14\bigl(f(a,a)+f(a,b)+f(b,a)+f(b,b)\bigr)
=
0.
\]
\end{itemize}
Thus
\[
A_e|_{\mathcal H_{\asym}}
=
(C_e\otimes C_e)|_{\mathcal H_{\asym}}.
\]
Averaging over $e$ yields the claim.
\end{proof}

Combining the symmetric and antisymmetric analyses,
we obtain a corresponding decomposition of the spectrum of \(\lkmpchain^{(2,G)}\).
Let \(\lambda_1(M)\) and \(\lambda_2(M)\) denote the largest and second largest eigenvalue in absolute value of a matrix \(M\), respectively.

\begin{proposition}
\label{prop:lkmp-spectrum-decomposition}
We have
\[
    \lambda_2(\lkmpchain^{(2,G)})
    =
    \max\st{
        \lambda_2(\kmpchain^{(2,G)}),
        \lambda_1(C|_{\mathcal H_{\asym}})
    }
    \enspace.
\]
\end{proposition}

\begin{proof}
By \cref{lem:lkmp-commutes-swap}, the spaces \(\mathcal H_{\sym}\) and \(\mathcal H_{\asym}\)
are both invariant under \(\lkmpchain^{(2,G)}\).
Therefore, the spectrum of \(\lkmpchain^{(2,G)}\) is the union of the spectra of its restrictions to these two subspaces.
By \cref{lem:symmetric-subspace-kmp}, the restriction
\(\lkmpchain^{(2,G)}|_{\mathcal H_{\sym}}\) has the same eigenvalues as
\(\kmpchain^{(2,G)}\).
By \cref{lem:antisymmetric-subspace-rw}, the restriction
\(\lkmpchain^{(2,G)}|_{\mathcal H_{\asym}}\) is exactly
\(C|_{\mathcal H_{\asym}}\).
Note that the eigenspace corresponding to the eigenvalue \(1\)
lies in the symmetric subspace,
since the constant function is symmetric.
Therefore,
every \(f\in\mathcal H_{\asym}\) is orthogonal to the all-ones vector
and the eigenvalue \(1\) does not appear in the restriction \(C|_{\mathcal H_{\asym}}\).

As a result, the second largest eigenvalue in absolute value of \(\lkmpchain^{(2,G)}\) is obtained by taking the maximum of the largest nontrivial eigenvalue from the symmetric block and the largest eigenvalue from the antisymmetric block:
\[
    \lambda_2(\lkmpchain^{(2,G)})
    =
    \max\st{
        \lambda_2(\kmpchain^{(2,G)}),
        \lambda_1(C|_{\mathcal H_{\asym}})
    }
    \enspace.\qedhere
\]
\end{proof}

\subsection{Comparison with the Single-Particle KMP Chain}

Both terms appearing in \cref{prop:lkmp-spectrum-decomposition} can be related to the single-particle KMP chain.
More precisely, the spectral gap of \(\kmpchain^{(2,G)}\) is bounded below by a constant multiple of the spectral gap of \(\kmpchain^{(1,G)}\),
while the largest eigenvalue of the antisymmetric block \(C|_{\mathcal H_{\asym}}\) coincides with \(\lambda_2(\kmpchain^{(1,G)})\).
We prove these two relations below.

\begin{restatable}{theorem}{sp}
\label{thm:two-particle-gap-compared-to-single-particle}
For every connected graph \(G=(V=[d],E)\),
\[
    1-\lambda_2\!\br{\kmpchain^{(2,G)}}
    \ge
    \frac{2(d+1)}{3d}\cdot
    \br{1-\lambda_2\!\br{\kmpchain^{(1,G)}}} \enspace.
\]
\end{restatable}

\begin{proof}
Let
\[
 \mathrm{gap}\bigl(\kmpchain^{(1,G)}\bigr)
=
1-\lambda_2\bigl(\kmpchain^{(1,G)}\bigr),
\qquad
 \mathrm{gap}\bigl(\kmpchain^{(2,G)}\bigr)
=
1-\lambda_2\bigl(\kmpchain^{(2,G)}\bigr)\enspace.
\]
By \cref{lem:continuous-gap-comparison-kim} with $k=2$, we have
\[
 \mathrm{gap}_2\bigl(N^{(G)}\bigr)
\ge
\frac{2(d+1)}{3d}\,
 \mathrm{gap}_1\bigl(N^{(G)}\bigr)\enspace.
\]
By \cref{lem:isomorphism} with $t=1$ and $t=2$, the restrictions
$N^{(G)}|_{\PP_1}$ and $N^{(G)}|_{\PP_2}$ have the same eigenvalues as
$\kmpchain^{(1,G)}$ and $\kmpchain^{(2,G)}$, respectively. Hence
\[
 \mathrm{gap}_1\bigl(N^{(G)}\bigr)
=
 \mathrm{gap}\bigl(\kmpchain^{(1,G)}\bigr),
\qquad
 \mathrm{gap}_2\bigl(N^{(G)}\bigr)
=
 \mathrm{gap}\bigl(\kmpchain^{(2,G)}\bigr)\enspace.
\]
Therefore,
\[
 \mathrm{gap}\bigl(\kmpchain^{(2,G)}\bigr)
=
 \mathrm{gap}_2\bigl(N^{(G)}\bigr)
\ge
\frac{2(d+1)}{3d}\,
 \mathrm{gap}_1\bigl(N^{(G)}\bigr)
=
\frac{2(d+1)}{3d}\,
 \mathrm{gap}\bigl(\kmpchain^{(1,G)}\bigr)\enspace.
\]
Equivalently,
\[
1-\lambda_2\bigl(\kmpchain^{(2,G)}\bigr)
\ge
\frac{2(d+1)}{3d}\,
\Bigl(1-\lambda_2\bigl(\kmpchain^{(1,G)}\bigr)\Bigr)\enspace.\qedhere
\]
\end{proof}

\begin{restatable}{lemma}{spc}
\label{lem:two-particle-gap-compared-to-single-particle-c}
For every connected graph \(G=(V=[d],E)\),
\[
    \lambda_1\!\br{C|_{\mathcal H_{\asym}}} = \lambda_2\!\br{\kmpchain^{(1,G)}} \enspace.
\]
\end{restatable}

\begin{proof}
    Note that \(C\) and \(\kmpchain^{(1,G)}\) are positive semidefinite.
    Let $e_i$ be the row vector whose $i$-th entry is $1$ and all other entries are $0$.
    For each edge \(e=\{a,b\}\in E\), let
    \[
        L_e \coloneqq (e_a-e_b)^{\intercal}(e_a-e_b) \enspace.
    \]
    Then, we can write \(C_e\) as
    \begin{align} \label{eq:Ce-Laplacian}
        C_e = I - \frac{1}{2}L_e \enspace,
    \end{align}
    and hence we have
    \begin{align} \label{eq:chain-1-average}
        \kmpchain^{(1,G)} = \frac{1}{|E|}\sum_{e\in E} C_e
        =
        I - \frac{1}{2|E|}\sum_{e\in E} L_e \enspace.
    \end{align}

    Now we choose an orthonormal eigenbasis \(\{v_1,\dots,v_d\}\) of \(\kmpchain^{(1,G)}\), with
    \[
        \kmpchain^{(1,G)}\cdot v_i=\theta_i \cdot v_i
        \qquad\text{for each } i\in[d],
    \]
    where
    \(
        1=\theta_1\ge \theta_2\ge \cdots \ge \theta_d \ge 0.
    \)
    For each \(1\le i<j\le d\), define
    \[
         u_{i,j}\coloneqq v_i\otimes v_j-v_j\otimes v_i \in \R^{V\times V} \enspace.
    \]
    The vectors \(\{ u_{i,j} : 1\le i<j\le d\}\) form a basis of \(\mathcal H_{\asym}\).
    We next show that each \( u_{i,j}\) is an eigenvector of \(C\).
    For each edge \(e=\{a,b\}\in E\), we have
    \begin{align*}
        (L_e\otimes L_e)\cdot u_{i,j}
        &=
        (L_ev_i)\otimes(L_ev_j)-(L_ev_j)\otimes(L_ev_i) \\
        &=
        \br{(e_a-e_b)v_i}\br{(e_a-e_b)v_j}
        \Bigl(
            (e_a-e_b)^{\intercal}\otimes(e_a-e_b)^{\intercal}
            -
            (e_a-e_b)^{\intercal}\otimes(e_a-e_b)^{\intercal}
        \Bigr) \\
        &= 0 \enspace.
    \end{align*}
    Hence, because \(C_e = I - \frac{1}{2}L_e\) in \cref{eq:Ce-Laplacian}, we have
    \[
        (C_e\otimes C_e)\cdot u_{i,j}
        =
         u_{i,j}
        -\frac12 (L_e\otimes I)\cdot u_{i,j}
        -\frac12 (I\otimes L_e)\cdot u_{i,j} \enspace.
    \]
    Averaging over \(e\in E\), and using
    \(
        \frac{1}{2|E|}\sum_{e\in E}L_e = I-\kmpchain^{(1,G)}
    \)
    by \cref{eq:chain-1-average},
    we obtain
    \begin{align*}
        C\cdot u_{i,j}
        &=
        \frac{1}{|E|}\sum_{e\in E}(C_e\otimes C_e)\cdot u_{i,j} \\
        &=
        u_{i,j}
        -
        \br{(I-\kmpchain^{(1,G)})\otimes I}\cdot u_{i,j}
        -
        \br{I\otimes (I-\kmpchain^{(1,G)})}\cdot u_{i,j} \\
        &=
        \br{\kmpchain^{(1,G)}\otimes I}\cdot u_{i,j}
        +
        \br{I\otimes \kmpchain^{(1,G)}}\cdot u_{i,j}
        -
        u_{i,j} \\
        &=
        (\theta_i+\theta_j-1)\cdot u_{i,j} \enspace.
    \end{align*}
    Therefore, the eigenvalues of \(C|_{\mathcal H_{\asym}}\) are exactly
    \[
        \set{
            \theta_i+\theta_j-1 : 1\le i<j\le d
        } \enspace.
    \]
    Since
    \(
        1=\theta_1\ge \theta_2\ge \cdots \ge \theta_d,
    \)
    the largest of these values is attained at \((i,j)=(1,2)\). Thus
    \[
        \lambda_1\!\br{C|_{\mathcal H_{\asym}}}
        =
        \theta_2
        =
        \lambda_2\!\br{\kmpchain^{(1,G)}} \enspace.\qedhere
    \]
\end{proof}

\subsection{Proof of \texorpdfstring{%
\cref{cor:two-particle-labeled-gap-compared-to-single-particle}}{%
Proposition \ref*{cor:two-particle-labeled-gap-compared-to-single-particle}}}
We can now prove \cref{cor:two-particle-labeled-gap-compared-to-single-particle},
which we restate below.

\twotoone*

\begin{proof}
    By \cref{prop:lkmp-spectrum-decomposition}, we have
    \[
        \lambda_2(\lkmpchain^{(2,G)})
        =
        \max\st{
            \lambda_2(\kmpchain^{(2,G)}),
            \lambda_1(C|_{\mathcal H_{\asym}})
        } \enspace.
    \]
    By \cref{thm:two-particle-gap-compared-to-single-particle} and \cref{lem:two-particle-gap-compared-to-single-particle-c}, we have
    \[
        1-\lambda_2\!\br{\lkmpchain^{(2,G)}}
        \ge
        \min\st{
            1-\lambda_2\!\br{\kmpchain^{(2,G)}},
            1-\lambda_1\!\br{C|_{\mathcal H_{\asym}}}
        }
        \ge
        \frac{2}{3}\cdot
        \br{1-\lambda_2\!\br{\kmpchain^{(1,G)}}} \enspace.\qedhere
    \]
\end{proof}

\section{The Parallel KMP Chain}
\label[appendix]{sec:parallel-kmp}
This appendix proves \cref{thm:partition-resampling-mixing}
by adapting the techniques from \cref{sec:mixing-kmp}
to the parallel setting.
We first formalize the lumping relation between the parallel KMP chain
and the partition resampling chain in \cref{sec:parallel-kmp-prc}.
We then analyze the \emph{labeled} parallel KMP chain through its
conditional product structure in \cref{sec:labeled-chain}.
Finally, in \cref{sec:parallel-uniform-mixing}, we combine this analysis
with the coupling method from \cref{sec:kmp-exact-coupling}
to obtain a mixing bound independent of the number of particles
for both parallel KMP chain and partition resampling chain.

\subsection{Lumping to the Partition Resampling Chain}
\label[appendix]{sec:parallel-kmp-prc}
We first make precise the relation between the lifted chain and
the original partition resampling chain.
Recall that the partition resampling chain has the transition matrix \(M\) and the stationary distribution \(\pi\),
as in
\cref{prop:markov-chain-transition-matrix-correspondence,lem:property-partition-resampling-chain}.

\begin{lemma}\label{lem:partition-dominates-by-lifted}
Let $\lambda\in \partition{t}{d}$, and let
$r_\lambda\in\Delta$ be obtained from $\lambda$ by padding zeros if necessary,
as in \cref{def:partition-resampling-chain}.
Then for every \(k\ge 0\),
\[
    \bigl\|
        M^k(\lambda,\cdot)-\pi
    \bigr\|_{\mathrm{TV}}
    \le
    \bigl\|
        \liftedchain^k(r_\lambda,\cdot)-\lifteddist
    \bigr\|_{\mathrm{TV}} \enspace.
\]
\end{lemma}

\begin{proof}
    Let \(g:\Delta\to\partition{t}{d}\) be the map
    \[
        g(r)\coloneqq r^{\downarrow} \enspace,
    \]
    where $r^{\darrow}$ is the vector
    obtained by rearranging the entries of $r$ in non-increasing order and removing all zeros.
    For a distribution \(\mu\) on \(\Delta\), write \(\mu g^{-1}\) for the induced distribution on \(\partition{t}{d}\), namely
    \[
        \mu g^{-1}(\lambda)
        =
        \sum_{r\in\Delta:\,g(r)=\lambda}\mu(r)
        =
        \sum_{r\in \Xi_\lambda} \mu(r)
        \qquad \text{for every } \lambda\in\partition{t}{d} \enspace,
    \]
    where \(\Xi_\lambda\) is defined in \cref{eq:xi-lambda}.
    
    We first claim that the partition resampling chain is the lumping of the parallel KMP chain under the map
    \(
        g(r)=r^\downarrow
    \).
    More precisely, for every \(\lambda\in\partition{t}{d}\),
    \[
        \liftedchain(r_\lambda,\cdot)\,g^{-1}
        =
        M(\lambda,\cdot) \enspace.
    \]
    Fix \(r\in\Delta\), let \(\lambda=r^\downarrow\), and let \(\lambda'\in\partition{t}{d}\).
    By \cref{prop:markov-chain-transition-matrix-correspondence} and 
    the definition of the partition resampling chain,
    \[
    	M(\lambda,\lambda') =
        \Pr\Br{ R' \in \Xi_{\lambda'} | R = r_\lambda } =
        \sum_{r' \in \Xi_{\lambda'}} \Pr\Br{ R' = r' | R = r_\lambda} \enspace.
    \]
    Therefore, we have
    \begin{align*}
        M(\lambda,\lambda')
        =~ \sum_{r'\in \Xi_{\lambda'}}
        \liftedchain(r_\lambda,r') 
        =~ \bigl(\liftedchain(r_\lambda,\cdot)\,g^{-1}\bigr)(\lambda') \enspace.
    \end{align*}
    Since this holds for every \(\lambda'\in\partition{t}{d}\), we conclude that
    \[
        \liftedchain(r_\lambda,\cdot)\,g^{-1}
        =
        M(\lambda,\cdot) \enspace.
    \]

    We next extend this identity to \(k\) steps. We prove by induction on \(k\) that
    \[
        \liftedchain^k(r_\lambda,\cdot)\,g^{-1}
        =
        M^k(\lambda,\cdot)
        \qquad \forall\,k\ge 1 \enspace.
    \]
    The case \(k=1\) is established above.
    Now assume that the identity holds for some \(k\ge1\).
    Fix any \(\lambda'\in\partition{t}{d}\). Then
    \begin{align}\label{eq:induction-step}
        \bigl(\liftedchain^{k+1}(r_\lambda,\cdot)\,g^{-1}\bigr)(\lambda')
        =~ \sum_{r'\in \Xi_{\lambda'}} \liftedchain^{k+1}(r_\lambda,r') \nonumber
        =~& \sum_{r'\in \Xi_{\lambda'}} \sum_{r\in\Delta}
        \liftedchain^k(r_\lambda,r) \cdot \liftedchain(r,r')\nonumber\\
        =~& \sum_{r\in\Delta}\liftedchain^k(r_\lambda,r)
        \sum_{r'\in \Xi_{\lambda'}} \liftedchain(r,r') \enspace.
    \end{align}
    Now fix \(r\in\Delta\), and let \(\mu=g(r)=r^\downarrow\).
    Let \(\tau\in\S_d\) be any permutation such that
    \[
        r=\tau\cdot r_\mu \enspace,
    \]
    where $r_\mu$ is the vector obtained from $\mu$ by padding zeros if necessary.
    By the definition of the lifted chain,
    \[
        \liftedchain(\tau\cdot u,\tau\cdot v)=\liftedchain(u,v)
        \qquad \forall\,u,v\in\Delta \enspace.
    \]
    Therefore,
    \begin{align*}
        \sum_{r'\in \Xi_{\lambda'}} \liftedchain(r,r')
        = \sum_{r'\in \Xi_{\lambda'}} \liftedchain(\tau\cdot r_\mu,r') 
        = \sum_{r'\in \Xi_{\lambda'}} \liftedchain(r_\mu,\tau^{-1}\cdot r') 
        = \sum_{r'\in \Xi_{\lambda'}} \liftedchain(r_\mu,r') \enspace.
    \end{align*}
    By the one-step identity proved above, we have
    \[
        \sum_{r'\in \Xi_{\lambda'}} \liftedchain(r,r')
        =
        M(\mu,\lambda') \enspace.
    \]
    Therefore, inserting this into \cref{eq:induction-step}, we obtain
    \begin{align*}
        \bigl(\liftedchain^{k+1}(r_\lambda,\cdot)\,g^{-1}\bigr)(\lambda')
        =~& \sum_{r\in\Delta}\liftedchain^k(r_\lambda,r)\cdot M(g(r),\lambda') \\
        =~& \sum_{\mu\in\partition{t}{d}}
        \sum_{r\in\Xi_{\mu}}
        \liftedchain^k(r_\lambda,r)\cdot M(\mu,\lambda') \\
        =~& \sum_{\mu\in\partition{t}{d}}
        \bigl(\liftedchain^k(r_\lambda,\cdot)\,g^{-1}\bigr)(\mu)\cdot M(\mu,\lambda') \\
        =~& \sum_{\mu\in\partition{t}{d}}
        M^k(\lambda,\mu)\cdot M(\mu,\lambda') \\
        =~& M^{k+1}(\lambda,\lambda') \enspace,
    \end{align*}
    where in the fourth equality we used the induction hypothesis.
    Since this holds for every \(\lambda'\in\partition{t}{d}\), we conclude that
    \[
        \liftedchain^{k+1}(r_\lambda,\cdot)\,g^{-1}
        =
        M^{k+1}(\lambda,\cdot) \enspace.
    \]
    This completes the induction.

    Moreover, for every \(\lambda\in\partition{t}{d}\),
    \[
        \bigl(\lifteddist\,g^{-1}\bigr)(\lambda)
        =
        \sum_{r\in\Xi_\lambda}\lifteddist(r)
        =
        \sum_{r\in\Xi_\lambda}\frac{1}{|\Delta|}
        =
        \frac{|\Xi_{\lambda}|}{\binom{t+d-1}{d-1}}
        =
        \pi(\lambda) \enspace.
    \]
    Hence \(\lifteddist\,g^{-1}=\pi\).
    Therefore, by \cref{fact:tv-distance-function},
    \[
        \bigl\|
            M^k(\lambda,\cdot)-\pi
        \bigr\|_{\mathrm{TV}}
        =
        \bigl\|
            \liftedchain^k(r_\lambda,\cdot)\,g^{-1}-\lifteddist\,g^{-1}
        \bigr\|_{\mathrm{TV}}
        \le
        \bigl\|
            \liftedchain^k(r_\lambda,\cdot)-\lifteddist
        \bigr\|_{\mathrm{TV}} \enspace. \qedhere
    \]
\end{proof}

\subsection{The Labeled Parallel KMP Chain}
\label[appendix]{sec:labeled-chain}

To analyze the mixing time of the parallel KMP chain,
we introduce a labeled version of the parallel KMP chain
with auxiliary randomness.
We then exploit its conditional product structure
to bound its mixing time, as in \cref{sec:mixing-kmp}.

\begin{definition}[Labeled parallel KMP chain with auxiliary randomness]
\label{def:labeled-parallel-kmp-anc}
Fix a positive integer \(t\) and an even integer \(d\geq2\), and let
\(
    \Omega \coloneqq [d]^t
\)
be the space of labeled configurations of \(t\) particles on \(d\) sites.
For \(x=(x_1,\dots,x_t)\in\Omega\),
the coordinate \(x_\ell\in[d]\) records the location of particle \(\ell\).

The labeled parallel KMP process is the Markov chain
\(\{x^{(k)}\}_{k\ge 0}\) on \(\Omega\) defined as follows.
Given the current configuration \(x^{(k)}\), the next configuration
\(x^{(k+1)}\) is obtained by the procedure below.

\begin{enumerate}
    \item Sample a uniformly random perfect matching of \([d]\), denoted by
        \[
            \mathcal{M} = \{\{i_1,j_1\},\dots,\{i_m,j_m\}\}\ , \qquad m = d/2\enspace.
        \]  

    \item For each matched pair \(\{i,j\}\in\mathcal M\), independently sample a uniform random variable
        \[
        	p_{i,j} \sim \mathrm{Unif}[0,1] \enspace.
        \]
    Then, among all particles currently located at sites \(i\)
    and \(j\), choose each particle independently with probability \(p_{i,j}\) and place it at site \(i\); place all remaining particles at site \(j\).
\end{enumerate}
We denote the transition matrix of this labeled chain by $\labeledchain^{(t,d)}$ and omit the superscripts when they are clear from the context.
\end{definition}

Similarly to \cref{lem:labeled-kmp-chain-stationary},
the labeled parallel KMP chain has the following stationary
distribution, where \(\#_i(x)\) is defined in
\cref{eq:counting-function}.
The proof is analogous and is omitted here.

\begin{lemma}
\label{lem:labeled-chain-stationary}
The labeled parallel KMP chain is irreducible and aperiodic.
Moreover, it is reversible with respect to the probability distribution
\(\labeleddist\) on \(\Omega\) defined by
\[
    \labeleddist(x)
    \coloneqq
    \frac{\prod_{i=1}^d \#_i(x)!}{t!\cdot\binom{t+d-1}{d-1}}
    \qquad \text{for every } x\in\Omega \enspace.
\]
In particular, \(\labeleddist\) is the unique stationary
distribution of \(\labeledchain\).
\end{lemma}

By the same argument as in \cref{lem:labeled-dominates-lifted-general-kmp},
the parallel KMP chain mixes at least as fast as its labeled version.
Formally, we have the following lemma.

\begin{lemma}
\label{lem:labeled-dominates-lifted}
Let \(x\in\Omega\), and let
\(
    r=(\#_1(x),\dots,\#_d(x))\in\Delta.
\)
Then for every \(k\ge 0\),
\[
    \Vert \liftedchain^k(r,\cdot)-\lifteddist \Vert_{\mathrm{TV}}
    \le
    \Vert \labeledchain^k(x,\cdot)-\labeleddist \Vert_{\mathrm{TV}} \enspace.
\]
\end{lemma}

The labeled parallel KMP chain has the same conditional product structure
as in \cref{sec:labeled-kmp-chain-conditional-structure}.

\paragraph{Conditional Product Structure.}
  Let us fix
    an initial state \(x^{(0)}\in\Omega\),
    the perfect matchings
    \(\{\mathcal M_s\}_{1\le s\le k}\)
    and the uniform random variables
    \(\{p_{i,j}^{(s)}\}_{1\le s\le k,\ \{i,j\}\in\mathcal M_s}\)
    used in the first \(k\) steps of the labeled chain.
    Similarly to \cref{eq:def-A-non-parallel,eq:def-B-non-parallel},
    for each \(1\le s\le k\), define a random row-stochastic matrix
    \(A_s\in\R^{d\times d}\) by setting
    \begin{equation}\label{eq:def-A}
        A_s(i,i) = A_s(j,i) = p_{i,j}^{(s)} \enspace,
        \qquad
        A_s(i,j) = A_s(j,j) = 1-p_{i,j}^{(s)} \enspace,
    \end{equation}
    for each pair \(\{i,j\}\in\mathcal M_s\),
    with all other entries equal to zero,
    and define
    \begin{equation}\label{eq:def-B}
        B_s := A_1 \cdot A_2 \cdots A_s \enspace.
    \end{equation}

\begin{remark}\label{rmk:transition-matrix-labeled-chain}
    Let $A$ be the random row-stochastic matrix defined as in \cref{eq:def-A} using independent $\Unif[0,1]$ variables $\{p_{i,j}\}_{\{i,j\}\in\mathcal M}$ and a random perfect matching $\mathcal M$.
    Similar to \cref{rmk:transition-matrix-labeled-chain-non-parallel},
    the transition matrix of the labeled chain can be expressed as
    \[
        \labeledchain
        =
        \E_{\mathcal M, \{p_{i,j}\}}[A^{\otimes t}] \enspace.
    \]
\end{remark}

Similarly to \cref{obs:conditional-product-structure,lem:labeled-kmp-stationary-product-structure},
we have the conditional product structure:
\begin{itemize}
\item Let $e_i$ be the row vector in $\R^d$ with a $1$ in the $i$-th coordinate and $0$ elsewhere.
    We have that for any $y\in\Omega$,
\begin{equation}\label{eq:conditional-product-structure-parallel}
    \labeledchain^k(x^{(0)},y)
    =
    \E_{\substack{\{\mathcal M_s\},\,\{p_{i,j}^{(s)}\}}}
    [\mu_k(y)] \enspace,
\end{equation}
    where $\mu_k
    \coloneq
    \bigotimes_{\ell=1}^t q_{x_\ell^{(0)}}^{(k)}$
    and $q_i^{(k)} = e_i\cdot B_k$.
\item Let $u$ be a random vector drawn from $\Dirichlet(1,\dots,1)$,
    independently of the matchings and split parameters.
    Then for any $y\in\Omega$,
\begin{equation}
\label{eq:labeled-stationary-product-structure-parallel}
    \labeleddist(y)
    =
    \E_{\{\mathcal M_s\},\,\{p_{i,j}^{(s)}\},\,u}
    \Br{\prod_{i=1}^t w^{(k)}(y_i)} \enspace,
\end{equation}
    where $w^{(k)} = u\cdot B_k$.
\end{itemize}

As in \cref{sec:mixing-time-labeled-kmp-chain},
by the conditional product structure,
it suffices to study the contraction of the single-particle
transition matrix \(A\).
In the parallel setting, the following lemma gives an explicit
one-step contraction under the random matrix \(A\)
defined in \cref{eq:def-A}.

\begin{lemma}\label{lem:contraction-under-A}
    Let \(v\in\R^d\) be a vector satisfying
    \(
        \sum_{j=1}^d v_j = 0.
    \)
    Let \(A\) be the random row-stochastic matrix
    defined from a uniformly random perfect matching \(\mathcal M\) of \([d]\)
    and i.i.d.\ $\Unif[0,1]$ random variables \(\{p_{i,j}\}\) as in \cref{eq:def-A}.
    We have
    \[
        \E_{A}\Br{\norm{vA}_2^2}
        =
        \frac{2(d-2)}{3(d-1)} \norm{v}_2^2 \enspace.
    \]
\end{lemma}

\begin{proof}
    Let \(v\in\R^d\) be a vector such that
    \(
        \sum_{j=1}^d v_j = 0.
    \)
    If \(\{i,j\}\in\mathcal M\), then
    \[
        (vA)_i = p_{i,j}\cdot (v_i+v_j)\enspace,
        \qquad
        (vA)_j = (1-p_{i,j})\cdot (v_i+v_j)\enspace.
    \]
    Hence, conditioning on the matching \(\mathcal M\), we have
    \[
        \E_{\{p_{i,j}\}}\Br{\norm{vA}_2^2}
        =
        \sum_{\{i,j\}\in\mathcal M}
        \E_{\{p_{i,j}\}}[(p_{i,j}^2+(1-p_{i,j})^2)\cdot (v_i+v_j)^2] 
        =
        \frac{2}{3}\sum_{\{i,j\}\in\mathcal M}(v_i+v_j)^2
        \enspace,
    \]
    where we used the fact that for \(p_{i,j}\sim\mathrm{Unif}[0,1]\),
    \(
        \E[p_{i,j}^2+(1-p_{i,j})^2]
        =
        \frac{2}{3} .
    \)
    Averaging over the uniformly random perfect matching \(\mathcal M\),
    we obtain
    \begin{align*}
        \E_{\mathcal{M}, \{p_{i,j}\}} \Br{\norm{vA}_2^2}
        =
        \frac{2}{3} \cdot \E_{\mathcal M}\Br{\sum_{\{i,j\}\in\mathcal M}(v_i+v_j)^2}
        &=
        \frac{2}{3}\cdot
        \br{\sum_{i=1}^d v_i^2+\frac{2}{d-1}\sum_{i<j}v_i\cdot v_j} \\
        &=
        \frac{2}{3}\cdot\norm{v}_2^2
        +\frac{4}{3(d-1)}\sum_{i<j} v_i\cdot v_j\\
        &=
        \frac{2}{3}\cdot\norm{v}_2^2
        +\frac{2}{3(d-1)}\br{ \br{\sum_{i=1}^d v_i }^2 - \sum_{i=1}^d v_i^2} \\
        &=
        \frac{2(d-2)}{3(d-1)}\norm{v}_2^2 \enspace,
    \end{align*}
    where in the last step we used \(\sum_{i=1}^d v_i=0\).
\end{proof}

\paragraph{Mixing Time Analysis.} 
We now analyze the mixing time of the labeled parallel KMP chain. 

\begin{theorem}
\label{thm:labeled-chain-mixing}
For an even integer \(d\geq2\) and positive integers \(t\) and \(k\),
let $\{ x^{(s)} \in \Omega \}_{s\geq 0}$
be the labeled parallel KMP chain with transition matrix $\labeledchain$
in \cref{def:labeled-parallel-kmp-anc}.
Then, for any initial state \(x^{(0)} \in \Omega\),
we have
\[
    \norm{\labeledchain^k(x^{(0)},\cdot)-\labeleddist}_{\mathrm{TV}}
    \le
    \frac{t\sqrt{d}}{2} \cdot \br{\frac23}^{k/2} \enspace ,
\]
where \(\labeleddist\) is the stationary distribution
given in \cref{lem:labeled-chain-stationary}.
\end{theorem}

\begin{proof}
    The case \(d=2\) is immediate:
    there is only one perfect matching,
    and after one step the chain samples
    a uniform split of the \(t\) labeled particles
    between the two sites,
    which is exactly the stationary distribution \(\labeleddist\).

    Assume now that \(d\ge 4\), and fix an initial state
    \(x^{(0)}\in\Omega\).
    Recall the conditional product structure in \cref{eq:conditional-product-structure-parallel,eq:labeled-stationary-product-structure-parallel}.
    Combining \cref{eq:conditional-product-structure-parallel,eq:labeled-stationary-product-structure-parallel}, we have
    \[
        \norm{\labeledchain^k(x^{(0)},\cdot)-\labeleddist}_{\mathrm{TV}}
        \le
        \E_{
            \substack{\{\mathcal M_s\},\,\{p_{i,j}^{(s)}\},\,u}    
        }
        \Br{\norm{\mu_k-(w^{(k)})^{\otimes t}}_{\mathrm{TV}}} \enspace.
    \]
    Using the tensor product structure, by \cref{fact:tv-tensor} we can further bound the right-hand side by
    \begin{align}
        \norm{\labeledchain^k(x^{(0)},\cdot)-\labeleddist}_{\mathrm{TV}}
        &\le~
        \E_{
            \substack{\{\mathcal M_s\},\,\{p_{i,j}^{(s)}\},\,u}    
        }
        \Br{\sum_{i=1}^d \#_i(x^{(0)}) \cdot \Vert q_i^{(k)}-w^{(k)} \Vert_{\mathrm{TV}}} \nonumber\\
        &=~
        \sum_{i=1}^d \#_i(x^{(0)}) \cdot \E_{
            \substack{\{\mathcal M_s\},\,\{p_{i,j}^{(s)}\},\,u}    
        }
        \Br{\Vert q_i^{(k)}-w^{(k)} \Vert_{\mathrm{TV}}} \nonumber\\
        &\le
        \sum_{i=1}^d \#_i(x^{(0)}) \cdot \frac{\sqrt{d}}{2} \cdot \E_{    
            \substack{\{\mathcal M_s\},\,\{p_{i,j}^{(s)}\},\,u}    
        }
        \Br{\Vert q_i^{(k)}-w^{(k)} \Vert_{2}}  \enspace. \label{eq:tv-to-l2}
    \end{align}
    Since $\sum_{i=1}^d \#_i(x^{(0)}) = t$, to prove the theorem, it suffices to show that for every \(i\),
    \[
        \E_{
            \substack{\{\mathcal M_s\},\ \{p_{i,j}^{(s)}\},\ u}}
        \Br{\Vert q_i^{(k)}-w^{(k)} \Vert_{2}} \leq \br{\frac23}^{k/2} \enspace.
    \]
    For every \(i\in[d]\), we have
    \(
        q_i^{(k)}-w^{(k)}
        =
        (e_i-u)\cdot B_k.
    \)
    Since
    \(
        \sum_{j=1}^d (e_i-u)_j = 1-\sum_{j=1}^d u_j = 0,
    \)
    the vector \(e_i-u\) has sum zero.
    Therefore, by \cref{lem:contraction-under-A},
    \begin{align*}
        \E
        \Br{\Vert q_i^{(k)}-w^{(k)} \Vert_2^2}
        =
        \E\Br{\norm{(e_i-u)\cdot B_k}_2^2}
        &=
        \E_{u,A_1,\dots,A_k}\Br{\norm{(e_i-u)\cdot A_1\cdots A_k}_2^2}\\
        &=
        \br{\frac{2(d-2)}{3(d-1)}}^k
        \E\Br{\norm{e_i-u}_2^2} \enspace.
    \end{align*}
    Note that since \(u_i\sim\btdist{1,d-1}\),
    \[
        \E\Br{\norm{e_i-u}_2^2}
        =
        1-2\E[u_i]+\sum_{j=1}^d \E[u_j^2]
        =
        1-\frac{2}{d}+\frac{2}{d+1}
        \le 1 \enspace.
    \]
    Therefore, 
    \begin{equation}\label{eq:vectors_close}
        \E
        \Br{\Vert q_i^{(k)}-w^{(k)} \Vert_2^2}\leq\br{\frac23}^k.
    \end{equation}
    By Jensen's inequality,
    \[
        \E
        \Br{\Vert q_i^{(k)}-w^{(k)} \Vert_{2}}
        \le
        \br{
            \E\Br{\Vert q_i^{(k)}-w^{(k)} \Vert_2^2}
        }^{1/2}
        \leq
        \br{\frac23}^{k/2} \enspace.
    \]
    Substituting this bound into \cref{eq:tv-to-l2}, we get
    \begin{align*}
        \norm{\labeledchain^k(x^{(0)},\cdot)-\labeleddist}_{\mathrm{TV}}
        &\le~
        \sum_{i=1}^d \#_i(x^{(0)})\cdot \frac{\sqrt d}{2}
        \cdot
        \br{\frac23}^{k/2}  =~ \frac{t\sqrt d}{2} \cdot
        \br{\frac23}^{k/2} \enspace.
    \end{align*}
    This proves the theorem.
\end{proof}

\subsection{Mixing Time of the Parallel KMP Chain and the Partition Resampling Chain}
\label[appendix]{sec:parallel-uniform-mixing}

We now prove \cref{thm:partition-resampling-mixing}.
As in \cref{sec:kmp-uniform-mixing}, we consider two cases
according to the number of particles.
The labeled mixing bound suffices when \(t\) is polynomially
bounded in \(d\).
For larger \(t\), we use a coupling argument to obtain
a mixing bound independent of \(t\).
Since each perfect matching can be viewed as a set of edges
of the complete graph, its parallel updates thus can be viewed
as a sequence of single-edge updates.
This allows us to apply \cref{lem:kmp-scheduled-coalescence}
and obtain the following lemma.

\begin{lemma}
\label{lem:parallel-kmp-coalescence}
Let \(d\geq4\) be even, \(t>d^{1000}\)
and set \(T=24\log d\).
Let \(R^{(0)}\) and \(\widetilde R^{(0)}\) be random vectors
in \(\Delta\) with an arbitrary joint distribution such that
\(\widetilde R^{(0)}\sim\Unif(\Delta)\), and define
\(\eta=\Pr[\|R^{(0)}-\widetilde R^{(0)}\|_1>td^{-19}]\).
Then there exists a coupling of two unlabeled parallel KMP chains
\((R^{(s)})_{s=0}^T\) and \((\widetilde R^{(s)})_{s=0}^T\)
with the given initial joint distribution, such that
\[
    \Pr\!\left[R^{(T)}\neq\widetilde R^{(T)}\right]
    ~\leq~\eta+\frac{3}{2d^2}\enspace.
\]
\end{lemma}

\begin{proof}
We first sample \(T\) independent uniform perfect matchings
\(\mathcal M_1,\ldots,\mathcal M_T\) of \([d]\).
Viewing each pair as an edge of the complete graph on \([d]\),
we list the edges within each matching in a fixed order,
following the order \(\mathcal M_1,\ldots,\mathcal M_T\),
to obtain \(e_1,\ldots,e_N\), where \(N=dT/2\).

Let \(\mathcal C\) be the event that the graph
\(([d],\{e_1,\ldots,e_N\})\) is connected.
For each matching sequence satisfying \(\mathcal C\),
we apply \cref{lem:kmp-scheduled-coalescence} to the edge sequence
\(e_1,\ldots,e_N\), with \(X^{(0)}=R^{(0)}\) and
\(\widetilde X^{(0)}=\widetilde R^{(0)}\).
We therefore obtain two coupled
parallel KMP chains by setting
\[
    R^{(k)}=X^{(kd/2)}
    \quad\text{and}\quad
    \widetilde R^{(k)}=\widetilde X^{(kd/2)}
    \quad \text{for}\quad 0\leq k\leq T\enspace.
\]
For matching sequences outside \(\mathcal C\), we use independent
redistributions in the two chains.
Therefore, applying \cref{lem:kmp-scheduled-coalescence}, we have
\begin{align*}
    \Pr\!\left[R^{(T)}\neq\widetilde R^{(T)}\right]
    &~\leq~\Pr[\mathcal C^{\mathsf c}]
      ~+~ \Pr\!\left[
           R^{(T)}\neq\widetilde R^{(T)}
           \,\middle|\,\mathcal C
       \right]\notag\\
    &~\leq~\Pr[\mathcal C^{\mathsf c}] ~+~ \eta ~+~ \frac{1}{d^2}\enspace.
\end{align*}

It remains to show that \(\Pr[\mathcal C^{\mathsf c}]\leq1/(2d^2)\).
For a fixed set \(S\subseteq[d]\) with \(1\leq |S|\leq d/2\),
we first bound the probability that a uniform perfect matching
contains no pair joining \(S\) to \([d]\setminus S\).
This probability is zero when \(|S|\) is odd.
When \(|S|\) is even, the probability that all vertices in \(S\)
are matched within \(S\) is
\[
    \prod_{j=0}^{|S|/2-1}
    \frac{|S|-2j-1}{d-2j-1}
    \leq\left(\frac{|S|}{d}\right)^{|S|/2}
    \leq2^{-|S|/2}\enspace.
\]
Since the \(T\) matchings are independent, the probability
that none contains a pair joining \(S\) to its complement
is at most \(2^{-|S|T/2}\).
If the graph is disconnected, some set \(S\) with
\(1\leq|S|\leq d/2\) has no edge to its complement.
Since \(T=24\log d\), we have \(2^{-T/2}=d^{-12}\).
A union bound over all such sets therefore gives
\begin{align*}
    \Pr[\mathcal C^{\mathsf c}]
    ~\leq~\sum_{k=1}^{d/2} \binom{d}{k} \cdot  2^{-kT/2}
    =\sum_{k=1}^{d/2} \binom{d}{k} \cdot  d^{-12k}
    ~\leq~(1+d^{-12})^d-1
    ~\leq~\frac{1}{2d^2}\enspace,
\end{align*}
where the last inequality follows from
\((1+d^{-12})^d-1\leq \e^{d^{-11}}-1\leq2d^{-11}\leq1/(2d^2)\),
using \(1+x\leq \e^x\), \(\e^x-1\leq2x\) for \(0\leq x\leq1\),
and \(d\geq4\).
This completes the proof.
\end{proof}

We now use the coupling method in
\cref{lem:parallel-kmp-coalescence} to establish a mixing time bound
independent of \(t\) for both the unlabeled parallel KMP chain
and the partition resampling chain.
We restate the theorem here for convenience.

\mixingtimeboundpkmp*

\begin{proof}
By \cref{lem:partition-dominates-by-lifted},
it suffices to bound the mixing time of the parallel KMP chain.
The case \(d=2\) holds trivially. Assume that \(d\geq4\), and set
\(k_0=6000\log d\), \(T=24\log d\), and \(\ell=k_0+T\).
We first show that
\begin{equation}
    \delta_{\scriptscriptstyle \mathsf{PKMP}}(\ell)
    :=\sup_{r\in\Delta}
      \|\liftedchain^{\ell}(r,\cdot)-\lifteddist\|_{\mathrm{TV}}
    \leq\frac{4}{d^2} \enspace.
    \label{eq:parallel-one-block}
\end{equation}
If \(t\leq d^{1000}\), then
\cref{lem:labeled-dominates-lifted,thm:labeled-chain-mixing} give
\[
    \delta_{\scriptscriptstyle \mathsf{PKMP}}(\ell)
    ~\leq~\delta_{\scriptscriptstyle \mathsf{PKMP}}(k_0)
    ~\leq~\tfrac12 d^{1000+1/2}(2/3)^{k_0/2}
    ~\leq~ d^{-2} \enspace.
\]

Now suppose that \(t>d^{1000}\). Consider two copies of the parallel KMP chain,
starting from \(R^{(0)}=r\) and \(\widetilde R^{(0)}\sim\Unif(\Delta)\), respectively.
We use a similar two-stage coupling construction as in the proof of \cref{lem:kmp-large-particle-one-block}.

\paragraph{The first stage.}
For the first \(k_0\) rounds, we construct the coupling by
assigning labels to the particles in each copy and using
the update rule in \cref{def:labeled-parallel-kmp-anc}.
Specifically, we use the same matchings
\(\mathcal M_s\) and split parameters \(p_{i,j}^{(s)}\) in both chains.

Recall the definitions of $A_s$ and $B_s$ from \cref{eq:def-A,eq:def-B}.
Let \(u\sim\Dirichlet(1,\ldots,1)\).
As in \cref{eq:kmp-conditional-mean,eq:kmp-stationary-conditional-mean},
we have that for every \(j\in[d]\),
\[
    \mathbb E[R_j^{(k_0)}]
    ~=~ \sum_{i=1}^d r_i \cdot B_{k_0}(i,j)
    ~=~ (rB_{k_0})_j
    \quad
    \text{and}
    \quad
    \mathbb E[\widetilde R_j^{(k_0)}]
    ~=~ t \cdot (uB_{k_0})_j \enspace,
\] 
where the expectation is conditional on \(u\), the random matchings and the random parameters.
Define \(a^{(k_0)}:=t^{-1}(rB_{k_0})\) and \(b^{(k_0)}:=uB_{k_0}\).
We therefore have \(\mathbb E[R^{(k_0)}/t]=a^{(k_0)}\) and
\(\mathbb E[\widetilde R^{(k_0)}/t]=b^{(k_0)}\).
Similarly to \cref{eq:kmp-conditional-mean-squared},
we use \cref{eq:vectors_close} to obtain
\begin{align}
    \mathop{\mathbb E}_{\substack{u,\{\mathcal M_s\},\{p_{i,j}^{(s)}\}}}
    \!\left[\|a^{(k_0)}-b^{(k_0)}\|_2^2\right]
    ~\leq~\sum_{i=1}^d\frac{r_i}{t}\cdot
    \mathop{\mathbb E}_{\substack{u,\{\mathcal M_s\},\{p_{i,j}^{(s)}\}}}
    \!\left[\|(e_i-u)B_{k_0}\|_2^2\right]
    ~\leq~\sum_{i=1}^d\frac{r_i}{t}
    \left(\frac{2}{3}\right)^{k_0}
    =\left(\frac{2}{3}\right)^{k_0}\,. \label{eq:parallel-kmp-mean-squared}
\end{align}
Thus, by Markov’s inequality, we obtain
\begin{align*}
    \prob{u,\{\mathcal M_s\},\{p_{i,j}^{(s)}\}}
    {\|a^{(k_0)}-b^{(k_0)}\|_1>\frac{1}{3}d^{-19}}
    &~\leq~ 9d^{38}\cdot
    \expect{u,\{\mathcal M_s\},\{p_{i,j}^{(s)}\}}
    {\|a^{(k_0)}-b^{(k_0)}\|_1^2}\\
    &~\leq~ 9d^{39}\cdot
    \expect{u,\{\mathcal M_s\},\{p_{i,j}^{(s)}\}}
    {\|a^{(k_0)}-b^{(k_0)}\|_2^2}\\
    &~\leq~ 9d^{39}\left(\frac{2}{3}\right)^{k_0}
    ~\leq~\frac{1}{2d^2}\enspace.
\end{align*}
By the same argument as for \cref{eq:kmp-large-particle-first-stage},
we therefore have
\begin{equation}
    \Pr\!\left[\|R^{(k_0)}-\widetilde R^{(k_0)}\|_1>t d^{-19}\right]
    \leq\frac{5}{2d^2} \enspace. \label{eq:parallel-kmp-diff}
\end{equation}

\paragraph{The second stage.}
We now apply \cref{lem:parallel-kmp-coalescence} with
\((R^{(k_0)},\widetilde R^{(k_0)})\) as the initial pair.
The lemma gives a coupling
for the next \(T\) rounds such that
\[
    \Pr\!\left[
        R^{(k_0+T)}\neq\widetilde R^{(k_0+T)}
    \right]
    \leq\frac{5}{2d^2}+\frac{3}{2d^2}
    =\frac{4}{d^2}\enspace.
\]
By \cref{lem:coupling_lemma}, this proves
\cref{eq:parallel-one-block}.

Applying \cref{fact:tv-block-decay} with
\(\varepsilon=4/d^2<1/2\), we obtain
$\delta_{\scriptscriptstyle\mathsf{PKMP}}(k) \leq\left(\frac{8}{d^2}\right)^{\lfloor k/\ell\rfloor}$.
Using \(\lfloor k/\ell\rfloor\geq k/\ell-1\),
\(\ell\leq6040\log d\), and
\(\log(d^2/8)\geq\frac13\log d\) for \(d\geq4\), we obtain
\[
    \delta_{\scriptscriptstyle\mathsf{PKMP}}(k)
    \leq d^2\cdot\e^{-k\log(d^2/8)/\ell}
    \leq d^2\cdot\e^{-k/18120}\enspace.
\]
Taking \(c=1/18120\) proves the theorem.
\end{proof}

\endgroup
\end{document}